\pdfoutput=1
\documentclass{lmcs}

\usepackage{lastpage}
\lmcsdoi{17}{3}{11}
\lmcsheading{}{\pageref{LastPage}}{}{}%
{Dec.~01,~2020}{Jul.~28,~2021}{}

\usepackage[utf8]{inputenc}


\usepackage{breakcites}
\usepackage{amssymb}
\usepackage{locallabel}
\usepackage{pftools}
\usepackage{wasysym}
\usepackage{stmaryrd}
\usepackage{fontawesome}
\usepackage{mleftright}
\usepackage{xparse}
\usepackage{ifthen}
\usepackage{scalerel}
\usepackage{epigraph}
\usepackage{datetime}
\usepackage{mathpartir}
\usepackage{array}
\usepackage{mathtools}
\usepackage[section]{placeins}

\definecolor{Matterhorn}{RGB}{77,77,77}
\definecolor{RegalBlue}{RGB}{3,69,117}
\definecolor{RedDevil}{RGB}{134,1,17}

\usepackage{macros}
\usepackage{categories}
\usepackage{jon-tikz}

\title{Multimodal Dependent Type Theory}

\author[D.~Gratzer]{Daniel Gratzer\rsuper{a}}  

\author[G.A.~Kavvos]{G.A. Kavvos\rsuper{b}}  

\author[A.~Nuyts]{Andreas Nuyts\rsuper{c}}  

\author[L.~Birkedal]{Lars Birkedal\rsuper{a}}  

\address{\lsuper{a}Aarhus University}  
\email{gratzer@cs.au.dk, birkedal@cs.au.dk}  

\address{\lsuper{b}University of Bristol}  
\email{alex.kavvos@bristol.ac.uk}  

\address{\lsuper{c}Vrije Universiteit Brussel}  
\email{andreas.nuyts@vub.be}  

\begin{document}
\maketitle

\begin{abstract}
  We introduce \MTT{}, a dependent type theory which supports multiple modalities.  \MTT{} is
  parametrized by a mode theory which specifies a collection of modes, modalities, and
  transformations between them.
  We show that different choices of mode theory allow us to use the same type theory to compute and
  reason in many modal situations, including guarded recursion, axiomatic cohesion, and parametric
  quantification. We reproduce examples from prior work in guarded recursion and axiomatic cohesion,
  thereby demonstrating that \MTT{} constitutes a simple and usable syntax whose instantiations
  intuitively correspond to previous handcrafted modal type theories. In some cases, instantiating
  \MTT{} to a particular situation unearths a previously unknown type theory that improves upon
  prior systems.
  Finally, we investigate the metatheory of \MTT{}. We prove the consistency of \MTT{} and establish
  canonicity through an extension of recent type-theoretic gluing techniques. These results hold
  irrespective of the choice of mode theory, and thus apply to a wide variety of modal situations.

\end{abstract}
\maketitle

\section{Introduction}
\label{sec:intro}

In order to increase the expressivity of Martin-L\"of Type Theory (\MLTT{}) we often wish to extend
it with unary type operators that we call \emph{modalities} or \emph{modal operators}. Some modal
operators arise as shorthands for internally definable structure~\cite{rijke:2020}, while others
are used as a device for internalising non-definable structure from particular models. In the latter
case, we are sometimes even able to prove that a modality cannot be internally expressed---at least
not without extensive changes to the judgmental structure of type theory: see \eg~the `no-go'
theorems by \cite[\S 4.1]{shulman:2018} and \cite{licata:2018}. This paper is concerned with the
development of a systematic approach to the judgmental formulation of type theories with multiple
interacting modalities.

The addition of a modality to a dependent type theory is a non-trivial exercise. Modal operators
often interact with the context of a type or term in a complicated way, and na\"{i}ve approaches
lead to undesirable interplay with other type formers and substitution. However, the consequent gain
in expressivity is substantial, and so it is well worth the effort. For example, modalities have
been used to express guarded recursive
definitions~\cite{birkedal:2012,bizjak:2016,bahr:2017,guatto:2018}, parametric
quantification~\cite{nuyts:2017,nuyts:2018}, proof
irrelevance~\cite{pfenning:2001,abel:2012,nuyts:2018}, and to define global operations which cannot
be localized to an arbitrary context~\cite{licata:2018}. There has also been concerted
effort towards the development of a dependent type theory corresponding to Lawvere's \emph{axiomatic
  cohesion}~\cite{lawvere:2007}, which has many interesting applications \cite{schreiber:2013,
  schreiber:2014, shulman:2018, gross:2017, kavvos:2019}.

Despite this recent flurry of developments, a unifying account of modal dependent type theory has
yet to emerge. Faced with a new modal situation, a type theorist must handcraft a brand new system,
and then prove the usual battery of metatheorems. This introduces formidable difficulties on two
levels. First, an increasing number of these applications are \emph{multimodal}: they involve
multiple interacting modalities, which significantly complicates the design of the appropriate
judgmental structure. Second, the technical development of each such system is entirely separate, so
that one cannot share the burden of proof even between closely related systems.  To take a recent
example, there is no easy way to transfer the work done in the 80-page-long normalization proof for
$\MLTTLock$~\cite{gratzer:2019} to a normalization proof for the modal dependent type theory of
\cite{clouston:dra:2018}, even though these systems are only marginally different. Put simply, if
one wished to prove that type-checking is decidable for the latter, then one would have to start
afresh.

We intend to avoid such duplication in the future. Rather than designing a new dependent type theory
for some preordained set of modalities, we will introduce a system that is \emph{parametrized by a
mode theory}, \ie{}~an algebraic specification of a modal situation.  This system, which we
call \MTT{}, solves both problems at once. First, by instantiating it with different mode theories
we will show that it can capture a wide class of situations. Some of these, \eg{}~the one for
guarded recursion, lead to a previously unknown system that improves upon earlier work. Second, the
predictable behavior of our rules allows us to prove metatheoretic results about large classes of
instantiations of our system. For example, our canonicity theorem applies irrespective of the
chosen mode theory. As a result, we only need to prove such theorems \emph{once}. Returning to the
previous examples, careful choices of mode theory yield two systems that closely resemble the calculi of
\cite{clouston:dra:2018} and $\MLTTLock$~\cite{gratzer:2019} respectively, so that our proof of
canonicity applies to both.

In fact, we take things one step further: \MTT{} is not just multimodal, but also \emph{multimode}.
That is, each judgment of \MTT{} can be construed as existing in a particular \emph{mode}. All modes
have some things in common---\eg{} there will be dependent sums in each---but some might
possess distinguishing features. From a semantic point of view, different modes correspond to
different context categories. In this light, modalities intuitively correspond to \emph{functors}
between those categories: in fact, they will be structures slightly weaker than \emph{dependent
right adjoints} (DRAs)~\cite{clouston:dra:2018}.

\paragraph{Mode theories}

At a high level, \MTT{} can be thought of as a machine that converts a concrete description of modes
and modalities into a type theory. This description, which is often called a \emph{mode theory}, is
given in the form of a \emph{small strict 2-category}~\cite{reed:2009,licata:2016,licata:2017}. A
mode theory gives rise to the following correspondence:
\begin{align*}
  \text{object} &\sim \text{mode} \\
  \text{morphism} &\sim \text{modality} \\
  \text{2-cell} &\sim \text{natural map between modalities}
\end{align*}
The equations between morphisms and between 2-cells in a mode theory can be used to precisely specify
the interactions we want between different modalities. We will illustrate this point with an example.


\paragraph{Instantiating \MTT{}}


Suppose we have a mode theory $\Mode$ with a single object $m$, a single generating morphism $\mu :
m \to m$, and no non-trivial 2-cells. Equipping \MTT{} with $\Mode$ produces a type theory with a
single modal type constructor, $\Modify$. This is the simplest non-trivial setting, and we can prove
very little about it without additional 2-cells.

If we add a 2-cell $\epsilon : \mu \To 1$ to $\Mode$, we can define a function
\[
  \mathsf{extract}_A : \Modify{A} \to A
\]
inside the type theory. If we also add a 2-cell $\delta : \mu \To \mu \circ \mu$ then we can also
define
\[
  \mathsf{duplicate}_A : \Modify{A} \to \Modify{\Modify{A}}
\]
Furthermore, we can control the precise interaction between $\mathsf{duplicate}_A$ and $\mathsf{extract}_A$ by
adding more equations that relate $\epsilon$ and $\delta$. For example, we may ask that $\Mode$ be
the \emph{walking comonad} \cite{schanuel:1986} which leads to a type theory with a dependent
\textsf{S4}-like modality~\cite{pfenning:2001,de-paiva:2015,shulman:2018}. We
can be even more specific, \eg~by asking that $(\mu, \epsilon, \delta)$ be \emph{idempotent}.

Thus, a morphism $\mu : n \to m$ introduces a modality $\Modify$, and a 2-cell
$\alpha : \mu \To \nu$ of $\Mode$ allows for the definition of a function of type
$\Modify{A} \to \Modify[\nu]{A}$ at mode $m$.

\paragraph{Relation to other modal type theories}

Most work on mo\-dal type theories still defies classification. However, we can informatively position
\MTT{} with respect to two qualitative criteria, \viz\ usability and generality.

Much of the prior work on modal type theory has focused on bolting a specific modality onto a type theory. The benefit of this
approach is that the syntax can be designed to be as convenient as possible for the application at
hand. For example, spatial/cohesive type theory~\cite{shulman:2018} features two modalities,
$\flat$ and $\sharp$, and is presented in a dual-context style. This judgmental structure, however,
is applicable only because of the particular properties of $\flat$ and $\sharp$. Nevertheless, the numerous
pen-and-paper proofs in \emph{op. cit.} demonstrate that the resulting system is easy
to use.

At the other end of the spectrum, the framework of Licata-Shulman-Riley (LSR)~\cite{licata:2017}
comprises an extremely general toolkit for simply-typed, substructural modal type theory. Its
dependent generalization, which is currently under development, is able to handle a very large class
of modalities. However, this generality comes at a price: its syntax is complex and unwieldy,
even in the simply-typed case.

\MTT{} attempts to strike a delicate balance between those two extremes. By avoiding substructural
settings and some kinds of modalities we obtain a noticeably simpler apparatus. Unlike LSR, we need
not annotate our term formers with delayed substitutions, and our approach extends to dependent
types in a straightforward manner. Most of the pleasant type-theoretic behaviour of \MTT{} is
achieved by ensuring that none of its rules `trim' the context, which would necessitate either
delayed substitutions~\cite{bizjak:2016,licata:2017} or delicate proofs of the admissibility of
substitution~\cite{bahr:2017,clouston:dra:2018,gratzer:2019}. We also show that \MTT{} can be
employed to reason about many models of interest, and that it is simple enough to be used in
pen-and-paper calculations.

\paragraph{Contributions}

In summary, we make the following contributions:
\begin{itemize}
\item We introduce \MTT{}, a general type theory for multiple modes and multiple interacting
  modalities.
\item We present a semantics, which constitute a category of models deriving from the generalized algebraic theory that underlies \MTT{}.
\item Using the semantics, we prove that---subject to a technical restriction---\MTT{} satisfies
  \emph{canonicity}, an important metatheoretic property. This is achieved through a modern
  \emph{gluing} argument~\cite{shulman:2015,altenkirch:2016,coquand:2018,kaposi:gluing:2019}.
\item Finally, we instantiate \MTT{} with various mode theories, and show its use in reasoning about two specific modal situations, viz.\ guarded recursion~\cite{bizjak:2016}, and internal adjunctions~\cite{shulman:2018,licata:2018}.
\end{itemize}


\section{The Syntax of \MTT{}}
\label{sec:towards-mtt}

As mentioned in the introduction, the syntax of \MTT{} is parameterized by a small 2-category called
a \emph{mode theory}. We will later show how to instatiate \MTT{} with a mode theory in order to
reason about particular scenarios, but for now we will work over an arbitrary mode theory. We thus
fix a mode theory $\Mode$, and use $m, n, o$ to stand for modes (the objects of $\Mode$), $\mu, \nu,
\tau$ for modalities (the morphisms), and $\alpha, \beta, \gamma$ for 2-cells.

In broad terms, \MTT{} consists of a collection of type theories, one for each mode $m \in \Mode$.
These type theories will eventually appear in one another, but only as spectres under a modality. We
thus begin by describing the individual type theories at each mode, and then discuss how modalities
are used to relate them.

\subsection{The Type Theory at Each Mode}
\label{sec:towards-mtt:no-modalities}

Each mode of \MTT{} is inhabited by a standard Martin-L{\"o}f Type Theory (\MLTT{}), and accordingly
includes the usual judgments. For example, we have the judgment $\IsCx{\Gamma}$ which states that
$\Gamma$ is a well-formed context \emph{in that particular mode $m$}. There are likewise judgments
for types, terms, and substitutions at each mode.

\begin{figure*}
  \begin{mathpar}
    \JdgFrame{\IsTy{A}}
    \\
    \inferrule{
      \IsCx{\Gamma}
    }{
      \IsTy{\Uni}[1]
    }
    \and
    \inferrule{
      \IsCx{\Gamma}
    }{
      \IsTy{\Bool}
    }
    \and
    \inferrule{
      \IsCx{\Gamma}\\
      \IsTy{A}\\
      \ell \le \ell'
    }{
      \IsTy{\TyLift{A}}[\ell']
    }
    \and
    \inferrule{
      \IsCx{\Gamma}\\
      \IsTy{A}\\
      \IsTm{M,N}{\TyLift{A}}
    }{
      \IsTy{\Id{A}{M}{N}}
    }
    \and
    \inferrule{
      \IsCx{\Gamma}\\
      \IsTy{A}\\
      \IsTy[\Gamma, x : \TyLift{A}]{B}
    }{
      \IsTy{(x : A) \to B} \\ \IsTy{(x : A) \times B}
    }
  \end{mathpar}
  \caption{Selected mode-local rules.}
  \label{fig:towards-mtt:local-rules}
\end{figure*}

In lieu of an exhaustive list of rules, which we will present in Section~\ref{sec:algebraic-mtt}, we
illustrate this point by only showing the important ones in
Figure~\ref{fig:towards-mtt:local-rules}. In brief, each mode comprises an ordinary intensional type
theory with dependent sums, dependent products, intensional identity types, booleans, and one
universe. Both sums and products satisfy an $\eta$-rule.

\paragraph{Universes \`a la Coquand}
  \label{paragraph:universes}

There are several ways to introduce universes in type theory \cite[\S 2.1.6]{hofmann:1997}
\cite{palmgren:1998, luo:2012}. We use the approach of \cite{coquand:2013}, which is close to
Tarski-style universes. However, instead of inductively defining \emph{codes} that represent
particular types, Coquand-style universes come with an \emph{explicit isomorphism} between types and
terms of the universe $\Uni$. However, we must remember to exercise caution: if this isomorphism
were to cover all types then \emph{Girard's paradox} \cite{coquand:1986} would apply, so we must
restrict it to \emph{small types}. This, in turn, forces us to stratify our types into small and large.

The judgment $\IsTy{A}[0]$ states that $A$ is a small type, and $\IsTy{A}[1]$ that it
is large. The universe itself must be a large type, but otherwise both levels are closed under all
other connectives. Finally, we introduce an operator that \emph{lifts} a small type to a large one:
\[
  \inferrule{
    \ell \le \ell'\\
    \IsTy{A}[\ell]
  }{
    \IsTy{\TyLift{A}}[\ell']
  }
\]
The lifting operation commutes definitionally with all the connectives,
\eg~$\TyLift{(A \to B)} = \TyLift{A} \to \TyLift{B}$. We will use large types for the most part:
only they will be allowed in contexts, and the judgment $\IsTm{M}{A}$ will presuppose that $A$ is
large. As we will not have terms at small types, we will not need the term lifting operations
used by \cite{coquand:2013} and \cite{sterling:2019}.

Following this stratification, we may introduce operations that exhibit the isomorphism:
\begin{mathparpagebreakable}
  \inferrule{
    \IsTm{M}{\Uni}
  }{
    \IsTy{\Dec{M}}[0]
  }
  \and
  \inferrule{
    \IsTy{A}[0]
  }{
    \IsTm{\Enc{A}}{\Uni}
  }
\end{mathparpagebreakable}
along with the equations $\Enc{\Dec{M}} = M$ and $\Dec{\Enc{A}} = A$.

The advantages of universes \`a la Coquand are now evident: rather than having to introduce
Tarski-style codes, we now find that they are \emph{definable}. For example, assuming $M : \Uni$ and
$x : \Dec{M} \vdash N : \Uni$, we let
\[
  (x : M) \mathrel{\widehat{\to}} N \defeq \Enc{(x : \Dec{M}) \to \Dec{N}} : \Uni
\]
We can then calculate that
\[
  \Dec{(x : M) \mathrel{\widehat{\to}} N} = \Dec{\Enc{(x : \Dec{M}) \to \Dec{N}}}
                                          = (x : \Dec{M}) \to \Dec{N}
\]
We will often suppress $\TyLift{-}$ as well as the explicit isomorphism.

\subsection{Introducing a Modality}
\label{sec:towards-mtt:single-modality}

Having sketched the basic type theory inhabiting each mode, we now turn to the interaction between them. This is the domain of the modalities.

Suppose $\Mode$ contains a modality $\mu : n \to m$. We would like to think of $\mu$ as a `map' from
mode $n$ to mode $m$. Then, for each $\IsTy[]{A}[]<n>$ we would like a type
$\IsTy[]{\Modify{A}}[]<m>$. On the level of terms we would similarly like for each
$\IsTm[]{M}{A}<n>$ an induced term $\IsTm[]{\MkBox{M}}{\Modify{A}}<m>$.

These constructs would be entirely satisfactory, were it not for the presence of \emph{open terms}.
To illustrate the problem, suppose we have a type $\IsTy{A}[]<n>$. We would hope that the
corresponding modal type lives in the same context, \ie~that $\IsTy{\Modify{A}}[]<m>$. However,
this is not possible, as $\Gamma$ is only a context at mode $n$, and cannot be
carried over verbatim to mode $m$. Hence, the only pragmatic option is to introduce an operation
that allows a context to be mapped to another mode.

\paragraph{Forming a modal type}

There are several different proposed solutions to this problem in the literature
\cite{pfenning-davies:2001, clouston:fitch:2018}. We will use a \emph{Fitch-style}
discipline~\cite{bahr:2017, clouston:dra:2018, gratzer:2019}: we will require that a modality $\mu$
induce an operation on contexts in the \emph{opposite} direction. We will denote this operation by a \emph{lock}:
\[
  \inferrule[cx/lock]{
    \IsCx{\Gamma}
  }{
    \IsCx{\LockCxV{\Gamma}}<n>
  }
\]
Intuitively, $\Lock_\mu$ will behave somewhat like a left adjoint to $\Modify{-}$. However,
$\Modify{-}$ acts on types while $\LockCxV{-}$ acts on contexts, so this cannot be an ordinary
adjunction. Instead, $\Modify{-}$ will be what \cite{clouston:dra:2018} call a \emph{dependent
right adjoint} (DRA). A DRA essentially consists of a type former $\mathbf{R}$ and a context
operation $\mathbf{L}$ such that
\begin{equation*}
  \{N \mid \mathbf{L}(\Gamma) \vdash N : A \} \cong \{ M \mid \Gamma \vdash M : \mathbf{R}(A) \}
  \TagEq[$\dagger$]\label{eq:towards-mtt:dra}
\end{equation*}
See \cite{clouston:dra:2018} for a formal definition.

Just as with DRAs, the \MTT{} formation and introduction rules for modal types effectively
\emph{transpose} types and terms across this adjunction:
\begin{mathpar}
  \inferrule[tp/modal]{
    \mu : m \to n \\
    \IsTy[\LockCxV{\Gamma}]{A}<n>
  }{
    \IsTy{\Modify{A}}
  }
  \and
  \inferrule[tm/modal-intro]{
    \mu : m \to n \\
    \IsTm[\LockCxV{\Gamma}]{M}{A}<n>
  }{
    \IsTm{\MkBox{M}}{\Modify{A}}
  }
\end{mathpar}

It remains to show how to eliminate modal types. Previous work on Fitch-style calculi
\cite{clouston:dra:2018, gratzer:2019} has employed elimination rules which essentially invert the
introduction rule \ruleref{tm/modal-intro}. Such rules \emph{remove} one or more locks from the
context during type-checking, and sometimes even trim a part of it. For example, a rule of this sort
would be
\[
  \inferrule{
    \Lock_\mu \not\in \Gamma'\\
    \IsTm{M}{\Modify{A}}
  }{
    \IsTm[\LockCxV{\Gamma}, \Gamma']{\mathsf{open}(M)}{A}<n>
  }
\]
This kind of rule tends to be unruly, and delicate work is required to prove even basic results
about it. For example, see the technical report \cite{gratzer:tech-report:2019} for a particularly
laborious proof of the admissibility of substitution. The results in \emph{op. cit.} could not
possibly reuse any of the work of \cite{clouston:dra:2018}, as a small change in the syntax leads
to many subtle differences in the metatheory. Consequently, it seems unlikely that one could adapt
this approach to a modality-agnostic setting like ours.

We will use a different technique, which is reminiscent of dual-context
calculi \cite{kavvos:2020}. First, we will let the variable rule control the use of modal
variables. Then, we will take a `modal cut' rule, which will allow
the substitution of modal terms for modal variables, to be our modal elimination rule.

\paragraph{Accessing a modal variable}

The behavior of modal types can often be clarified by asking a simple question: when can we use a
variable $x : \Modify{A}$ of modal type to construct a term of type $A$? In previous Fitch-style
calculi we would use the modal elimination rule to reduce the goal to $\Modify{A}$, and
then---\emph{had the modal elimination rule not eliminated $x$ from the context}---we would simply
use the variable. We may thus write down a term of type $A$ using a variable $x : \Modify{A}$ only
when our context is structured in a way that does not obstruct the use of $x$, and the final arbiter
of that is the modal elimination rule.

\MTT{} turns this idea on its head: rather than handing control over to the modal elimination rule,
we delegate this decision to the variable rule itself. In order to ascertain whether we can use a
variable in our calculus, the variable rule examines \emph{the locks to the right of the variable}.
The rule of thumb is this: we should always be able to access $\Modify{A}$ behind $\Lock_\mu$.
Carrying the illustrative analogy of an adjunction $\LockCxV{-} \Adjoint \Modify{-}$ further, we see
that the simplest judgment that fits this, namely $\IsTm[\LockCxV{\Gamma, x :
\Modify{A}}]{x}{A}<n>$, corresponds to the \emph{counit} of the adjunction.

To correctly formulate the variable rule, we will require one more idea: following modal type
theories based on \emph{left
  division}~\cite{pfenning:2001,abel:phd,abel:2008,nuyts:2017,nuyts:2018}, every variable in the
context will be annotated with a modality, $\DeclVar{x}{A}$. Intuitively a variable $\DeclVar{x}{A}$
is the same as a variable $x : \Modify{A}$, but the annotations are part of the structure of
a context while $\Modify{A}$ is a type. This small circumlocution will ensure that the variable rule
respects substitution.

The most general form of the variable rule will be able to handle the interaction of modalities, so
we present it in stages. A first counit-like approximation is then
\begin{mathpar}
  \inferH{tm/var/counit}{
    \Lock \not\in \Gamma_1\\
    \IsTy[\LockCxV{\Gamma_0}]{A}[1]<n>
  }{
    \IsTm[\LockCxV{\ECxV{\Gamma_0}{x}{A}}, \Gamma_1]{x}{A}<n>
  }
\end{mathpar}
The first premise requires that no further locks occur in $\Gamma_1$, so that the conclusion remains in the same mode $n$. The second premise is just enough to derive $\IsTy[\Gamma_0]{\Modify{A}}[1]$.

\paragraph{Context extension}

The switch to modality-annotated declarations $\DeclVar{x}{A}$ also requires us to revise the
context extension rule. The revised version, \ruleref{cx/extend}, appears in
Figure~\ref{fig:towards-mtt:modal-rules} and closely follows the formation rule for $\Modify{-}$: if
$\IsTy[\LockCxV{\Gamma}]{A}[1]<n>$ is a type in the locked context $\Gamma$, then we may extend the
context $\Gamma$ to include a declaration $\DeclVar{x}{A}$, so that $x$ stands for a term of type
$A$ \emph{under the modality $\mu$}.

\paragraph{The elimination rule}

The difference between a modal type $\Modify{A}$ and an annotated declaration $\DeclVar{x}{A}<\mu>$
in the context is navigated by the modal elimination rule. In brief, its role is to enable the
substitution of a term of the former type for a variable with the latter declaration. The full rule
is complex, so we first discuss the case of a single modality $\mu : n \to m$. The correspoding rule is
\begin{mathpar}
  \inferH{tm/modal-elim/single-modality}{
    \IsTm{M_0}{\Modify{A}}\\
    \IsTy[\ECxV{\Gamma}{x}{\Modify{A}}<1>]{B}[1]\\
    \IsTm[\ECxV{\Gamma}{y}{A}]{M_1}{\Sb{B}{\MkBox{y}/x}}
  }{
    \IsTm{\Open{M_0}[y]{M_1}<\mu>[1]}{\Sb{B}{M_0/x}}
  }
\end{mathpar}
Forgetting dependence for a moment, we see that this rule is close to the dual-context style
\cite{pfenning-davies:2001, kavvos:2020}: if we think of annotations as separating the context into
multiple zones, then $\DeclVar{y}{A}<\mu>$ clearly belongs to the `modal' part.

In the dependent case we also need a motive $\IsTy[\ECxV{\Gamma}{x}{\Modify{A}}<1>]{B}[1]$, which
depends on a variable of modal type, but under the identity modality $1$. This premise is then
fulfilled by $M_0$ in the conclusion. In a sense, this rule permits a form of \emph{modal
induction}: every variable $\DeclVar{x}{\Modify{A}}<1>$ can be assumed to be of the form $\MkBox{y}$
for some $\DeclVar{y}{A}<\mu>$. This kind of rule has appeared before in the spatial and cohesive type theories of \cite{shulman:2018}.

In the type theory of \cite{clouston:dra:2018} modalities are taken to be dependent right adjoints,
with terms witnessing Equation~\ref{eq:towards-mtt:dra}. This isomorphism can encode
\ruleref{tm/modal-elim/single-modality}, but that rule alone cannot encode
Equation~\ref{eq:towards-mtt:dra}. As a result, modalities in \MTT{} are weaker than DRAs.

\subsection{Multiple Modalities}
\label{sec:towards-mtt:multiple-modalities}

\begin{figure*}
  \begin{mathpar}
    \JdgFrame{\IsCx{\Gamma}}\\
    \inferH{cx/lock}{
      \mu : n \to m\\
      \IsCx{\Gamma}\\
    }{
      \IsCx{\LockCxV{\Gamma}<\mu>}<n>
    }
    \and
    \inferH{cx/extend}{
      \mu : n \to m\\
      \IsCx{\Gamma}\\
      \IsTy[\LockCxV{\Gamma}]{A}[1]<n>
    }{
      \IsCx{\ECxV{\Gamma}{x}{A}}
    }
    \\
    \inferH{cx/id}{
      \IsCx{\Gamma}
    }{
      \EqCx{\Gamma}{\LockCxV{\Gamma}<1>}
    }
    \and
    \inferH{cx/compose}{
      \nu : o \to n\\
      \mu : n \to m\\
      \IsCx{\Gamma}
    }{
      \EqCx{\LockCxV{\LockCxV{\Gamma}}<\nu>}{\LockCxV{\Gamma}<\mu \circ \nu>}<o>
    }
    \\
    \JdgFrame{\IsTy{A} \qquad \IsTm{M}{A}}
    \\
    \inferH{tp/modal}{
      \mu : n \to m \\
      \IsTy[\LockCxV{\Gamma}]{A}<n>
    }{
      \IsTy{\Modify{A}}
    }
    \and
    \inferH{tm/var}{
      \nu : m \to n\\
      \alpha : \nu \To \Locks{\Gamma_1}
    }{
      \IsTm[\ECxV{\Gamma_0}{x}{A}<\nu>, \Gamma_1]{x^\alpha}{A^\alpha}
    }
    \and
    \inferH{tm/modal-intro}{
      \mu : n \to m \\
      \IsTm[\LockCxV{\Gamma}]{M}{A}<n>
    }{
      \IsTm{\MkBox{M}}{\Modify{A}}
    }
    \and
    \inferH{tm/modal-elim}{
      \mu : n \to m \\
      \nu : m \to o \\
      \IsTy[\ECxV{\Gamma}{x}{\Modify[\mu]{A}}<\nu>]{B}[1]<o> \\\\
      \IsTm[\LockCxV{\Gamma}<\nu>]{M_0}{\Modify[\mu]{A}}<m> \\
      \IsTm[\ECxV{\Gamma}{x}{A}<\nu \circ \mu>]{M_1}{\Sb{B}{\MkBox[\mu]{x}/x}}<o>
    }{
      \IsTm{\Open{M_0}[x]{M_1}<\mu>[\nu]}{\Sb{B}{M_0/x}}<o>
    }
    \and
    \inferH{tm/modal-beta}{
      \nu : m \to o \\
      \mu : n \to m \\
      \IsTy[\ECxV{\Gamma}{x}{\Modify[\mu]{A}}<\nu>]{B}[1]<o> \\\\
      \IsTm[\LockCxV{\Gamma}<\nu \circ \mu>]{M_0}{A}<n> \\
      \IsTm[\ECxV{\Gamma}{x}{A}<\nu \circ \mu>]{M_1}{\Sb{B}{\MkBox[\mu]{x}/x}}<o>
    }{
      \EqTm{
        \Open{\MkBox[\mu]{M_0}}[x]{M_1}<\mu>[\nu]
      }{
        \Sb{M_1}{M_0/x}
      }{
        \Sb{B}{\MkBox[\mu]{M_0}/x}
      }<o>
    }
    \\
    \JdgFrame{\Locks{\Gamma}}
    \\
    \Locks{\Emp} = 1 \and
    \Locks{\ECxV{\Gamma}{x}{A}} = \Locks{\Gamma} \and
    \Locks{\LockCxV{\Gamma}} = \Locks{\Gamma} \circ \mu
  \end{mathpar}
  \caption{Selected modal rules.}
  \label{fig:towards-mtt:modal-rules}
\end{figure*}

So far we have only considered a single modality. In this section we discuss the few additional
tweaks that are needed to support multiple interacting modalities. The final version of the modal
rules is given in Figure~\ref{fig:towards-mtt:modal-rules}.

\paragraph{Multimodal locks}

Up to this point the operation $\LockCxV{-}$ on contexts has referred to a single modality $\mu : n
\to m$. The rule \ruleref{cx/lock} generalizes it to work with any modality. The only question then
is how the resulting operations should interact. This is where the mode theory comes in: locks
should be \emph{functorial}, so that $\nu : o \to n$, $\mu : n \to m$, and $\IsCx{\Gamma}$ imply
$\EqCx{\LockCxV{\LockCxV{\Gamma}}<\nu>}{\LockCxV{\Gamma}<\mu \circ \nu>}<o>$. We additionally ask
that the identity modality $1 : m \to m$ at each mode has a trivial, invisible action on contexts,
\ie that $\LockCxV{\Gamma}<1> = \Gamma$.

These two actions, which are encoded by \ruleref{cx/compose} and \ruleref{cx/id}, ensure that
\emph{$\Lock$ is a contravariant functor on $\Mode$}, mapping each mode $m$ to the category of
contexts $\IsCx{\Gamma}$. The contravariance originates from the fact that $\Mode$ is a
specification of the behavior of the modalities $\Modify{-}$, so that their left-adjoint-like
counterparts $\LockCxV{-}$ act with the opposite variance.

\paragraph{The full variable rule}

We have seen that $\Lock$ induces a functor from $\Mode$ to categories of contexts, but we have not
yet used the 2-cells of $\Mode$. In short, a 2-cell $\alpha : \mu \To \nu$ contravariantly induces a
substitution from $\LockCxV{\Gamma}<\nu>$ to $\LockCxV{\Gamma}<\mu>$. We will discuss this further
in Section~\ref{sec:algebraic-mtt}, but for now we only mention that this arrangement gives rise to
an \emph{admissible operation on types}: for each 2-cell we obtain an operation $(-)^\alpha$ such
that $\IsTy[\LockCxV{\Gamma}]{A}[]$ implies $\IsTy[\LockCxV{\Gamma}<\nu>]{A^\alpha}[]$.

In order to prove the admissibility of this operation we need a more expressive variable rule
that builds in the action of 2-cells. The first iteration (\ruleref{tm/var/counit}) required that
the lock and the variable annotation were an exact match. We relax this requirement by allowing
for a mediating 2-cell:
\begin{mathpar}
  \inferH{tm/var/combined}{
    \mu,\nu : n \to m\\
    \alpha : \mu \To \nu
  }{
    \IsTm[\LockCxV{\ECxV{\Gamma}{x}{A}}<\nu>]{x^\alpha}{A^\alpha}<n>
  }
\end{mathpar}

The superscript in $x^\alpha$ is now part of the syntax: each variable must be annotated with the
2-cell that `unlocks' it and enables its occurrence, though we will still write $x$ to mean
$x^{1_\mu}$.  The final form of the variable rule, which appears as \ruleref{tm/var} in
Figure~\ref{fig:towards-mtt:modal-rules}, is only a slight generalization of this last rule: it allows the variable to
occur at positions other than the very front of the context. In fact, \ruleref{tm/var} can be
reduced to \ruleref{tm/var/combined} by using weakening to remove variables to the right of $x$, and
then invoking functoriality to fuse all the locks to the right of $x$ into a single one with
modality $\Locks{\Gamma_1}$.

\paragraph{The full elimination rule}

Recall that the elimination rule for a single modality allowed us to plug in a term of type
$\Modify{A}$ for an assumption $\DeclVar{x}{A}<\mu>$. Some additional generality is needed to cover
the case where the motive $\IsTy[\DeclVar{x}{\Modify{A}}<\nu>]{B}[]$ depends on $x$ under a modality
$\nu \neq 1$. This is where the composition of modalities in $\Mode$ comes in handy: our new rule
will use it to absorb $\nu$ by replacing the assumption $\DeclVar{x}{\Modify{A}}<\nu>$ with
$\DeclVar{x}{A}<\nu \circ \mu>$.

The new rule, \ruleref{tm/modal-elim}, is given in Figure~\ref{fig:towards-mtt:modal-rules}. The simpler
rule may be recovered by setting $\nu \defeq 1$. In this simpler case, we will suppress the
subscript $1$ on $\mathsf{let}$, just as in \ruleref{tm/modal-elim/single-modality}.  However,
many natural examples require eliminations where $\nu \neq 1$. For instance, in
Section~\ref{sec:programming-in-mtt} we show that $\Modify[\nu \circ \mu]{A} \Equiv
\Modify[\nu]{\Modify{A}}$. The function from the right-hand side to the left crucially depends on
the ability to pattern-match on a variable $\DeclVar{x}{\Modify{A}}<\nu>$, which requires the
stronger \ruleref{tm/modal-elim}.



\paragraph{Definitional equality in \MTT{}}

A perennial problem in type theory is that of deciding where the boundary between those equalities
that are \emph{provable} in the system (e.g. using various forms of induction), and those that are
\emph{definitional}, i.e. hold by fiat. While we have simply followed standard practices in the $\MLTT$
connectives at each mode, the situation is somewhat more complicated regarding modal types. On the
one hand, we have the expected $\beta$-rule \ruleref{tm/modal-beta}: see
Figure~\ref{fig:towards-mtt:modal-rules}. On the other hand, we do not include any definitional
$\eta$-rules: as the eliminator is a \emph{positive} pattern-matching construct, the proper
$\eta$-rule would need \emph{commuting conversions}, which would enormously complicate the
metatheory.

\paragraph{Notational conventions}

In the rest of the paper we shall make use of the following notational conventions.

\begin{nota}
  When opening a modal term under the modality $1$ we will suppress the $1$ in the $\textsf{let}_1$
  part of the term, and write $\Open{M}{N}[1]$ instead.
\end{nota}


\begin{nota}
  As remarked before, Coquand-style universes do not require the introduction of codes that
  represent various types in the universe, for they are definable. Nevertheless, in examples we will
  often suppress both $\Dec{-}$ and $\Enc{-}$, and in some straightfoward cases even elide the
  coercion $\TyLift{-}$. This not only makes our terms more perspicuous, but can also be formally
  justified by an \emph{elaboration procedure} which inserts the missing isomorphisms and coercions
  when needed.
\end{nota}


\section{Programming with Modalities}
\label{sec:programming-in-mtt}

In this section we show how \MTT{} can be used to program and reason with modalities. We we identify
a handful of basic modal combinators which demonstrate the behaviour of our modal types. Then, in
Section~\ref{sec:programming-in-mtt:comonads} we use them to present a type theory featuring an
idempotent comonad with almost no additional effort.

\subsection{Modal Combinators}
\label{sec:programming-in-mtt:combinators}

We first show how each 2-cell $\alpha : \mu \To \nu$ with $\mu, \nu : n \to m$ induces a natural
transformation $\Modify{-} \to \Modify[\nu]{-}$. We call the components of this natural
transformation \emph{coercions}. Given $\IsTy[\LockCxV{\Gamma}]{A}[1]$, define
\[
  \arraycolsep=1.4pt
  \begin{array}{lcl}
    \Coe{-}{\mu}{\nu}{\alpha} & : & \Modify{A} \to \Modify[\nu]{A^\alpha}\\
    \Coe{x}{\mu}{\nu}{\alpha} & \defeq & \Open{x}[z]{\MkBox[\nu]{z^\alpha}}<\mu>[1]
  \end{array}
\]
The heart of this combinator is a use of the rule \ruleref{tm/var}. This operation completes the
correspondence sketched in Section~\ref{sec:intro}: objects of $\Mode$ correspond to modes, morphisms to
modalities, and 2-cells to coercions.

Additionally, the assignment $\mu \mapsto \Modify{-}$ is \emph{functorial}. Unlike the action of
locks, this functoriality is not definitional, but only a type-theoretic \emph{equivalence}
\cite[\S 4]{hottbook}. Fixing $\nu : o \to n$, $\mu : n \to m$, and $\IsTy[\LockCxV{\Gamma}<\mu
\circ \nu>]{A}[1]$, we let
\[
  \arraycolsep=1.4pt
  \begin{array}{lcl}
    \multicolumn{3}{l}{\MComp{}{\mu}{\nu} : \Modify{\Modify[\nu]{A}} \to \Modify[\mu \circ \nu]{A}}\\
    \MComp{x}{\mu}{\nu} & \defeq & \Open{x}[x_0]{}<\mu>[{\ \ }]\\
                        && \Open{x_0}[x_1]{}<\nu>[\mu]\\
                        && \MkBox[\mu \circ \nu]{x_1}
  \end{array}
\]
and
\[
  \arraycolsep=1.4pt
  \begin{array}{lcl}
    \multicolumn{3}{l}{\MComp*{}{\mu}{\nu} : \Modify[\mu \circ \nu]{A} \to \Modify{\Modify[\nu]{A}}}\\
    \MComp*{x}{\mu}{\nu} & \defeq & \Open{x}[x_0]{}<\mu \circ \nu>[1] \MkBox[\mu]{\MkBox[\nu]{x_0}}
  \end{array}
\]
We elide the 2-cell annotations on variables, as they are all identities (\ie~we only need
\ruleref{tm/var/counit}). Even in this small example the context equations for locks are
essential: for $\Modify{\Modify[\nu]{A}}$ to be a valid type we need that
$\LockCxV{\LockCxV{\Gamma}}<\nu> = \LockCxV{\Gamma}<\mu \circ \nu>$, which is ensured by
\ruleref{cx/compose}. Furthermore, observe that $\MComp{}{\mu}{\nu}$ crucially relies on the
multimodal elimination rule \ruleref{tm/modal-elim}: we must pattern-match on $x_0$, which is under
$\mu$ in the context.

Similarly, fixing $\IsTy{A}[1]$ we have
\begin{alignat*}{2}
  &\Triv{-} : \Modify[1]{A} \to A           \qquad &&\Triv*{-} : A \to \Modify[1]{A} \\
  &\Triv{x} \defeq \Open{x}[x_0]{x_0}<1>[1] \qquad &&\Triv*{x} \defeq \MkBox[1]{x}
\end{alignat*}
In both cases, these combinators are only propositionally inverse. For example, the proof for one direction of the composition combinator is
\[
  \arraycolsep=1.4pt
  \begin{array}{lcl}
    \_ & : & (x : \Modify{\Modify[\nu]{A}}) \to \Id{\Modify{\Modify[\nu]{A}}}{x}{\MComp*{\MComp{x}{\mu}{\nu}}{\mu}{\nu}}\\
    \_ & \defeq & \Lam[x]{}\Open{x}[x_0]{}<\mu>[1] \Open{x_0}[x_1]{}<\nu>[\mu] \Refl{\MkBox{\MkBox[\nu]{x}}}
  \end{array}
\]
This is in many ways a typical example: we use the modal elimination rule to induct on a
modally-typed term, which reduces it to a term of the form $\MkBox[]{-}$. This is just enough to
make various terms compute, and the result then follows by reflexivity.

As a final example, we will show that each modal type satisfies \emph{the \textsf{K}
axiom},\footnote{Not to be confused with Streicher's \emph{axiom K}.} a central axiom of
Kripke-style modal logics. This combinator will be immediately recognizable to functional
programmers: it is the term that witnesses that $\Modify{-}$ is an \emph{applicative
functor}~\cite{mcbride:2008}.
\[
  \arraycolsep=1.4pt
  \begin{array}{lcl}
    \ZApp{-}{-}{\mu} & : & \Modify{A \to B} \to \Modify{A} \to \Modify{B}\\
    \ZApp{f}{a}{\mu} & \defeq & \Open{f}[f_0]{\Open{a}[a_0]{\MkBox{\App{f_0}{a_0}}}<\mu>[1]}<\mu>[1]
  \end{array}
\]
We can also define a stronger combinator which corresponds to a dependent form of the Kripke axiom
\cite{clouston:dra:2018} along the same lines. As it generalizes $\ZApp{}{}{\mu}$ to dependent
products, this operation has precisely the same implementation but a more complex type:
\[
  \Modify{(x : A) \to B} \to{} (x_0 : \Modify{A}) \to (\Open{x_0}[x]{\Modify{B}}<\mu>[1])
\]
In order to ensure that $\Modify{B}$ is well-typed, the context must contain $\DeclVar{x}{A}$, but
instead we have bound $\DeclVar{x_0}{\Modify{A}}<1>$. We correct this mismatch by eliminating $x_0$
and binding the result to $x$.

\subsection{Idempotent Comonads in \MTT{}}
\label{sec:programming-in-mtt:comonads}
\NewDocumentCommand{\IComonadMode}{}{\mathcal{M}_{\mathsf{ic}}}

A great deal of prior work in modal type theory has focused on \emph{comonads}
\cite{pfenning-davies:2001,de-paiva:2015,shulman:2018,gratzer:2019}, and in particular
\emph{idempotent} comonads. \cite[Theorem 4.1]{shulman:2018} has shown that such modalities
necessitate changes to the judgmental structure, as the only idempotent comonads that are internally
definable in type theory are of the form $- \times U$ for some proposition $U$. In this section we
present a mode theory for idempotent comonads, and prove that the resulting type theory internally
satisfies the expected equations. In fact, we only use the combinators of the previous section.


We define the mode theory $\IComonadMode$ to consist of a single mode $m$, and a single non-trivial
morphism $\mu : m \to m$. We will enforce idempotence by setting $\mu \circ \mu = \mu$.  Finally, in
order to induce a morphism $\Modify{A} \to A$ we include a unique non-trivial 2-cell
$\epsilon : \mu \To 1$. In order to ensure that this 2-cell to be unique, we add equations such as
$\epsilon \Whisker 1_\mu = 1_\mu \Whisker \epsilon : \mu \circ \mu \To \mu$, where $\Whisker$
denotes the horizontal composition of 2-cells. The resulting mode theory is a 2-category, albeit a
very simple one: it is in fact only a \emph{poset-enriched} category.

We can show that $\Modify{A}$ is a comonad by defining the expected operations using the combinators
of Section~\ref{sec:programming-in-mtt:combinators}:
\begin{alignat*}{2}
  &\Dup_A : \Modify{A} \to \Modify{\Modify{A}} \qquad &&\Unbox_A : \Modify{A} \to A^\epsilon \\
  &\Dup_A \defeq \MComp*{}{\mu}{\mu}           \qquad &&\Unbox_A \defeq \Triv*{-} \circ \Coe{}{\mu}{1}{\epsilon}
\end{alignat*}
We must also show that $\Dup_A$ and $\Unbox_A$ satisfy the como\-nad laws, but that automatically
follows from general facts pertaining to $\mathbf{coe}$ and $\mathbf{comp}$.\footnote{In
  particular, our modal combinators satisfy a variant of the \emph{interchange law} of a 2-category.}
This is indicative of the benefits of using \MTT{}: every general result about it also applies to
this instance, including the canonicity theorem of Section~\ref{sec:semantics}.


\section{Algebraic Syntax}
\label{sec:algebraic-mtt}

Until this point we have presented a curated, high-level view of \MTT{}, and we have avoided any
discussion of its metatheory. Yet, syntactic matters can be quite complex, and have historically
proven to be sticking points for modal type theory. While such details are not necessary for the
casual reader, it is essential to validate that \MTT{} is syntactically well-behaved, enjoying \eg{}
a substitution principle. The aim of this section is to provide a setting for this study: we
introduce the formal counterpart of \MTT{}, which is given as a \emph{generalized algebraic theory}
(GAT)~\cite{cartmell:1978,kaposi:qiits:2019}.

Historically, GATs were used in the semantics of type theory, but modern techniques show that they
are also useful in the analysis of syntax. For example, recasting \MTT{} as a GAT naturally leads us
to include \emph{explicit substitutions}~\cite{curien:1990, martin-lof:1992, granstrom:2011} in the
syntax. Thus, substitution in \MTT{} is not a metatheoretic operation on raw terms, but a syntactic
operation within the theory. This presentation helps us carefully state the equations that govern
substitutions and their interaction with type formers. We consequently obtain an elegant
\emph{substitution calculus}, which can often be quite complex for modal type theories.

This approach proffers a number of technical advantages. Amongst other things, the theorems proven in the aforementioned works on GATs imply the following points:
\begin{enumerate}
  \item
    We absolve ourselves from having to prove tedious syntactic metatheorems, e.g. admissibility of
    substitution.
  \item
    We automatically obtain a notion of \emph{model} of our theory, which is given in entirely
    algebraic terms.
  \item
    We obtain a notion of \emph{homomorphism of models}. (NB that this notion is rather
    \emph{strict} and not fit for every purpose.)
  \item
    In an equally automatic fashion, we obtain an \emph{initial model} for the algebraic theory,
    which we consider as our main formal object of study.
  \item
    The unique morphism of models from this initial model to any other is the \emph{semantic
    interpretation map}. We then have no need to explicitly describe these semantic maps and prove
    that they are well-defined on derivations, as done \eg~by \cite{hofmann:1997}.
\end{enumerate}

While this approach is straightforward and uncluttered, some readers might object to the lack of a
more traditional formulation, \eg{} a \emph{named syntax} with variables and a metatheoretic
substitution operation, like the one we informally presented in Section~\ref{sec:towards-mtt}. We believe
that it is indeed possible to define such a syntax and systematically show how to \emph{elaborate}
its terms to the algebraic syntax.

However, such a named syntax would not be directly suitable for implementation: for that purpose we
ought to develop an entirely different \emph{algorithmic syntax}. We believe that such a syntax can
be constructed as an extension of existing \emph{bidirectional} presentations of type
theory~\cite{coquand:1996,pierce:2000} as has been done for existing modal
calculi~\cite{gratzer:2019}. Such a bidirectional presentation would occupy a midpoint between the
maximally annotated algebraic syntax we present here, and the more typical named syntax of
Section~\ref{sec:towards-mtt}: it would contain only a select few annotations to ensure the decidability
of typechecking, yet maintain readability. The development of such a syntax is a substantial
undertaking that requires a proof of normalization, and is orthogonal to the foundational
metatheoretic results that we seek to develop here. We thus refrain from developing it, and instead
work directly with the GAT.

\subsection{Sorts}

We begin by defining the different \emph{sorts} (contexts, types, terms, \etc{}) that constitute our
type theory. In order to support multiple modes, our sorts will be parameterized in modes. Thus,
rather than having a single sort of types, we will have a sort of types \emph{at mode $m \in
\Mode$}, and likewise for contexts at mode $m$, terms at mode $m$, \etc{}

Moreover, we take care to index our types by \emph{levels}. The reason for doing so was discussed in
Section~\ref{paragraph:universes}: we seek to introduce a hierarchy of sizes, which we can then use to
introduce universes \`a la \cite{coquand:2013}. We stratify our types in two levels, drawn from the
set $\Lvl = \{0, 1\}$. There are no technical obstacles on the way to a richer hierarchy, but two
levels suffice for our purposes: we aim to divide our types into \emph{small types} (\ie~those that
can be reified in a universe) and \emph{large types} (which also include the universe itself). In
order to enforce cumulativity we will also include an explicit \emph{coercion} operator, which
includes small types into large types.

The levelled approach raises an obvious question: on which level should we admit terms? We could
follow the approach of \cite{sterling:2019} in allowing terms at both, but this requires the
introduction of term-level coercions, which then require equations relating term formers at
different levels. Thus, for the sake of simplicity we will only allow the formation of terms at
large types. Similarly, we will only allow the extension of a context by a large type.

\MTT{} has four families of sorts, which are introduced by the following rules:

\begin{mathpar}
  \inferrule{
    m : \Mode
  }{
    \Sort{\Cx{m}}
  }
  \and
  \inferrule{
    \ell : \Lvl\\
    m : \Mode\\
    \Gamma : \Cx{m}
  }{
    \Sort{\Ty{m}(\Gamma)}
  }
  \and
  \inferrule{
    m : \Mode\\
    \Gamma : \Cx{m}\\
    A : \Ty{m}[1](\Gamma)
  }{
    \Sort{\Tm{m}(\Gamma, A)}
  }
  \and
  \inferrule{
    m : \Mode\\
    \Gamma,\Delta : \Cx{m}
  }{
    \Sort{\Sub{m}(\Gamma, \Delta)}
  }
\end{mathpar}
In the interest of clarity we will use the following shorthands:
\begin{align*}
  \IsCx{\Gamma} &\defeq \Gamma : \Cx{m} &
  \IsTy{A} &\defeq A : \Ty{m}(\Gamma) \\
  \IsTm{M}{A} &\defeq M : \Tm{m}(\Gamma, A) &
  \IsSb{\delta}{\Delta} &\defeq \delta : \Sub{m}(\Gamma, \Delta)
\end{align*}
Even though we will use this more familiar notation, we will take no prisoners in terms of rigour: we will carefully avoid overloading and ambiguity, and we will enforce \emph{presupposition}.

\subsection{Judgments}

We shall now introduce the type theory itself by writing down the constructors and equalities of its GAT. In the interest of brevity, we elide a number of standard rules, including
\begin{itemize}
  \item the congruence rules pushing substitutions inside terms and types;
  \item the congruence rules pushing explicit lifts inside of type formers;
  \item the associativity, unit, and weakening laws for the explicit substitutions;
  \item the $\beta$ laws for $\Pi$, $\Sigma$, $\Bool$ and $\textsf{Id}$;
  \item the $\eta$ laws for $\Pi$ and $\Sigma$;
\end{itemize}

The specification of the GAT is given in Figures
\ref{fig:mtt-gat-contexts}--\ref{fig:mtt-gat-eqsubst}. As the judgments are defined in a mutually
recursive manner, the division of the rules between different figures is merely presentational.
Given $\IsSb[\Delta]{\gamma}{\Gamma}$ and $\IsTy[\LockCx{\Gamma}]{A}$ we write
\[
  \IsSb[\ECx{\Delta}{\Sb{A}{\LockSb{\gamma}}}]
    {\gamma^{+} \defeq \ESb{(\gamma \circ \Wk)}{\Var{0}}}
    {\ECx{\Gamma}{A}}
\]
for the `weakened' substitution.

\begin{figure}
  \small
\begin{mathpar}
  \JdgFrame{\IsCx{\Gamma}}\\
  \inferrule{ }{
    \IsCx{\Emp}
  }
  \and
  \inferrule{
    \IsCx{\Gamma} \\
    \mu : \Hom[\Mode]{n}{m}
  }{
    \IsCx{\LockCx{\Gamma}}<n>
  }
  \and
  \inferrule{
    \IsCx{\Gamma}<m>\\
    \mu : \Hom[\Mode]{n}{m}\\
    \IsTy[\LockCx{\Gamma}]{A}[1]<n>
  }{
    \IsCx{\ECx{\Gamma}{A}}<m>
  }
  \and
  \inferrule{
    \IsCx{\Gamma} \\
    \nu : \Hom[\Mode]{o}{n} \\
    \mu : \Hom[\Mode]{n}{m}
  }{
    \EqCx{\LockCx{\LockCx{\Gamma}}<\nu>}{\LockCx{\Gamma}<\mu \circ \nu>}<o>
  }
  \and
  \inferrule{
    \IsCx{\Gamma}
  }{
    \EqCx{\LockCx{\Gamma}<1>}{\Gamma}
  }
  \end{mathpar}
  \caption{\MTT{} Contexts}
  \label{fig:mtt-gat-contexts}
\end{figure}

\begin{figure}[p]
  \small
  \begin{mathpar}
  \JdgFrame{\IsTy{A}}
  \\
  \inferrule{
    \IsCx{\Gamma}
  }{
    \IsTy{\Bool}
  }
  \and
  \inferrule{
    \IsCx{\Gamma}
  }{
    \IsTy{\Uni}[1]
  }
  \and
  \inferrule{
    \IsCx{\Gamma}\\
    \IsTm{M}{\Uni}
  }{
    \IsTy{\Dec{M}}[0]
  }
  \and
  \inferrule{
    \ell \le \ell'\\
    \IsCx{\Gamma}\\
    \IsTy{A}
  }{
    \IsTy{\TyLift{A}}[\ell']
  }
  \and
  \inferrule{
    \IsCx{\Gamma}\\
    \IsTy{A}\\
    \IsTm{M,N}{\TyLift{A}}
  }{
    \IsTy{\Id{A}{M}{N}}
  }
  \and
  \inferrule{
    \IsCx{\Gamma}\\
    \mu : \Hom[\Mode]{n}{m}\\
    \IsTy[\LockCx{\Gamma}]{A}<n>\\
  }{
    \IsTy{\Modify{A}}
  }
  \and
  \inferrule{
    \mu : \Hom[\Mode]{n}{m}\\
    \IsCx{\Gamma}\\
    \IsTy[\LockCx{\Gamma}]{A}<n>\\
    \IsTy[\ECx{\Gamma}{\TyLift{A}}]{B}
  }{
    \IsTy{\Fn{A}{B}}
  }
  \and
  \inferrule{
    \IsCx{\Gamma}\\
    \IsTy[\Gamma]{A}\\
    \IsTy[\ECx{\Gamma}{\TyLift{A}}<1>]{B}
  }{
    \IsTy{\Prod{A}{B}}
  }
  \and
  \inferrule{
    \IsCx{\Gamma,\Delta}\\
    \IsTy[\Delta]{A}\\
    \IsSb{\delta}{\Delta}
  }{
    \IsTy{\Sb{A}{\delta}}
  }
  \end{mathpar}
  \caption{\MTT{} Types}
   \label{ref:mtt-gat-types}
\end{figure}

\begin{figure}[p]
  \small
\begin{mathpar}
  \JdgFrame{\IsSb{\delta}{\Delta}}
  \\
  \and
  \inferrule{
    \IsCx{\Gamma}
  }{
    \IsSb{\Emp}{\Emp}
  }
  \and
  \inferrule{
    \IsCx{\Gamma}<n>\\
    \mu : \Hom[\Mode]{n}{m}\\
    \IsTy[\LockCx{\Gamma}]{A}[1]<n>
  }{
    \IsSb[\ECx{\Gamma}{A}]{\Wk}{\Gamma}
  }
  \and
  \inferrule{
    \IsCx{\Gamma}
  }{
    \IsSb[\Gamma]{\ISb}{\Gamma}
  }
  \and
  \inferrule{
    \IsCx{\Gamma, \Delta,\Xi}\\
    \IsSb[\Gamma]{\gamma}{\Delta}\\
    \IsSb[\Delta]{\delta}{\Xi}
  }{
    \IsSb[\Gamma]{\delta \circ \gamma}{\Xi}
  }
  \and
  \inferrule{
    \IsCx{\Gamma, \Delta}\\
    \mu : \Hom[\Mode]{n}{m}\\
    \IsSb{\delta}{\Delta}
  }{
    \IsSb[\LockCx{\Gamma}]{\LockSb{\delta}}{\LockCx{\Delta}}<n>
  }
  \and
  \inferrule{
    \IsCx{\Gamma} \\
    \mu, \nu : \Hom[\Mode]{n}{m} \\
    \alpha : \nu \To \mu
  }{
    \IsSb[\LockCx{\Gamma}]{\Key{\alpha}{\Gamma}}{\LockCx{\Gamma}<\nu>}<n>
  }
  \and
  \inferrule{
    \mu : \Hom[\Mode]{n}{m} \\
    \IsCx{\Gamma, \Delta} \\
    \IsSb{\delta}{\Delta}<m> \\
    \IsTy[\LockCx{\Delta}]{A}[1]<n> \\
    \IsTm[\LockCx{\Gamma}]{M}{\Sb{A}{\LockSb{\delta}}}<n>
  }{
    \IsSb{\ESb{\delta}{M}}{\ECx{\Delta}{A}}
  }
  \end{mathpar}
  \caption{\MTT{} Substitutions}
   \label{fig:mtt-gat-subst}
\end{figure}

\begin{figure}[p]
  \small
  \begin{mathpar}
  \JdgFrame{\IsTm{M}{A}}
  \\
  \inferrule{
    \mu : \Hom[\Mode]{n}{m}\\
    \IsCx{\Gamma}\\
    \IsTy[\LockCx{\Gamma}]{A}[1]<n>
  }{
    \IsTm[\LockCx{\ECx{\Gamma}{A}}]{\Var{0}}{\Sb{A}{\LockSb{\Wk}}}<n>
  }
  \and
  \inferrule{
    \IsCx{\Gamma}
  }{
    \IsTm{\True,\False}{\Bool}
  }
  \and
  \inferrule{
    \IsCx{\Gamma}\\
    \IsTy[\ECx{\Gamma}{\Bool}<1>]{A}[1]\\
    \IsTm{M_t}{\Sb{A}{\ESb{\ISb}{\True}}}\\
    \IsTm{M_f}{\Sb{A}{\ESb{\ISb}{\False}}}\\
    \IsTm{N}{\Bool}
  }{
    \IsTm{\BoolRec{A}{M_t}{M_f}{N}}{\Sb{A}{\ESb{\ISb}{N}}}
  }
  \and
  \inferrule{
    \IsCx{\Gamma}\\
    \IsTy{A}[0]
  }{
    \IsTm{\Enc{A}}{\Uni}
  }
  \and
  \inferrule{
    \IsCx{\Gamma} \\
    \IsTy{A}[1]\\
    \IsTm{M}{A}
  }{
    \IsTm{\Refl{M}}{\Id{A}{M}{M}}
  }
  \and
  \inferrule{
    \IsCx{\Gamma}\\
    \IsTy{A}[1]\\
    \IsTy[\ECx{\ECx{\ECx{\Gamma}{A}<1>}{\Sb{A}{\Wk}}<1>}{\Id{\Sb{A}{\Wk[2]}}{\Var{1}}{\Var{0}}}<1>]{B}[1]\\
    \IsTm[\ECx{\Gamma}{A}<1>]{M}{\Sb{B}{\ESb{\ESb{\ESb{\Wk}{\Var{0}}}{\Var{0}}}{\Refl{\Var{0}}}}}\\
    \IsTm{N_0,N_1}{A}\\
    \IsTm{P}{\Id{A}{N_0}{N_1}}
  }{
    \IsTm{\IdRec{B}{M}{P}}{\Sb{B}{\ESb{\ESb{\ESb{\ISb}{N_0}}{N_1}}{P}}}
  }
  \and
  \inferrule{
    \IsCx{\Gamma}\\
    \mu : \Hom[\Mode]{n}{m}\\
    \IsTy[\LockCx{\Gamma}]{A}[1]<n>\\
    \IsTm[\LockCx{\Gamma}]{M}{A}<n>
  }{
    \IsTm{\MkBox{M}}{\Modify{A}}
  }
  \and
  \inferrule{
    \nu : \Hom[\Mode]{o}{n} \\
    \mu : \Hom[\Mode]{n}{m} \\
    \IsCx{\Gamma} \\
    \IsTy[\LockCx{\LockCx{\Gamma}}<\nu>]{A}[1]<o> \\
    \IsTm[\LockCx{\Gamma}]{M_0}{\Modify[\nu]{A}}<n> \\
    \IsTy[\ECx{\Gamma}{\Modify[\nu]{A}}]{B}[1] \\
    \IsTm[\ECx{\Gamma}{A}<\mu \circ \nu>]{M_1}{\Sb{B}{\ESb{\Wk}{\MkBox[\nu]{\Var{0}}}}}
  }{
    \IsTm{\Open{M_0}{M_1}<\nu>[\mu]}{\Sb{B}{\ESb{\ISb}{M_0}}}
  }
  \and
  \inferrule{
    \mu : \Hom[\Mode]{n}{m}\\
    \IsCx{\Gamma}\\
    \IsTy[\LockCx{\Gamma}]{A}[1]<n>\\
    \IsTy[\ECx{\Gamma}{A}<\mu>]{B}[1]\\
    \IsTm[\ECx{\Gamma}{A}<\mu>]{M}{B}
  }{
    \IsTm{\Lam{M}}{\Fn{A}{B}}
  }
  \and
  \inferrule{
    \mu : \Hom[\Mode]{n}{m}\\
    \IsCx{\Gamma}\\
    \IsTy[\LockCx{\Gamma}]{A}[1]<n>\\
    \IsTy[\ECx{\Gamma}{A}<\mu>]{B}[1]\\
    \IsTm{M_0}{\Fn{A}{B}}\\
    \IsTm[\LockCx{\Gamma}]{M_1}{A}<n>
  }{
    \IsTm{\App{M_0}{M_1}}{\Sb{B}{\ESb{\ISb}{M_1}}}
  }
  \and
  \inferrule{
    \IsCx{\Gamma}\\
    \IsTy[\Gamma]{A}[1]\\
    \IsTy[\ECx{\Gamma}{A}<1>]{B}[1]\\
    \IsTm[\Gamma]{M_0}{A}\\
    \IsTm{M_1}{\Sb{B}{\ESb{\ISb}{M_0}}}
  }{
    \IsTm{\Pair{M_0}{M_1}}{\Prod{A}{B}}
  }
  \and
  \inferrule{
    \IsCx{\Gamma}\\
    \IsTy[\Gamma]{A}[1]\\
    \IsTy[\ECx{\Gamma}{A}<1>]{B}[1]\\
    \IsTm{M}{\Prod{A}{B}}
  }{
    \IsTm[\Gamma]{\Proj[0]{M}}{A}\\
    \IsTm{\Proj[1]{M}}{\Sb{B}{\ESb{\ISb}{\Proj[0]{M}}}}
  }
  \and
  \inferrule{
    \IsCx{\Gamma,\Delta}\\
    \IsTy[\Delta]{A}[1]\\
    \IsSb{\delta}{\Delta}\\
    \IsTm[\Delta]{M}{A}
  }{
    \IsTm{\Sb{M}{\delta}}{\Sb{A}{\delta}}
  }
  \end{mathpar}
  \caption{\MTT{} Terms}
  \label{fig:mtt-gat-terms}
\end{figure}

\begin{figure}[p]
  \small
  \begin{mathpar}
  \JdgFrame{\EqTy{A}{B}}
  \\
  \inferrule{
    \mu : \Hom[\Mode]{n}{m}\\
    \IsCx{\Gamma, \Delta}\\
    \IsSb[\Gamma]{\delta}{\Delta}\\
    \IsTy[\LockCx{\Delta}]{A}<n>
  }{
    \EqTy{\Sb{\Modify{A}}{\delta}}{\Modify{\Sb{A}{\LockSb{\delta}}}}
  }
  \and
  \inferrule{
    \mu : \Hom[\Mode]{n}{m}\\
    \IsCx{\Gamma, \Delta}\\
    \IsSb[\Gamma]{\delta}{\Delta}\\
    \IsTy[\LockCx{\Gamma}]{A}<n>\\
    \IsTy[\ECx{\Gamma}{\TyLift{A}}]{B}
  }{
    \EqTy{\Sb{(\Fn{A}{B})}{\delta}}
         {\Fn{\Sb{A}{\LockSb{\delta}}}
             {\Sb{B}{\delta^{+}}}
         }
  }
  \and
  \inferrule{
    \IsTy{A}
  }{
    \EqTy{\TyLift{A}}{A}
  }
  \and
  \inferrule{
    \ell_0 \le \ell_1 \le \ell_2\\
    \IsTy{A}[\ell_0]
  }{
    \EqTy{\TyLift{\TyLift{A}}}{\TyLift{A}}[\ell_2]
  }
  \and
  \inferrule{
    \ell \le \ell'\\
    \mu \in \Hom[\Mode]{n}{m}\\
    \IsCx{\Gamma}\\
    \IsTy[\LockCx{\Gamma}]{A}<n>\\
    \IsTy[\ECx{\Gamma}{\TyLift{A}}]{B}
  }{
    \EqTy{\TyLift{(\Fn{A}{B})}}{\Fn{\TyLift{A}}{\TyLift{B}}}[\ell']
  }
  \and
  \inferrule{
    \IsCx{\Gamma}\\
    \IsTy{A}[0]
  }{
    \EqTy{\Dec{\Enc{A}}}{A}[0]
  }
  \and
  \inferrule{
    \IsCx{\Gamma, \Delta}\\
    \IsSb[\Gamma]{\delta}{\Delta}
  }{
    \EqTy{\Sb{\Uni}{\delta}}{\Uni}[1]
  }
  \end{mathpar}
  \caption{Equality of Types}
  \label{fig:mtt-gat-eqtypes}
\end{figure}

\begin{figure}[p]
  \small
  \begin{mathpar}
  \JdgFrame{\EqTm{M}{N}{A}}
  \\
  \inferrule{
    \mu : \Hom[\Mode]{n}{m}\\
    \IsSb{\delta}{\Delta}\\
    \IsTy[\LockCx{\Delta}]{A}[1]<n>\\
    \IsTm[\LockCx{\Gamma}]{M}{\Sb{A}{\LockSb{\delta}}}<n>
  }{
    \EqTm[\LockCx{\Gamma}]{\Sb{\Var{0}}{\LockSb{(\ESb{\delta}{M})}}}{M}{\Sb{A}{\LockSb{\delta}}}<n>
  }
  \and
  \inferrule{
    \IsCx{\Gamma}\\
    \IsTm{M}{\Uni}
  }{
    \EqTm{\Enc{\Dec{M}}}{M}{\Uni}
  }
  \and
  \inferrule{
    \nu : \Hom[\Mode]{o}{n}\\
    \mu : \Hom[\Mode]{n}{m}\\
    \IsCx{\Gamma}\\
    \IsTy[\LockCx{\LockCx{\Gamma}}<\nu>]{A}[1]<o>\\
    \IsTm[\LockCx{\LockCx{\Gamma}}<\nu>]{M_0}{A}<o>\\
    \IsTy[\ECx{\Gamma}{\Modify[\nu]{A}}]{B}[1]\\
    \IsTm[\ECx{\Gamma}{A}<\mu \circ \nu>]{M_1}{\Sb{B}{\ESb{\Wk}{\MkBox{\Var{0}}}}}
  }{
    \EqTm{
      \Open{\MkBox[\nu]{M_0}}{M_1}<\nu>[\mu]
    }{
      \Sb{M_1}{\ESb{\ISb}{M_0}}
    }{
      \Sb{B}{\ESb{\ISb}{\MkBox[\nu]{M_0}}}
    }
  }
  \and
  \inferrule{
    \mu : \Hom[\Mode]{n}{m}\\
    \IsCx{\Delta,\Gamma}\\
    \IsSb[\Gamma]{\delta}{\Delta}\\
    \IsTy[\LockCx{\Delta}]{A}[1]<n>\\
    \IsTm[\LockCx{\Delta}]{M}{A}<n>
  }{
    \EqTm{\Sb{\MkBox{M}}{\delta}}{\MkBox{\Sb{M}{\LockSb{\delta}}}}{\Modify{\Sb{A}{\LockSb{\delta}}}}
  }
  \and
  \inferrule{
    \nu : \Hom[\Mode]{o}{n}\\
    \mu : \Hom[\Mode]{n}{m}\\
    \IsCx{\Gamma,\Delta}\\
    \IsSb[\Gamma]{\delta}{\Delta}\\
    \IsTy[\LockCx{\LockCx{\Delta}}<\nu>]{A}[1]<o>\\
    \IsTm[\LockCx{\Delta}]{M_0}{\Modify[\nu]{A}}<n>\\
    \IsTy[\ECx{\Delta}{\Modify[\nu]{A}}]{B}[1]\\
    \IsTm[\ECx{\Delta}{A}<\mu \circ \nu>]{M_1}{\Sb{B}{\ESb{\Wk}{\MkBox{\Var{0}}}}}
  }{
    \EqTm{
      \Sb{\parens{\Open{M_0}{M_1}<\nu>[\mu]}}{\delta}
    }{
      \Open{\Sb{M_0}{\LockCx{\delta}}}{\Sb{M_1}{\delta^{+}}}<\nu>[\mu] 
    }{
      \Sb{B}{\ESb{\delta}{\Sb{M_0}{\LockCx{\delta}}}}
    }
  }
  \end{mathpar}
  \caption{Equality of Terms}
  \label{fig:mtt-gat-eqterms}
\end{figure}

\begin{figure}
  \small
  \begin{mathpar}
  \JdgFrame{\EqSb{\gamma}{\delta}{\Delta}}
  \\
  \and
  \inferrule{
    \IsCx{\Gamma_0,\Gamma_1}<n>\\
    \IsCx{\Delta}<n>\\
    \IsTy[\LockCx{\Delta}]{A}[1]<m>\\
    \mu : \Hom[\Mode]{m}{n}\\
    \IsSb[\Gamma_0]{\gamma}{\Gamma_1}<n>\\
    \IsSb[\Gamma_1]{\delta}{\Delta}<n>\\
    \IsTm[\LockCx{\Gamma_1}]{M}{\Sb{A}{\LockSb{\delta}}}<m>
  }{
    \EqSb[\Gamma_0]{
      (\ESb{\delta}{M}) \circ \gamma
    }{
      \ESb{(\delta \circ \gamma)}{\Sb{M}{\LockSb{\gamma}}}
    }{
      \ECx{\Delta}{A}
    }<n>
  }
  \and
  \inferrule{
    \IsCx{\Gamma,\Delta}<o>\\
    \mu : \Hom[\Mode]{m}{n}\\
    \nu : \Hom[\Mode]{n}{o}\\
    \IsSb{\delta}{\Delta}
  }{
    \EqSb[\LockCx{\Gamma}<\nu \circ \mu>]{
      \LockSb{\delta}<\nu \circ \mu>
    }{
      \LockSb{\LockSb{\delta}<\nu>}<\mu>
    }{
      \LockCx{\Delta}<\nu \circ \mu>
    }
  }
  \and
  \inferrule{
    \IsCx{\Gamma,\Delta}\\
    \IsSb{\delta}{\Delta}
  }{
    \EqSb{\LockSb{\delta}<1>}{\delta}{\Delta}
  }
  \and
  \inferrule{
    \IsCx{\Gamma}<n>\\
    \mu : \Hom[\Mode]{m}{n}\\
  }{
    \EqSb[\LockCx{\Gamma}]{\LockSb{\ISb}}{\ISb}{\LockCx{\Gamma}}
  }
  \and
  \inferrule{
    \IsCx{\Gamma,\Delta,\Xi}<n>\\
    \mu : \Hom[\Mode]{m}{n}\\
    \IsSb{\delta}{\Delta}<n>\\
    \IsSb[\Delta]{\xi}{\Xi}<n>
  }{
    \EqSb[\LockCx{\Gamma}]{\LockSb{(\xi \circ \delta)}}{\LockSb{\xi} \circ \LockSb{\delta}}{\LockCx{\Xi}}
  }
  \and
  \inferrule{
    \IsCx{\Gamma}<n>\\
    \mu : \Hom[\Mode]{m}{n}
  }{
    \EqSb[\LockCx{\Gamma}]{\ISb}{\Key{1_\mu}{\Gamma}}{\LockCx{\Gamma}}
  }
  \and
  \inferrule{
    \IsCx{\Gamma,\Delta}<n>\\
    \mu,\nu : \Hom[\Mode]{m}{n}\\
    \IsSb{\delta}{\Delta}<n>\\
    \alpha : \nu \To \mu
  }{
    \EqSb[\LockCx{\Gamma}]{
      \Key{\alpha}{\Gamma} \circ \parens{\LockSb{\delta}}
    }{
      \parens{\LockSb{\delta}<\nu>} \circ \Key{\alpha}{\Delta}
    }{
      \LockCx{\Delta}<\nu>
    }
  }
  \and
  \inferrule{
    \IsCx{\Gamma}<m>\\
    \mu_0, \mu_1, \mu_2 : \Hom[\Mode]{n}{m}\\
    \alpha_0 : \mu_0 \To \mu_1 \\
    \alpha_1 : \mu_1 \To \mu_2
  }{
    \EqSb[
      \LockCx{\Gamma}<\mu_2>
    ]{
      \Key{\alpha_1 \circ \alpha_0}{\Gamma}
    }{
      \Key{\alpha_0}{\Gamma} \circ \Key{\alpha_1}{\Gamma}
    }{
      \LockCx{\Gamma}<\mu_0>
    }<n>
  }
  \and
  \inferrule{
    \IsCx{\Gamma} \\
    \nu_0, \nu_1 : \Hom[\Mode]{o}{n} \\
    \mu_0, \mu_1 : \Hom[\Mode]{n}{m} \\
    \beta : \nu_0 \Rightarrow \nu_1 \\
    \alpha : \mu_0 \Rightarrow \mu_1
  }{
    \EqSb[
      \LockCx{\Gamma}<\mu_0 \circ \nu_0>
    ]{
      \Key{\alpha \Whisker \beta}{\Gamma}
    }{
      \LockSb{\Key{\alpha}{\Gamma}}<\nu_1> \circ \Key{\beta}{\LockCx{\Gamma}<\mu_0>}
    }{
      \LockCx{\Gamma}<\mu_1 \circ \nu_1>
    }<o>
  }
  \end{mathpar}
  \caption{Equality of Substitutions}
  \label{fig:mtt-gat-eqsubst}
\end{figure}

\subsection{Discussion}

We record some points on the generalized algebraic theory.

\paragraph{Modal dependent products}

The algebraic presentation of \MTT{} includes a primitive \emph{modal dependent product} type
$\Fn{A}{B}$. This is a combination of the modality $\Modify$ and the ordinary dependent product.
Using a named syntax, it may be understood as
\[
  \FnV{x}{A}{B} \defeq (x_0 : \Modify{A}) \to (\Open{x_0}[x]{B}<\mu>[1])
\]
However, the modal types of \MTT{} do not readily support a definitional $\eta$-equality, so this
definition is not equivalent to the modal dependent product of the GAT. We use the latter because it
is convenient for programming, and also has a natural semantics, which we will present in
Section~\ref{sec:semantics:standard-pi}.

\paragraph{Modal substitutions}

In addition to the usual rules, \MTT{} features substitutions corresponding to the 1-cells and
2-cells of the mode theory. First, recall that for each modality $\mu : n \to m$ we have the
operation $\Lock_\mu$ on contexts. Its action extends to substitutions:
\begin{mathpar}
  \inferH{sb/lock}{
    \mu : n \to m\\
    \IsSb{\delta}{\Delta}
  }{
    \IsSb[\LockCx{\Gamma}]{\LockSb{\delta}}{\LockCx{\Delta}}<n>
  }
\end{mathpar}
Second, each 2-cell $\alpha : \mu \To \nu$ induces a \emph{natural transformation} between
$\Lock_\nu$ and $\Lock_\mu$, whose component at $\Gamma$ is the `key' substitution
\begin{mathpar}
  \inferH{sb/key}{
    \alpha : \mu \To \nu
  }{
    \IsSb[\LockCx{\Gamma}<\nu>]{\Key{\alpha}{\Gamma}}{\LockCx{\Gamma}}<n>
  }
\end{mathpar}
Recalling that $\Coop{\Mode}$ is the 2-category with morphisms and 2-cells opposite from $\Mode$, we
see that these substitutions come with equations postulating that $\LockSb{-}<\mu>$ is a functor,
$\Key{\alpha}{\Gamma}$ is a natural transformation, and that together they form a 2-functor
$\Coop{\Mode} \to \CAT$. As a consequence, our type theory is forced to contain a calculus of
(strict) 2-categories. Indeed, the given equations for keys above suffice to derive the two ways of
internally stating the \emph{interchange laws}, viz.
\begin{mathpar}
  \inferrule{
    \IsCx{\Gamma} \\
    \nu_0,\nu_1,\nu_2 : \Hom[\Mode]{o}{n} \\
    \mu_0,\mu_1,\mu_2 : \Hom[\Mode]{n}{m} \\
    \alpha_0 : \mu_0 \To \mu_1 \\
    \alpha_1 : \mu_1 \To \mu_2 \\
    \beta_0 : \nu_0 \To \nu_1 \\
    \beta_1 : \nu_1 \To \nu_2 \\
  }{
    \EqSb[\LockCx{\Gamma}<\mu_2 \circ \nu_2>]{
      \Key{\alpha_0 \Whisker \beta_0}{\Gamma}
        \circ
      \Key{\alpha_1 \Whisker \beta_1}{\Gamma}
    }{
      \LockSb{\Key{\alpha_1 \circ \alpha_0}{\Gamma}}<\nu_0>
        \circ
      \Key{\beta_1 \circ \beta_0}{\LockCx{\Gamma}<\mu_2>}
    }{
      \LockCx{\Gamma}<\mu_0 \circ \nu_0>
    }<o>
  }
  \and
  \inferrule{
    \IsCx{\Gamma} \\
    \nu_0,\nu_1,\nu_2 : \Hom[\Mode]{o}{n} \\
    \mu_0,\mu_1,\mu_2 : \Hom[\Mode]{n}{m} \\
    \alpha_0 : \mu_0 \To \mu_1 \\
    \alpha_1 : \mu_1 \To \mu_2 \\
    \beta_0 : \nu_0 \To \nu_1 \\
    \beta_1 : \nu_1 \To \nu_2 \\
  }{
    \EqSb[\LockCx{\Gamma}<\mu_2 \circ \nu_2>]{
      \Key{\alpha_0 \Whisker \beta_0}{\Gamma}
        \circ
      \Key{\alpha_1 \Whisker \beta_1}{\Gamma}
    }{
      \Key{\beta_1 \circ \beta_0}{\LockCx{\Gamma}<\mu_0>}
        \circ
      \LockSb{\Key{\alpha_1 \circ \alpha_0}{\Gamma}}<\nu_2>
    }{
      \LockCx{\Gamma}<\mu_0 \circ \nu_0>
    }<o>
  }
\end{mathpar}
In fact, the second version of the interchange law follows from the first one and the equation that
expresses the naturality of $\Key{-}{-}$. Conversely, except the two laws for the identity 2-cell
and naturality, the given equations follow from either one of the two interchange laws.

While it is no longer necessary to prove that substitution is \emph{admissible} in the setting of
the GAT, we would still like to show that explicit substitutions can be eliminated on closed terms.
The proof of canonicity implicitly contains such an algorithm, but that is overkill: a simple,
direct argument proves that explicit substitutions can be propagated down to variables. Moreover, we
may define the admissible operation mentioned in Section~\ref{sec:towards-mtt} by
\begin{align*}
  A^\alpha &\defeq \Sb{A}{\Key{\alpha}{\Gamma}}
  &
  M^\alpha &\defeq \Sb{M}{\Key{\alpha}{\Gamma}}
\end{align*}
We may then use the aforementioned algorithm to eliminate the keys.

\paragraph{Pushing substitutions under modalities}

In order for the aforementioned algorithm to work, we must specify how substitutions commute with
the modal connectives of \MTT{}. Unlike previous work~\cite{gratzer:tech-report:2019}, the
necessary equations are straightforward:
\begin{align*}
  \Sb{\Modify{A}}{\delta} &= \Modify{\Sb{A}{\LockSb{\delta}}} &
  \Sb{\MkBox{M}}{\delta} &= \MkBox{\Sb{M}{\LockSb{\delta}}}
\end{align*}

This simplicity is not coincidental. Previous modal type theories included rules that, in one way or
another, \emph{trimmed} the context during type checking: some removed
variables~\cite{prawitz:1965,pfenning-davies:2001,shulman:2018}, while others erased context
formers, \eg~locks~\cite{clouston:dra:2018,gratzer:2019}. In either case, it was necessary to show
that the trimming operation, which we may write as $\Verts{\Gamma}$, is functorial: $\Gamma \vdash
\delta : \Delta$ should imply $\Verts{\Gamma} \vdash \Verts{\delta} : \Verts{\Delta}$.
Unfortunately, the proof of this fact is almost always very complicated. Some type theories avoid it
by `forcing' substitution to be admissible using delayed substitutions~\cite{bierman:2000,
licata:2017}, but this causes serious complications in the equational theory.

\MTT{} circumvents this by avoiding any context trimming. As a result, we need neither delayed
substitutions nor a complex proof of admissibility.


\section{Models}
\label{sec:semantics}

In the preceding section we presented the formal definition of \MTT{} in the form of a GAT. As a
consequence we automatically obtained a \emph{category of models} of \MTT{}, as well as
\emph{(strict) homomorphisms} between them~\cite{cartmell:1978, kaposi:qiits:2019}. Moreover, this
category of models had an initial object, \ie~the syntax of \MTT{} itself. This category of models
is inhabited by algebras for this GAT. Hence, showing that a mathematical structure is a model of
\MTT{} becomes a laborious task: one must show that each and every construct can be interpreted in a
manner that makes the postulated equations hold. In this section we shall take on the task of
decomposing these algebraic models into more tractable pieces.

A moment's thought reveals that many of the equations of the GAT given in
Section~\ref{sec:algebraic-mtt}
are rather close to the familiar notion of \emph{categories with families}
(CwFs)~\cite{dybjer:1996}, which can be adapted to the present setting.\footnote{The conference
version of this paper used such a presentation in the interest of brevity.} However, we will take
things a bit further by opting for a category-theoretic reformulation of CwFs known as \emph{natural
models}~\cite{awodey:2018}.

Natural models build on the view of CwFs as consisting of a presheaf of types (over the category of
contexts), coupled with a presheaf of terms (over the elements of those types). We find this
relatively recent technology helpful for two reasons. First, it concisely encodes the many
naturality conditions normally required of a CwF. Second, it aids in uncovering the implicit
universal properties of type-theoretic connectives, which are not quite so evident in the usual
GAT-like formulation of CwFs.

In Section~\ref{sec:semantics:standard} we demonstrate how the basic notions of context, type, term, and
context extension in \MTT{} can be presented in terms of natural models. Then, in
Section~\ref{sec:semantics:connectives} we show how to interpret the various connectives---including the
modality---in the language of natural models; this discussion concludes with a concise definition of
a model of \MTT{} in Section~\ref{sec:semantics:model}. Following that, in
Section~\ref{sec:semantics:morphism}
we briefly discuss a strict notion of morphism of models.

\subsection{Contexts, Types, and Terms}
\label{sec:semantics:standard}

\subsubsection{Contexts}
\label{sec:semantics:standard-context-categories}

First, we observe that a model of our type theory must contain a set of contexts at each mode $m \in
\Mode$. Equipped with the substitutions at the same mode, which can be composed associatively and
have the identity substitution as a unit, these sets are readily seen to form a category---the \emph{context category} at $m \in \Mode$---for which we write $\InterpMode{m}$.

Moreover, recall that for $\IsCx{\Gamma}$ and $\mu : n \to m$ we have a context
$\IsCx{\LockCxV{\Gamma}}<n>$, and that this construction extends to substitutions in a functorial
fashion. Hence, we will require for each modality $\mu : n \to m$ a functor
\[
  \InterpModal{\mu} : \InterpMode{m} \to \InterpMode{n}
\]
Similarly, each $\alpha : \mu \To \nu$ induces a natural transformation. Accordingly, a model should
come with a natural transformation
\[
  \InterpKey{\alpha} : \InterpModal{\nu} \To \InterpModal{\mu}
\]
The equalities of the GAT require that the assignments $\mu \mapsto
\Lock_\mu$ and $\alpha \mapsto \Key{\alpha}{}$ be strictly 2-functorial. Thus, this part of the
model can be succinctly summarized as follows.

\begin{defi}
  \label{def:semantics:context-structure}
  A \emph{modal context structure} for a mode theory $\Mode$ is a (strict) 2-functor
  \[
    \Interp{-}  : \Coop{\Mode} \to \CAT_1
  \]
  where $\Coop{\Mode}$ is the 2-category $\Mode$ with the direction of \emph{both} 1-cells and
  2-cells reversed, and $\CAT_1$ is the full subcategory of (large) categories with a terminal
  object.
\end{defi}

This double contravariance may seem peculiar at first sight. Recall that the 2-category $\Mode$
specifies the behaviour of the modal types $\Modify{-}$, which are supposed to have a
right-adjoint-like behaviour, with the corresponding left-adjoint-like operators being the lock
functors $\LockCx{-}$. Being left-adjoint-like, the interpretation $\InterpModal{-}$ of each lock
will behave with variance opposite to the specification of $\Mode$. Of course, this is merely an
analogy, as these constructions are not truly adjoint.

\subsubsection{Types and Context Extension}
\label{sec:semantics:standard-cwf}

The following definition plays a central role.

\begin{defi}[Representable natural transformation]
  \label{def:semantics:representable-nat}
  Let $\mathbb{C}$ be a small category, and let $P, Q : \PSH{\mathbb{C}}$ be presheaves on
  $\mathbb{C}$. A natural transformation $\alpha : P \Rightarrow Q$ is \emph{representable} just if
  for every $\Gamma : \mathbb{C}$ and $x : \Yo{\Gamma} \Rightarrow P$ there exists a $y :
  \Yo{\Delta} \Rightarrow Q$ and a morphism $\gamma : \Delta \rightarrow \Gamma$ in $\mathbb{C}$
  such that there is a pullback square
  \[
    \begin{tikzpicture}[node distance=2.5cm, on grid]
      \node[pullback] (Delta) {$\Yo{\Delta}$};
      \node (El) [right = 3cm of Delta] {$Q$};
      \node (Gamma) [below = 2cm of Delta] {$\Yo{\Gamma}$};
      \node (UU) [below = 2cm of El] {$P$};
      \path[->] (Delta) edge node [left] {$\Yo{\gamma}$} (Gamma);
      \path[->] (Delta) edge node [above] {$y$} (El);
      \path[->] (Gamma) edge node [below] {$x$} (UU);
      \path[->] (El) edge node [right] {$\alpha$} (UU);
    \end{tikzpicture}
  \]
\end{defi}
This enables a very succinct definition of a model of type theory \cite{awodey:2018}.
\begin{defi}
  Let $\mathbb{C}$ be a small category with a terminal object $\mathbf{1}$, and let $\CTm{}, \CTy{}
  : \PSH{\mathbb{C}}$. A \emph{natural model of type theory} is a representable natural
  transformation $\El{} : \CTm{} \Rightarrow \CTy{}$.
\end{defi}
It is shown in \emph{op. cit.} that this corresponds to the usual notion of CwF: the
representability of $\El{} : \CTm{} \Rightarrow \CTy{}$ is a clever way to encode context
extension and comprehension in a manner that automatically ensures naturality with respect to
substitution: see also \cite{fiore:2012}. Moreover, one can use this economy to write down very
concise interpretations of type formers. Our objective here is to adapt this to modes and
modalities.

To begin, given a mode $m \in \Mode$ we define two presheaves on the context category
$\InterpMode{m}$:
\begin{align*}
  \CTy{m}(\Gamma) &\defeq \Ty{m}[1](\Gamma) &
  \CTm{m}(\Gamma) &\defeq
   \{ \parens{A, M} \mid A \in \Ty{m}[1](\Gamma), M \in \Tm{m}(\Gamma, A) \}
\end{align*}
The first one maps a context at mode $m \in \Mode$ to the set of large types over it. The second one
maps a context to the set of pointed types, i.e. to the set of pairs consisting of a type and a term
of that type. The presheaf action is given by substitution. We immediately obtain a natural
transformation $\El{m} : \CTm{m} \Rightarrow \CTy{m}$: at each context $\Gamma$, $\El{m, \Gamma}$
projects a pair $(A, M)$ to the underlying type $A$. As a result, the fibres of $\El{m}$ are the
terms of a given type.

Context extension postulates that for any object $\Gamma : \InterpMode{m}$, modality $\mu
\in \Hom{n}{m}$, and large type $A \in \Ty{n}[1](\InterpModal{\mu}\Gamma)$ there exists an object
$\ECx{\Gamma}{A} : \InterpMode{m}$ along with a morphism and a term
\begin{align*}
  \CWk :\Hom[\InterpMode{m}]{\ECx{\Gamma}{A}}{\Gamma} &&
  \CVar \in \Tm{n}(\InterpModal{\mu}(\ECx{\Gamma}{A}),\Sb{A}{\InterpModal{\mu}(\CWk)})
\end{align*}
The object $\ECx{\Gamma}{A}$ is universal with respect to $\CWk$ and $\CVar$: for any $\gamma \in
\Hom[\InterpMode{m}]{\Delta}{\Gamma}$ and term $M \in
\Tm{n}(\InterpModal{\mu}\Delta,\Sb{A}{\LockSb{\gamma}})$ there is a \emph{unique} $\ESb{\gamma}{M} :
\Delta \to \ECx{\Gamma}{A}$ such that
\begin{alignat}{2}
  \label{eq:p-proj}
  \CWk \circ (\ESb{\gamma}{M}) &= \gamma &&: \Delta \to \Gamma \\
  \label{eq:q-proj}
  \Sb{\CVar}{\LockSb{(\ESb{\gamma}{M})}}
    &= M &&: \Tm{n}(\InterpModal{\mu}(\Gamma),\Sb{A}{\LockSb{\gamma}})
\end{alignat}
As usual, \eqref{eq:q-proj} is only well-typed because of \eqref{eq:p-proj}. The only difference to the usual context extension of CwFs is that $A$ and $\Gamma$ are in different modes.

This can be encoded in the style of natural models as follows. We write $\YoEm{-}$ for the Yoneda
isomorphism. Given $\mu : \Hom[\Mode]{n}{m}$, context $\Gamma : \InterpMode{m}$, and a type $A :
\Ty{n}[1]{\parens{\InterpModal{\mu}(\Gamma)}}$, there is a chosen context $\Gamma'$ (cf.
$\ECx{\Gamma}{A}$), a chosen morphism $\CWk : \Gamma' \to \Gamma$, and a chosen morphism
$\YoEm{\CVar} : \Yo{\InterpModal{\mu}\Gamma'} \to \CTm{n}$ that make the following square commute:
\[
  \begin{tikzpicture}[node distance=2.5cm, on grid]
    \node (Gammax) {$\Yo{\InterpModal{\mu}{\Gamma'}}$};
    \node (El) [right = 3cm of Gammax] {$\CTm{n}$};
    \node (Gamma) [below = 1.5cm of Gammax] {$\Yo{\InterpModal{\mu}{\Gamma}}$};
    \node (UU) [below = 1.5cm of El] {$\CTy{n}$};
    \path[->] (Gammax) edge node [left] {$\Yo{\InterpModal{\mu}\CWk}$} (Gamma);
    \path[->] (Gammax) edge node [above] {$\YoEm{\CVar}$} (El);
    \path[->] (Gamma) edge node [below] {$\YoEm{A}$} (UU);
    \path[->] (El) edge node [right] {$\El{n}$} (UU);
  \end{tikzpicture}
\]
We have surreptitiously `decoded' the top arrow into a term $\CVar \in
\Tm{n}(\InterpModal{\mu}(\Gamma'),\Sb{A}{\InterpModal{\mu}\CWk})$.

The universality of these objects is expressed by asking that for a given $\Delta : \InterpMode{m}$,
$\gamma : \Hom[\InterpMode{m}]{\Delta}{\Gamma}$, and $\YoEm{M} : \Yo{\InterpModal{\mu}\Delta}
\Rightarrow \CTm{n}$, there must be a unique morphism $\gamma' : \Delta \to \Gamma'$ (which stands
for $\ESb{\gamma}{M}$) such that the following square commutes:
\[
  \begin{tikzpicture}[node distance=2.5cm, on grid]
    \node (Delta) {$\Yo{\InterpModal{\mu}(\Delta)}$};
    \node (Gammax) [below right = 2cm and 2cm of Delta] {$\Yo{\InterpModal{\mu}{\Gamma'}}$};
    \node (El) [right = 4cm of Gammax] {$\CTm{n}$};
    \node (Gamma) [below = 2cm of Gammax] {$\Yo{\InterpModal{\mu}\Gamma}$};
    \node (UU) [below = 2cm of El] {$\CTy{n}$};
    \path[->, densely dashed] (Delta) edge node[upright desc] {$\Yo{\InterpModal{\mu}{\gamma'}}$} (Gammax);
    \path[->, bend right] (Delta) edge node[upright desc] {$\Yo{\InterpModal{\mu}\gamma}$} (Gamma);
    \path[->, bend left]  (Delta) edge node[upright desc] {$\YoEm{M}$} (El);
    \path[->] (Gammax) edge node [upright desc] {$\Yo{\InterpModal{\mu}\CWk}$} (Gamma);
    \path[->] (Gammax) edge node [above] {$\YoEm{\CVar}$} (El);
    \path[->] (Gamma) edge node [below] {$\YoEm{A}$} (UU);
    \path[->] (El) edge node [right] {$\El{n}$} (UU);
  \end{tikzpicture}
\]
This diagram is not a pullback, but we can make it into one. Recall that for any functor $f :
\mathcal{C} \to \mathcal{D}$ we can define the precomposition functor $\Pre{f} : \PSH{\mathcal{D}}
\to \PSH{\mathcal{C}}$ by
\[
  \Pre{f}(P) \defeq \Op{\mathcal{C}} \xrightarrow{\Op{f}} \Op{\DD} \xrightarrow{P} \SET
\]
Then, for any $c : \mathcal{C}$ and $Q : \PSH{\mathcal{D}}$ we can use the Yoneda lemma to establish
a series of natural isomorphisms
\[
  \Hom[\PSH{\mathcal{D}}]{\Yo{f(c)}}{Q}
    \cong
  Q\parens*{f(c)}
    =
  \Pre{f}{Q}(c)
    \cong
  \Hom[\PSH{\mathcal{C}}]{\Yo{c}}{\Pre{f}{Q}}
\]
We can then \emph{transpose} the diagram in order to obtain
\begin{equation}
  \begin{tikzpicture}[node distance=2.5cm, on grid, baseline=(current  bounding  box.center)]
    \node (Delta) {$\Yo{\Delta}$};
    \node (Gammax) [below right = 2cm and 2cm of Delta] {$\Yo{\Gamma'}$};
    \node (El) [right = 4cm of Gammax] {$\Pre{\InterpModal{\mu}}{\CTm{n}}$};
    \node (Gamma) [below = of Gammax] {$\Yo{\Gamma}$};
    \node (UU) [below = of El] {$\Pre{\InterpModal{\mu}}{\CTy{n}}$};
    \path[->, densely dashed] (Delta) edge node[upright desc] {$\gamma'$} (Gammax);
    \path[->, bend right] (Delta) edge node[upright desc] {$\gamma$} (Gamma);
    \path[->, bend left]  (Delta) edge node[upright desc] {$\YoEm{M}$} (El);
    \path[->] (Gammax) edge node [left] {$\Yo{\CWk}$} (Gamma);
    \path[->] (Gammax) edge node [above] {$\YoEm{\CVar}$} (El);
    \path[->] (Gamma) edge node [below] {$\YoEm{A}$} (UU);
    \path[->] (El) edge node [right] {$\Pre{\InterpModal{\mu}}{\El{n}}$} (UU);
  \end{tikzpicture}
  \label{diag:modal-natural-model}
\end{equation}
where $\gamma' : \Delta \to \Gamma'$ is the unique arrow that makes the diagram commute.
The requirement that this diagram be a pullback leads us to the following definition.
\begin{defi}
  \label{def:semantics:modal-natural-model}
  A \emph{modal natural model} on a context structure $\Interp{-}  : \Coop{\Mode} \to \CAT_1$
  consists of a family of natural transformations of presheaves
  \[
      \parens*{\El{m} : \CTm{m} \Rightarrow \CTy{m}}_{m \in \Mode}
  \]
  where $\CTm{m}, \CTy{m} : \PSH{\InterpMode{m}}$ such that for every $\mu : \Hom[\Mode]{m}{n}$ the
  natural transformation
  \[
    \Pre{\InterpModal{\mu}}{\El{n}} :
      \Pre{\InterpModal{\mu}}{\CTm{n}} \To \Pre{\InterpModal{\mu}}{\CTy{n}}
  \]
  is a natural model.
\end{defi}
We will write $\ECx{\Gamma}{A}$ for the object $\Gamma'$ that makes \eqref{diag:modal-natural-model}
a pullback, as we do in the type theory.


\subsection{Connectives}
\label{sec:semantics:connectives}

 We shall only discuss the key cases of $\Pi$ types, modal types, Boolean types, and universes. The
 interpretation of the other connectives largely follows the style of \cite{awodey:2018}. More
 details can be found in the tech report \cite{gratzer:multimodal:2020}.

\subsubsection{$\Pi$ Structure}
\label{sec:semantics:standard-pi}

Even though \MTT{} $\Pi$ types are close to traditional $\Pi$ types they are not quite the same, as
they involve a modality in the domain. Thus, we need to construct an appropriate variation of the
interpretation given by \cite{awodey:2018}. To begin, we need some way to represent the
\emph{binding} of an additional assumption. This is achieved through the use of \emph{polynomial
endofunctors}. 
Given a `display map' $\ell : E \to B$ we define a polynomial endofunctor $\Poly{\ell : E \to B} :
\PSH{\InterpMode{m}} \to \PSH{\InterpMode{m}}$ by\footnote{This is given in the internal language,
but may also be written purely categorically, as is done in \emph{op. cit.}}
\[
  \Poly{\ell : E \to B}{A} \defeq \sum_{b : B} A^{\ReIdx{\ell}{b}}
\]
When specialized to the `modalized' natural model $\ell \defeq \Pre{\InterpModal{\mu}}(\tau_n) :
\Pre{\InterpModal{\mu}}{\CTm{n}} \To \Pre{\InterpModal{\mu}}{\CTy{n}}$, this functor has a useful
property: morphisms $\Yo{\Gamma} \To \Poly{\Pre{\InterpModal{\mu}}{\El{n}}}{\CTy{m}}$ are in
bijection with tuples
\begin{equation}
  \parens{A \in \CTy{n}(\InterpModal{\mu}(\Gamma)), B \in \CTy{m}(\ECx{\Gamma}{A})}
  \label{eq:semantics:binding-pair}
\end{equation}
This enables the representation of a pair of types $\IsTy[\LockCx{\Gamma}]{A}[1]<n>$ and
$\IsTy[\ECx{\Gamma}{A}]{B}[1]<m>$---i.e. the premises of $\Pi$ formation---as a single morphism
$\Yo{\Gamma} \To \Poly{\Pre{\InterpModal{\mu}}{\El{n}}}{\CTy{n}}$. A similar observation applies to
the presheaf of terms $\CTm{m}$. See \cite[Lemma 5]{awodey:2018} for a detailed proof of this
property.

A model is equipped with a $\prod$-structure if for $\mu : \Hom[\Mode]{n}{m}$ we have a pullback
\[
  \begin{tikzpicture}[node distance=1.75cm, on grid]
    \node[pullback] (PiTm) {$\Poly{\Pre{\InterpModal{\mu}}{\El{n}}}{\CTm{m}}$};
    \node (Tm) [right = 4cm of PiTm] {$\CTm{m}$};
    \node (PiTy) [below = of PiTm] {$\Poly{\Pre{\InterpModal{\mu}}{\El{n}}}{\CTy{m}}$};
    \node (Ty) [below = of Tm] {$\CTy{m}$};
    \path[->] (PiTm) edge node [left] {} (PiTy);
    \path[->] (PiTm) edge node [above] {$\CLam$} (Tm);
    \path[->] (PiTy) edge node [below] {$\CPi$} (Ty);
    \path[->] (Tm) edge node [right] {$\El{m}$} (Ty);
  \end{tikzpicture}
\]
The lower morphism $\prod$ models the formation rule: the premises of the rule constitute a pair of
the form \eqref{eq:semantics:binding-pair}. We may thus combine them into an arrow $\Yo{\Gamma} \To
\Poly{\Pre{\InterpModal{\mu}}{\El{n}}}{\CTy{n}}$, and then postcompose $\prod$ to obtain a morphism
$\Yo{\Gamma} \To \CTy{m}$, i.e. a type at mode $m$ in context $\IsCx{\Gamma}$. In a similar fashion,
the top morphism $\CLam$ models the introduction rule. The $\beta$ and $\eta$ laws then follow from
the pullback property of the square: see \cite{awodey:2018}.

\subsubsection{Modal Structure}
\label{sec:semantics:standard-modify}

The interpretation of the modal types is a bit more involved. Intuitively, the reason is that
$\Modify{A}$ behaves like a \emph{positive type former}, i.e. one with a `let-style'
pattern-matching eliminator, and no $\eta$-rule. These features render its behaviour closer to that
of intensional identity types.

First, for each $\mu : \Hom[\Mode]{n}{m}$ the formation and introduction rules for $\Modify{-}$ are
given by a commuting square
\begin{equation}
  \begin{tikzpicture}[node distance=2.5cm, on grid, baseline=(current  bounding  box.center)]
    \node (ModTm) {$\Pre{\InterpModal{\mu}}{\CTm{n}}$};
    \node (Tm) [right = 4cm of ModTm] {$\CTm{m}$};
    \node (ModTy) [below = 1.75cm of ModTm] {$\Pre{\InterpModal{\mu}}{\CTy{n}}$};
    \node (Ty) [below = 1.75cm of Tm] {$\CTy{m}$};
    \path[->] (ModTm) edge node [left] {$\Pre{\InterpModal{\mu}}{\El{n}}$} (ModTy);
    \path[->] (ModTm) edge node [above] {$\CMkBox$} (Tm);
    \path[->] (ModTy) edge node [below] {$\CModify$} (Ty);
    \path[->] (Tm) edge node [right] {$\El{m}$} (Ty);
  \end{tikzpicture}
  \label{diag:form-intro-modal}.
\end{equation}
By Yoneda, every type $\IsTy[\LockCx{\Gamma}]{A}[1]<n>$ can be seen as a morphism $\Yo{\Gamma} \To
\Pre{\InterpModal{\mu}}{\El{n}}$. Postcomposition with $\CModify$ gives a morphism $\Yo{\Gamma} \To
\CTy{m}$, which constitutes the interpretation of the type $\IsTy[\Gamma]{\Modify{A}}[1]<m>$.
$\CMkBox$ interprets the introduction rule in a similar fashion. Nevertheless, asking that this
square be a pullback is stronger than the elimination rule. In Section~\ref{sec:dra} we
shall see that states that $\CModify$ is a \emph{dependent right adjoint}.

Instead, we will model our elimination rule by a \emph{lifting structure}. We phrase this definition
in the internal language of the presheaf topos $\PSH{\InterpMode{m}}$, \ie~extensional type
theory.\footnote{This is derived from unpublished work by Jon Sterling, Daniel Gratzer, Carlo
Angiuli, and Lars Birkedal.} This has a serious technical advantage: as the definition is given in
an empty context, the given lifts are automatically natural.

\begin{defi}[Left lifting structure]
  \label{def:semantics:left-lifting-structure}
  Given $\vdash A, I, B\ \mathsf{type}$, a family $b : B \vdash E[b]\ \mathsf{type}$ and a section
  $a : A \vdash i[a] : I$, we define the type $\vdash i[-]\pitchfork E[-]\ \mathsf{type}$ of
  \emph{left lifting structures} for $i$ with respect to $E$ to be
  \[
    \textstyle
    i[-] \pitchfork E[-] \triangleq
    \prod_{C: I \to B}
    \prod_{c: \prod_{a:A} E[C(i[a])]}
    \braces*{
      j:\prod_{p:I} E[C(p)]
      \mid
      \forall a:A.\
      j(i[a])=c(a)
    }
  \]
\end{defi}
Informally, left lifting structures provide \emph{diagonal fillers} for the diagram
\[
  \begin{tikzpicture}[on grid, node distance = 1.75cm]
    \node (A) {$A$};
    \node (I) [below = of A] {$I$};
    \node (E) [right = 4.5cm of A] {$\sum_{b : B} E[b]$};
    \node (B) [below = of E] {$B$};
    \path[->] (A) edge node [above] {$\angles{C(i[-]), c}$} (E);
    \path[->] (A) edge node [left] {$i[-]$} (I);
    \path[->] (E) edge node [right] {$\pi_1$} (B);
    \path[->] (I) edge node [below] {$C$} (B);
    \path[->, densely dotted] (I) edge node [upright desc] {$j$} (E);
  \end{tikzpicture}
\]
Intuitively, $C : I \to B$ is the \emph{motive} of an elimination: we would like to prove $E[C(p)]$
for all $p : I$. At the same time, $c : \prod_{a : A} E[C(i[a])]$ is a given section that specifies
the desired computational behaviour of this elimination at the `special case' $A$. The left lifting
structure then provides a section $j$ of $E[-]$ defined on all of $I$. This section is above $C$,
and extends $c$.  Note that these fillers are not necessarily unique.  Moreover, they are
automatically \emph{natural}: as all the types involved in this definition are closed, we are at
liberty to weaken the context.

This style of lifting structure is an essential ingredient in recent work on models of intensional
identity types. First, they play an important r\^{ole} in natural models: \cite[Lemma
19]{awodey:2018} shows that they precisely correspond to enriched left lifting properties in the
sense of categorical homotopy theory \cite[\S 13]{riehl:2014}. In fact, the above definition
given above is a word-for-word restatement in the internal language. Second, such lifting structures
are also central devices in internal presentations of models of cubical type theory, in particular
the recent work of \cite{orton:2018}.

We can now approach this in a manner similar to intensional identity types in \emph{op. cit}. Recall
that the elimination rule for $\Modify[\nu]{A}$ is
\[
  \inferrule{
    \nu : \Hom[\Mode]{o}{n} \\
    \mu : \Hom[\Mode]{n}{m} \\
    \IsCx{\Gamma} \\
    \IsTy[\LockCx{\LockCx{\Gamma}}<\nu>]{A}[1]<o> \\
    \IsTm[\LockCx{\Gamma}]{M_0}{\Modify[\nu]{A}}<n> \\
    \IsTy[\ECx{\Gamma}{\Modify[\nu]{A}}]{B}[1] \\
    \IsTm[\ECx{\Gamma}{A}<\mu \circ \nu>]{M_1}{\Sb{B}{\ESb{\Wk}{\MkBox[\nu]{\Var{0}}}}}
  }{
    \IsTm{\Open{M_0}{M_1}<\nu>[\mu]}{\Sb{B}{\ESb{\ISb}{M_0}}}
  }
\]
First, we must remove the `implicit cut' with $M_0$. We construct the substitution
\[
  \IsSb[ \ECx{\ECx{\Gamma}{\Modify[\nu]{A}}}{\Sb{A}{\LockSb{\Wk}<\mu \circ \nu>}}<\mu \circ \nu> ]
  {\sigma \defeq \ESb{\Wk[2]}{\Var{0}}}{\ECx{\Gamma}{A}<\mu \circ \nu>}
\]
It then suffices to construct the elimination rule
\[
  \inferrule{
    \nu : \Hom[\Mode]{o}{n} \\
    \mu : \Hom[\Mode]{n}{m} \\
    \IsCx{\Gamma} \\
    \IsTy[\LockCx{\LockCx{\Gamma}}<\nu>]{A}[1]<o> \\
    \IsTy[\ECx{\Gamma}{\Modify[\nu]{A}}]{B}[1] \\
    \IsTm[\ECx{\Gamma}{A}<\mu \circ \nu>]{M_1}{\Sb{B}{\ESb{\Wk}{\MkBox[\nu]{\Var{0}}}}}
  }{
    \IsTm[\ECx{\Gamma}{\Modify[\nu]{A}}]{\Open{\Var{0}}{\Sb{M_1}{\sigma}}<\nu>[\mu]}{B}
  }
\]
because we can calculate that
\[
  \EqTm[\Gamma]{
    \Sb{\parens{\Open{\Var{0}}{\Sb{M_1}{\sigma}}<\nu>[\mu]}}{\ESb{\ISb}{M_0}}
  }{
    \Open{M_0}{M_1}<\nu>[\mu]
  }{\Sb{B}{\ESb{\ISb}{M_0}}}
\]
We can rephrase this as the existence of a diagonal filler in the diagram
\[
  \begin{tikzpicture}[on grid, node distance = 2.5cm]
    \node (Base) {$\Yo{\ECx{\Gamma}{A}<\mu \circ \nu>}$};
    \node (Motive) [below = 2cm of Base] {$\Yo{\ECx{\Gamma}{\Modify[\nu]{A}}}$};
    \node (Tm) [right = 6cm of Base] {$\CTm{m}$};
    \node (Ty) [below = 2cm of Tm] {$\CTy{m}$};
    \path[->] (Tm) edge node [right] {$\El{m}$} (Ty);
    \path[->] (Base) edge node [left] {$\Yo{\ESb{\Wk}{\MkBox[\nu]{\Var{0}}}}$} (Motive);
    \path[->] (Motive) edge node [below] {$\YoEm{B}$} (Ty);
    \path[->] (Base) edge node [above] {$\YoEm{\Sb{M_1}{\sigma}}$} (Tm);
    \path[->] (Motive) edge node [upright desc] {$\YoEm{\Open{\Var{0}}{\Sb{M_1}{\sigma}}<\nu>[\mu]}$}(Tm);
  \end{tikzpicture}
\]

We can use a left lifting structure on a carefully chosen slice category to obtain such diagonal
fillers. The internal language approach still applies because of the well-known lemma stating that
the slice of a presheaf topos is also a presheaf topos, but over the corresponding category of
elements. In symbols, for any $P : \PSH{\CC}$ we have an equivalence $\SLICE{\PSH{\CC}}{P} \Equiv
\PSH{\int_\CC P}$: see \cite[III Ex. 8]{maclane-moerdijk:1992}

First, given $\nu : \Hom[\Mode]{o}{n}$ we construct the following pullback:
\[
  \begin{tikzpicture}[node distance=2.5cm, on grid]
    \node (ModTm) {$\Pre{\InterpModal{\nu}}{\CTm{o}}$};
    \node[pullback] (M) [below right = 1cm and 2cm of ModTm] {$M$};
    \node (Tm) [right = 4cm of M] {$\CTm{n}$};
    \node (ModTy) [below = 1.75cm of M] {$\Pre{\InterpModal{\nu}}{\CTy{o}}$};
    \node (Ty) [below = 1.75cm of Tm] {$\CTy{n}$};
    \path[->, densely dashed] (ModTm) edge node [upright desc] {$m$} (M);
    \path[->] (M) edge node {} (Tm);
    \path[->] (M) edge node[left] {$h$} (ModTy);
    \path[->] (ModTm) edge[bend right] node [left] {$\Pre{\InterpModal{\nu}}{\El{o}}$} (ModTy);
    \path[->] (ModTm) edge[bend left] node [above] {$\CMkBox[\nu]$} (Tm);
    \path[->] (ModTy) edge node [below] {$\CModify[\nu]$} (Ty);
    \path[->] (Tm) edge node [right] {$\El{n}$} (Ty);
  \end{tikzpicture}
\]
The outer commuting square is that given by the formation and introduction for $\Modify[\nu]{-}$, as
in \eqref{diag:form-intro-modal}. Intuitively, $M$ is a `generic $\nu$-modal terms object' that
consists of terms $\IsTm[\Gamma]{M}{\Modify[\nu]{A}}<n>$, where
$\IsTy[\LockCx{\Gamma}<\nu>]{A}[1]<o>$. We know that $\Pre{\InterpModal{\mu}}$ has a left adjoint,
so it preserves pullbacks. Applying it to this diagram yields
\begin{equation}
  \begin{tikzpicture}[node distance=2.5cm, on grid, baseline=(current  bounding  box.center)]
    \node (ModTm) {$\Pre{\InterpModal{\mu \circ \nu}}{\CTm{o}}$};
    \node[pullback] (M) [below right = 1cm and 3.25cm of ModTm] {$\Pre{\InterpModal{\mu}}{M}$};
    \node (Tm) [right = 4.5cm of M] {$\Pre{\InterpModal{\mu}}{\CTm{n}}$};
    \node (ModTy) [below = of M] {$\Pre{\InterpModal{\mu \circ \nu}}{\CTy{o}}$};
    \node (Ty) [below = of Tm] {$\Pre{\InterpModal{\mu}}{\CTy{n}}$};
    \path[->, densely dashed] (ModTm) edge node [above] {$\Pre{\InterpModal{\mu}}{m}$} (M);
    \path[->] (M) edge node {} (Tm);
    \path[->] (M) edge node[left] {$\Pre{\InterpModal{\mu}}{h}$} (ModTy);
    \path[->] (ModTm) edge[bend right] node [left] {$\Pre{\InterpModal{\mu \circ \nu}}{\El{o}}$} (ModTy);
    \path[->] (ModTm) edge[bend left] node [above] {$\Pre{\InterpModal{\mu}}{\CMkBox[\nu]}$} (Tm);
    \path[->] (ModTy) edge node [below] {$\Pre{\InterpModal{\mu}}{\CModify[\nu]}$} (Ty);
    \path[->] (Tm) edge node [right] {$\Pre{\InterpModal{\mu}}{\El{n}}$} (Ty);
  \end{tikzpicture}
  \label{diag:pullback-generic-modal}
\end{equation}
We have also used the fact that $\Pre{(-)}$ is functorial to contract the two locks into
one. Moreover, we get that the unique mediating morphism is indeed $\Pre{\InterpModal{\mu}}{m}$.

From this point onwards we will also work in the slice $\SLICE{\PSH{\InterpMode{m}}}{Z}$, where $Z
\defeq \Pre{\InterpModal{\mu \circ \nu}}{\CTy{o}}$. In order to model the elimination rule
we will ask for a left lifting structure \emph{in the slice category}, of type
\begin{equation}
  \vdash \COpen[\nu]^\mu :
    \Pre{\InterpModal{\mu}}{m}
      \pitchfork
    \Pre*{Z}{\El{m}} 
    \label{eq:semantics:modal-elimination-structure}
\end{equation}
where both of these are considered as morphisms in the slice $\SLICE{\PSH{\InterpMode{m}}}{Z}$,
respectively of type
\begin{align*}
  \Pre{\InterpModal{\mu}}{m} :\
    &\Pre{\InterpModal{\mu \circ \nu}}{\El{o}} \to \Pre{\InterpModal{\mu}}{h} \\
  \Pre*{Z}{\El{m}} :\
    &\Pre*{Z}{\CTm{m}} \to \Pre*{Z}{\CTy{m}}
\end{align*}

Following \cite{awodey:2018} we may calculate that this models the rule. We suppose its
premises, and construct the diagram of Figure~\ref{fig:semantics:modal-elimination}.
\begin{figure}
  \begin{tikzpicture}[node distance=2.5cm, on grid]
    \node (Cx-munu-A)  {$\Yo{\ECx{\Gamma}{A}<\mu \circ \nu>}$};
    \node[pullback] (Cx-mu-nu-A) [below = 4cm of Cx-munu-A] {$\Yo{\ECx{\Gamma}{\Modify[\nu]{A}}}$};
    \node (Tm-munu-o) [right = 5cm of Cx-munu-A] {$\Pre{\InterpModal{\mu \circ \nu}}{\CTm{o}}$};
    \node[pullback] (M) [below = 4cm of Tm-munu-o] {$\Pre{\InterpModal{\mu}}{M}$};
    \node (Cx) [below = 3cm of Cx-mu-nu-A] {$\Yo{\Gamma}$};
    \node (Ty-munu-o) [below = 3cm of M] {$\Pre{\InterpModal{\mu \circ \nu}}{\CTy{o}}$};
    \node (Tm-mu-n) [right = 4cm of M] {$\Pre{\InterpModal{\mu}}{\CTm{n}}$};
    \node (Ty-mu-n) [below = 3cm of Tm-mu-n] {$\Pre{\InterpModal{\mu}}{\CTy{n}}$};
    \path[->, densely dashed] (Cx-munu-A) edge node [upright desc]
      {$\Yo{\ESb{\CWk}{\CMkBox[\nu](\CVar)}}$} (Cx-mu-nu-A);
    \path[->] (Cx-mu-nu-A) edge node[upright desc] {$\Yo{\CWk}$} (Cx);
    \path[->] (Cx-munu-A) edge node[upright desc] {$\YoEm{\CVar}$} (Tm-munu-o);
    \path[->] (Tm-munu-o) edge node[upright desc] {$\Pre{\InterpModal{\mu}}{m}$} (M);
    \path[->] (Tm-munu-o) edge node[upright desc] {$\Pre{\InterpModal{\mu}}{\CMkBox[\nu]}$} (Tm-mu-n);
    \path[->] (Tm-mu-n) edge node[right] {$\Pre{\InterpModal{\mu}}{\El{n}}$} (Ty-mu-n);
    \path[->] (Cx) edge[bend right] node[below] {$\YoEm{\Modify[\nu]{A}}$} (Ty-mu-n);
    \path[->] (Cx) edge node[upright desc] {$\YoEm{A}$} (Ty-munu-o);
    \path[->] (Ty-munu-o) edge node[below] {$\Pre{\InterpModal{\mu}}{\CModify[\nu]}$} (Ty-mu-n);
    \path[->] (Cx-munu-A) edge[bend left = 60] node[above, near start] {$\YoEm{\MkBox[\nu]{\CVar}}$} (Tm-mu-n);
    \path[->] (M) edge node[upright desc] {$\Pre{\InterpModal{\mu}}{h}$} (Ty-munu-o);
    \path[->] (M) edge (Tm-mu-n);
    \path[->] (Cx-munu-A) edge[bend right = 50] node[left] {$\Yo{\CWk}$} (Cx);
    \path[->, densely dashed] (Cx-mu-nu-A) edge (M);
    \path[->] (Cx-mu-nu-A) edge[bend left = 20] node[near start, upright desc] {$\YoEm{\CVar}$} (Tm-mu-n);
    \path[->] (Tm-munu-o) edge[bend right = 40] node[near start, left]
      {$\Pre{\InterpModal{\mu \circ \nu}}{\El{o}}$} (Ty-munu-o);
  \end{tikzpicture}
  \caption{Modelling the elimination rule}
  \label{fig:semantics:modal-elimination}
\end{figure}
The right (both top and bottom) part of the diagram is just \eqref{diag:pullback-generic-modal}. The
bottom composite is easily seen to correspond to the application of the introduction rule of
$\Modify[\nu]{-}$ to the type $\IsTy[\LockCx{\LockCx{\Gamma}}<\nu>]{A}[1]<o>$, and hence to the type
$\IsTy[\LockCx{\Gamma}]{\Modify[\nu]{A}}[1]<n>$. The outer bottom square is the natural model pullback
square that defines the object $\ECx{\Gamma}{\Modify[\nu]{A}}$, and we thus get a mediating morphism
to $\Pre{\InterpModal{\mu}}{M}$, and that the bottom-left square is also a pullback. The left (both
top and bottom) part of the diagram is the natural model pullback square that defines the object
$\ECx{\Gamma}{A}<\mu \circ \nu>$. We hence get a mediating morphism
$\ESb{\CWk}{\CMkBox[\nu](\CVar)} : \ECx{\Gamma}{A}<\mu \circ \nu> \to
\ECx{\Gamma}{\Modify[\nu]{A}}$. Finally, for the same reasons as the bottom composite, the top
composite is easily seen to correspond to the term $\MkBox[\nu]{\CVar}$.


We write $\sum_Z : \SLICE{\PSH{\InterpMode{m}}}{Z} \to \PSH{\InterpMode{m}}$ for the usual domain
projection functor, so that $\sum_Z \dashv Z^*$. Now, using the usual approach to slice
categories---where the cartesian product $\times_Z$ is the pullback---we see from the diagram that
\begin{equation}
  \begin{aligned}
    \sum_Z\parens{\YoEm{A} \times_Z \Pre{\InterpModal{\mu \circ \nu}}{\CTm{o}}}
      &\cong
    \Yo{\ECx{\Gamma}{A}<\mu \circ \nu>} \\
    \sum_Z\parens{\YoEm{A} \times_Z \Pre{\InterpModal{\mu}}{h}}
      &\cong
    \Yo{\ECx{\Gamma}{\Modify[\nu]{A}}<\mu>} \\
    \sum_Z\parens{\ISb_{\YoEm{A}} \times_Z \Pre{\InterpModal{\mu}}{m}}
      &\cong
    \Yo{\ESb{\CWk}{\CMkBox[\nu](\CVar)}}
  \end{aligned}
  \label{eq:modal-lift-isos}
\end{equation}
Recall that we are trying to find a diagonal filler to the $\PSH{\InterpMode{m}}$ diagram
\begin{equation}
  \begin{tikzpicture}[on grid, node distance = 2.5cm, baseline=(current  bounding  box.center)]
    \node (Base) {$\Yo{\ECx{\Gamma}{A}<\mu \circ \nu>}$};
    \node (Motive) [below = 2cm of Base] {$\Yo{\ECx{\Gamma}{\CModify[\nu](A)}<\mu>}$};
    \node (Tm) [right = 6cm of Base] {$\CTm{m}$};
    \node (Ty) [below = 2cm of Tm] {$\CTy{m}$};
    \path[->] (Tm) edge node [right] {$\El{m}$} (Ty);
    \path[->] (Base) edge node [left] {$\Yo{\ESb{\CWk}{\CMkBox[\nu](\CVar)}}$} (Motive);
    \path[->] (Motive) edge node [below] {$\YoEm{B}$} (Ty);
    \path[->] (Base) edge node [above] {$\YoEm{M_1}$} (Tm);
    \path[->, densely dashed] (Motive) edge (Tm);
  \end{tikzpicture}
  \label{diag:modal-fill}
\end{equation}
We use the adjunction $\sum_Z \Adjoint Z^*$ to transpose this diagram, and we compose with the
isomorphisms \eqref{eq:modal-lift-isos} to obtain the following diagram in
$\SLICE{\PSH{\InterpMode{m}}}{Z}$:
\[
  \begin{tikzpicture}[on grid, node distance = 2.5cm]
    \node (Base) {$\YoEm{A} \times_Z \Pre{\InterpModal{\mu \circ \nu}}{\CTm{o}}$};
    \node (Motive) [below = of Base]
      {$\YoEm{A} \times_Z \Pre{\InterpModal{\mu}}{h}$};
    \node (Tm) [right = 6cm of Base] {$Z^*\parens{\CTm{m}}$};
    \node (Ty) [below = of Tm] {$Z^*\parens{\CTy{m}}$};
    \path[->] (Tm) edge node [right] {$Z^*\parens{\El{m}}$} (Ty);
    \path[->] (Base) edge node [left] {$\ISb \times_Z \Pre{\InterpModal{\mu}}m$} (Motive);
    \path[->] (Motive) edge node [below] {$\Transpose{\YoEm{B}}$} (Ty);
    \path[->] (Base) edge node [above] {$\Transpose{\YoEm{M_1}}$} (Tm);
    \path[->, densely dashed] (Motive) edge node [upright desc] {$\COpen[\nu]^\mu$}(Tm);
  \end{tikzpicture}
\]
We may then use the lifting structure to prove a diagonal filler, and transpose backwards along the
adjunction to obtain a filler for \eqref{diag:modal-fill}. The naturality of all these steps
(composing isomorphisms, transposition, and lifting structure) ensure that the choice is natural.

\subsubsection{Boolean Structure}
\label{sec:semantics:standard-bool}

A boolean structure is defined similarly to the structure for modal types. First, we require
two constants, as well as naturally given diagonal fillers for the appropriate squares:
\[
  \begin{tikzpicture}[node distance=1.5cm, on grid]
    \node (BoolTm) {$1$};
    \node (Tm) [right = 4cm of BoolTm] {$\CTm{m}$};
    \node (BoolTy) [below = 1.5cm of BoolTm] {$1$};
    \node (Ty) [below = 1.5cm of Tm] {$\CTy{m}$};
    \draw[double, double distance=1.5pt] (BoolTm) -- (BoolTy);
    \path[->] (BoolTm) edge[bend left=20] node [above] {$\CTrue$} (Tm);
    \path[->] (BoolTm) edge[bend right=20] node [below] {$\CFalse$} (Tm);
    \path[->] (BoolTy) edge node [below] {$\CBool$} (Ty);
    \path[->] (Tm) edge node [right] {$\El{m}$} (Ty);
  \end{tikzpicture}
  \quad
  \begin{tikzpicture}[node distance=2.5cm, on grid]
    \node (Term) {$1 + 1$};
    \node (Tm) [right = 3cm of Term] {$\CTm{m}$};
    \node (Motive) [below = 2cm of Term] {$\ReIdx{\El{m}}{\CBool}$};
    \node (Ty) [below = 2cm of Tm] {$\CTy{m}$};
    \path[->] (Term) edge node [left] {$[\CTrue,\CFalse]$} (Motive);
    \path[->] (Term) edge node [above] {} (Tm);
    \path[->] (Motive) edge node [below] {} (Ty);
    \path[->] (Tm) edge node [right] {$\El{m}$} (Ty);
    \path[->,densely dotted] (Motive) edge node {} (Tm);
  \end{tikzpicture}
\]
$\ReIdx{\El{m}}{\CBool}$ is the fibre of $\El{m}$ over $\CBool$, and the map $[\CTrue, \CFalse]$ is
obtained as the cotuple of the maps obtained by factoring $\CTrue$ and $\CFalse$ through the fibre.
Requiring a left lifting structure
\[
  \CRec : \brackets*{\CTrue, \CFalse} \pitchfork \El{m}[-]
\]
in the internal language provides enough naturality to yield diagonal fillers for all squares
\[
  \begin{tikzpicture}[node distance=2.5cm, on grid]
    \node (Gamma) {$\Yo{\Gamma} + \Yo{\Gamma}$};
    \node (Tm) [right = 4cm of Gamma] {$\CTm{m}$};
    \node (GammaB) [below = 2cm of Gamma] {$\Yo{\CECx{\Gamma}{\CBool}}$};
    \node (Ty) [below = 2cm of Tm] {$\CTy{m}$};
    \path[->] (Gamma) edge node [left] {$[\ESb{\ISb}{\CTrue},\ESb{\ISb}{\CFalse}]$} (GammaB);
    \path[->] (Gamma) edge node [above] {} (Tm);
    \path[->] (GammaB) edge node [below] {} (Ty);
    \path[->] (Tm) edge node [right] {$\El{m}$} (Ty);
    \path[->,densely dotted] (GammaB) edge node {} (Tm);
  \end{tikzpicture}
\]

\subsubsection{Universe}
\label{sec:semantics:standard-uni}


The universe itself is given by a presheaf $\CSTy{m}$ at each mode $m$. The Coquand-style
isomorphism is implemented by a natural transformation $\CUni : 1 \To \El{m}$, which stands for the
universe type at mode $m$, as well as a natural isomorphism
\[
  \ReIdx{\El{m}}{\CUni} \cong \CSTy{m}
\]
As the pullback of $\El{m}$ along $\CUni$ is $\Yo{\CECx{1}{\CUni}}$, this exactly postulates an
isomorphism between terms of the universe and small types. The coercion from small to large type is
interpreted by a natural transformation $\CLift : \CSTy{m} \Rightarrow \CTy{m}$ that maps each small
type to its associated large type. Moreover, we ask that the formation rules \emph{factor} through
small types: we require a mediating morphism in each of the following diagrams:\footnote{There are
also similar diagrams for $\Sigma$ and intensional identity types.}
  \[
    \begin{tikzpicture}[node distance = 2.5cm, on grid]
      \node (SF) {$\Poly{\Pre{\InterpModal{\mu}}{\El{n}}}{\CSTy{m}}$};
      \node (S) [right = 3.5cm of SF] {$\CSTy{m}$};
      \node (TF) [below = 2cm of SF] {$\Poly{\Pre{\InterpModal{\mu}}{\El{n}}}{\CTy{m}}$};
      \node (T) [below = 2cm of S] {$\CTy{m}$};
      \path[->] (SF) edge node {} (TF);
      \path[->, densely dotted] (SF) edge node {} (S);
      \path[->] (S) edge node[right] {$\CLift$} (T);
      \path[->] (TF) edge node[below] {$\CPi$} (T);
    \end{tikzpicture}
    \quad
    \begin{tikzpicture}[node distance = 2.5cm, on grid]
      \node (SF) {$\Pre{\InterpModal{\mu}}{\CSTy{n}}$};
      \node (S) [right = 2.5cm of SF] {$\CSTy{m}$};
      \node (TF) [below = 2cm of SF] {$\Pre{\InterpModal{\mu}}{\CTy{n}}$};
      \node (T) [below = 2cm of S] {$\CTy{m}$};
      \path[->] (SF) edge node {} (TF);
      \path[->, densely dotted] (SF) edge node {} (S);
      \path[->] (S) edge node[right] {$\CLift$} (T);
      \path[->] (TF) edge node[below] {$\CModify$} (T);
    \end{tikzpicture}
    \quad
    \begin{tikzpicture}[node distance = 2.5cm, on grid]
      \node (SF) {$1$};
      \node (S) [right = 2.5cm of SF] {$\CSTy{m}$};
      \node (TF) [below = 2cm of SF] {$1$};
      \node (T) [below = 2cm of S] {$\CTy{m}$};
      \path[->] (SF) edge node {} (TF);
      \path[->, densely dotted] (SF) edge node {} (S);
      \path[->] (S) edge node[right] {$\CLift$} (T);
      \path[->] (TF) edge node[below] {$\CBool$} (T);
    \end{tikzpicture}
  \]
These factorisations ensure that type formation is closed under small types, and commutation ensures
that the coercions commute with the type formers defintitionally.

\subsubsection{The full definition}
\label{sec:semantics:model}

We have shown how to interpret each rule of \MTT{} through natural models. In fact, every step of
our working is reversible: each contraption we have introduced precisely corresponds to the portion
of the generalized algebraic theory it was used to interpret. In summary, we can make the following
definition.

\begin{defi}
  \label{def:semantics:semantics}
  A model of \MTT{} over $\Mode$ consists of
  \begin{itemize}
    \item a modal context structure for $\Mode$ (as in
    Definition~\ref{def:semantics:context-structure}), and a
    \item a modal natural model on that context structure (as in
    Definition~\ref{def:semantics:modal-natural-model})
  \end{itemize}
  such that the modal natural model supports
  \begin{itemize}
    \item dependent product types 
    \item dependent sum types (at each mode) 
    \item intensional identity types (at each mode) 
    \item modal types 
    \item a boolean type (at each mode), and 
    \item a universe of small types. 
  \end{itemize}
\end{defi}

\subsection{Morphisms of Models}
\label{sec:semantics:morphism}

The generalized algebraic theory (GAT) of \MTT{} also induces a notion of morphism between
models. Traditionally neglected, morphisms are of paramount importance when one produces semantic
proofs of metatheoretic properties, such as canonicity, a proof of which we will present in
Section~\ref{sec:canonicity}.

The last decade has seen much use of relatively weak morphisms of CwFs, \ie~morphisms which preserve
structures only up to isomorphism: see \eg~\cite{clairambault:2014,clouston:dra:2018}. However, our
proof of canonicity will require the strictest notion of CwF morphism, \ie~a GAT homomorphism. Such
morphisms preserve \emph{all} structure on-the-nose, including context extension, and were
introduced alongside CwFs by \cite{dybjer:1996}. As we are using natural models, we will use an
adaptation due to \cite[\S 2.3]{newstead:2018}. We believe that one may construct a biequivalence
or biadjunction between a category based on strict morphisms and one based on weaker ones, as done
by \eg~\cite{uemura:2019}, but we will leave that to future work.

\begin{defi}[Strict morphism of natural models]
  \label{def:semantics:morphism-of-natmod}
  A morphism of natural models $(\CC, \Mor[\El{c}]{\CTm{c}}{\CTy{c}}) \to (\DD,
  \Mor[\El{d}]{\CTm{d}}{\CTy{d}})$ comprises a functor $F : \CC \to \DD$ and a commuting diagram
  \begin{equation}
    \begin{tikzpicture}[on grid, node distance = 2.5cm, baseline=(current  bounding  box.center)]
      \node (TmC) {$\CTm{c}$};
      \node (TmD) [right = 4cm of TmC] {$\Pre{F}{\CTm{d}}$};
      \node (TyC) [below = 2cm of TmC] {$\CTy{c}$};
      \node (TyD) [below = 2cm of TmD] {$\Pre{F}{\CTy{d}}$};
      \path[->] (TmC) edge node [left] {$\El{c}$} (TyC);
      \path[->] (TmC) edge node [above] {$\widetilde{\varphi}$} (TmD);
      \path[->] (TyC) edge node [below] {$\varphi$} (TyD);
      \path[->] (TmD) edge node [right] {$\Pre{F}{\El{d}}$} (TyD);
    \end{tikzpicture}
    \label{diag:semantics:natmod-morphism}
  \end{equation}
  such that $F(\mathbf{1}) = \mathbf{1}$ and the canonical morphism $F(\CECx{\Gamma}{A}) \to
  \CECx{F(\Gamma)}{\varphi(A)}$ is an identity.
\end{defi}
The type $\varphi(A)$ in the last line is defined as follows. Given $\YoEm{A} : \Yo{\Gamma}
\Rightarrow \CTy{c}$ we let
\[
  k \defeq \Yo{\Gamma} \xrightarrow{\YoEm{A}} \CTy{c} \xrightarrow{\varphi} \Pre{F}{\CTm{d}} \\
\]
By Yoneda this induces a natural isomorphism
\begin{equation}
  \Hom[\PSH{\CC}]{\Yo{\Gamma}}{\Pre{F}{\CTy{d}}}
    \cong
  \Pre{F}{\CTy{d}}(\Gamma)
    =
  \CTy{d}(F(\Gamma))
    \cong
  \Hom[\PSH{\DD}]{\Yo{F\Gamma}}{\CTy{d}}
  \label{eq:natmod-morphism-natiso}
\end{equation}
We define $\YoEm{\phi(A)} : \Yo{F\Gamma} \Rightarrow \CTy{d}$ to be $k$ transported under this isomorphism. Also, let
\[
  \YoEm{M} : \Yo{\Gamma} \Rightarrow \CTy{c}
    \quad
      \longmapsto
    \quad
  \YoEm{\widetilde{\varphi}(M)} : \Yo{F(\Gamma)} \Rightarrow \CTy{d}
\]
which maps a term $\Gamma \vdash M : A$ to a term $F\Gamma \vdash \widetilde{\varphi
}(M) : \varphi(A)$ in a similar manner.

Returning to the last condition in the definition, we may now form the diagram
\[
  \begin{tikzpicture}[node distance=2.5cm, on grid]
    \node (FGA) {$\Yo{F(\CECx{\Gamma}{A})}$};
    \node[pullback] (FGFA) [below right = 1cm and 2cm of FGA] {$\Yo{\CECx{F\Gamma}{\varphi(A)}}$};
    \node (El) [right = 4cm of FGFA] {$\CTm{d}$};
    \node (FG) [below = 2cm of FGFA] {$\Yo{F\Gamma}$};
    \node (UU) [below = 2cm of El] {$\CTy{d}$};
    \path[->, densely dashed] (FGA) edge (FGFA);
    \path[->, bend right] (FGA) edge node[upright desc] {$\Yo{F\CWk}$} (FG);
    \path[->, bend left]  (FGA) edge node[upright desc] {$\YoEm{\widetilde{\varphi}(\CVar)}$} (El);
    \path[->] (FGFA) edge node [left] {$\Yo{\CWk}$} (FG);
    \path[->] (FGFA) edge node [above] {$\YoEm{\CVar}$} (El);
    \path[->] (FG) edge node [below] {$\YoEm{\varphi(A)}$} (UU);
    \path[->] (El) edge node [right] {$\El{d}$} (UU);
  \end{tikzpicture}
\]
where the outer square is the diagram composed by pasting together the context extension diagram for
$\CECx{\Gamma}{A}$ and \eqref{diag:semantics:natmod-morphism}, followed by transposing along the natural
isomorphism \eqref{eq:natmod-morphism-natiso}.  We then ask that the unique induced arrow be the
identity.

We can lift these natural transformations to the formation data of the connectives (making
special use of the final equality for the polynomial functors). For instance, we can define a
morphism
\[
  \Mor[(\varphi, \widetilde{\varphi})]{\Poly{\El{c}}{\CTy{c}}}{\Poly{\Pre{F}{\El{d}}}{\Pre{F}{\CTy{d}}}}
  \defeq
    \Poly{\El{c}}{\CTy{c}}
      \xrightarrow{}
    \Poly{\Pre{F}{\CTy{d}}}{\CTy{c}}
      \xrightarrow{\Poly{\Pre{F}{\CTy{d}}}{\varphi}}
    \Poly{\Pre{F}{\El{d}}}{\Pre{F}{\CTy{d}}}
\]
The first component comes from a natural transformation $\Poly{\El{c}}{-} \Rightarrow
\Poly{\Pre{F}{\CTy{d}}}{-}$, which exists because \eqref{diag:semantics:natmod-morphism} not only
commutes, but is a pullback square. That is a nontrivial fact proven laboriously by \cite[\S
2.3.14]{newstead:2018}.  A more conceptual proof is given by \cite[Cor. 3.14]{uemura:2019} in the
language of discrete fibrations.

We then require that all connectives ($\CPi$, $\CSig$, $\CRefl$) strictly commute with these
morphisms. Finally, we can extend this to a model of \MTT{} by requiring not just a functor, but a
natural transformation $\mathcal{C} \To \mathcal{D}$, where $\Mod{C},\Mod{D} : \Coop{\Mode} \to
\CAT$ satisfy the obvious generalizations of the conditions written above. Specifying this formally:
\begin{defi}
  A morphism between two models of \MTT{}, $\Mod{C},\Mod{D}$, is given by a 2-natural
  transformation $\alpha : \Mod{C} \To \Mod{D}$. Moreover, we require a choice of commuting
  squares:
  \[
    \begin{tikzpicture}[on grid, node distance = 2.5cm]
      \node[pullback] (TmC) {$\ModTm{C}{m}$};
      \node (TmD) [right = of TmC] {$\Pre{\alpha_m}{\ModTm{D}{m}}$};
      \node (TyC) [below = 2cm of TmC] {$\ModTy{C}{m}$};
      \node (TyD) [below = 2cm of TmD] {$\Pre{\alpha}{\ModTy{D}{m}}$};
      \path[->] (TmC) edge node [left] {$\ModEl{C}{m}$} (TyC);
      \path[->] (TmC) edge node [above] {$\widetilde{\varphi}_m$} (TmD);
      \path[->] (TyC) edge node [below] {$\varphi_m$} (TyD);
      \path[->] (TmD) edge node [right] {$\Pre{\alpha}{\ModEl{D}{m}}$} (TyD);
    \end{tikzpicture}
  \]
  Moreover, we require that $(\varphi,\widetilde{\varphi})$ strictly commutes with all operations.
  \begin{align*}
    \alpha_m(\ECx{\Gamma}{A}) &= \ECx{\alpha_m(\Gamma)}{\varphi(A)}&&\\
    \CPi \circ (\varphi, \varphi) &= \varphi \circ \CPi & \CLam \circ (\varphi, \widetilde{\varphi}) &= \widetilde{\varphi} \circ \CLam\\
    \CSig \circ (\varphi, \varphi) &= \varphi \circ \CSig & \CPair \circ (\varphi, \widetilde{\varphi}) &= \widetilde{\varphi} \circ \CPair\\
    \CModify \circ \Pre{\InterpModal{\mu}}{\varphi} &= \varphi \circ \CModify & \CMkBox \circ \Pre{\InterpModal{\mu}}{\widetilde{\varphi}} &= \widetilde{\varphi} \circ \CMkBox\\
    && \COpen[\mu]^\nu \circ (\widetilde{\varphi}, \Pre{\InterpModal{\mu}}{\widetilde{\varphi}}) &= \widetilde{\varphi} \circ \COpen[\mu]^\nu\\
    \CBool &= \varphi \circ \CBool & \CTrue = \widetilde{\varphi} \circ \CTrue &\quad \CFalse = \widetilde{\varphi} \circ \CFalse\\
     && \CRec \circ (\widetilde{\varphi},\widetilde{\varphi},\widetilde{\varphi}) &= \widetilde{\varphi} \circ \CRec\\
    \CId \circ (\varphi,\widetilde{\varphi},\widetilde{\varphi}) &= \varphi \circ \CId & \CRefl \circ \widetilde{\varphi} &= \widetilde{\varphi} \circ \CRefl\\
     && \CIdRec \circ (\widetilde{\varphi},\widetilde{\varphi}) &= \widetilde{\varphi} \circ \CIdRec
  \end{align*}
\end{defi}

\begin{rem}[The Initiality of Syntax]
  \label{rem:semantics:initiality}
  Under this definition of homomorphism, we immediately have an initial
  model~\cite{cartmell:1978,kaposi:qiits:2019}. We will \emph{define} this model to be our
  syntax and designate it $(\SynCat{m})_{m \in \Mode}$.
\end{rem}


\section{Canonicity}
\label{sec:canonicity}

Equipped with the generalized algebraic theory of Section~\ref{sec:algebraic-mtt} and its reformulation
through natural models in Section~\ref{sec:semantics}, we are ready to show that the syntax of \MTT{} is
well-behaved. In this section we will sketch the main parts of a proof of \emph{canonicity} for
\MTT{}. This is a basic well-behavedness property which guarantees that terms of ground type, e.g.~$\Bool$, can be normalized. As expected, the proof is independent of the mode theory.

\begin{prop}
  If $\IsTm[{}]{M}{\Bool}$, then either $\EqTm[{}]{M}{\True}{\Bool}$ or
  $\EqTm[{}]{M}{\False}{\Bool}$.
\end{prop}

This kind of result would traditionally be established by producing a rewriting system along with a
lengthy PER model construction. We will instead opt for a proof given by constructing a \emph{glued}
model~\cite{coquand:2018,kaposi:gluing:2019}. The contexts, types, and terms of this model will
contain a syntactic component (resp. a context, type, or term), along with a \emph{proof-relevant
predicate} that is appropriately fibred over it. The base types of this model are carefully chosen
so that a normal form can be extracted from proofs of the predicate. By interpreting a term of
ground type in the glued model we automatically obtain a proof of the predicate, from which we
extract a normal form.

Such proofs involve two steps: defining the glued construction, and proving that it is a model.
While the first step is often straightforward, the second usually involves checking innumerable
equations. In order to shorten the proof sketched here we will make a simplifying assumption (effectively adding an equation to the algebraic syntax): we
will assume that locks preserve the empty context, i.e. that
\[
  \EqCx{\LockCx{\Emp}}{\Emp}
\]
for $\mu : \Hom[\Mode]{m}{n}$. Using the universal property of the terminal context, this implies
\begin{equation}
  \EqSb[\LockSb{\Emp}]{\Key{\alpha}{\Emp}}{\Emp = \ISb}{\LockSb{\Emp}<\nu>}
  \label{eq:canonicity:terminal-key}
\end{equation}
Requiring this equation unfortunately limits our class of models to those where the left adjoint
strictly preserves the terminal product. Despite this simplification the proof remains rather long,
so we will only sketch the construction of the modal natural models. The missing details may be
found in an accompanying technical report.

\begin{rem}
  In what follows we will assume the existence of two Grothendieck universes $\VV' \subset \VV :
  \SET$. We could make do with just one but at the price of some contortions, which are both
  unnecessary and tiresome. We will assume that the sets of contexts, substitutions, types, and
  terms of the syntactic model are $\VV'$-small.
\end{rem}

\subsection{The Glued Model}

We begin by defining the context structure.

\begin{defi}[Glued Contexts]
  A glued context $\Gamma$ at mode $m$ consists of a context $\Syn{\Gamma} \in \Cx{m}$, a predicate
  $\Sem{\Gamma} \in \VV$, and a function
  \[
    \Pred{\Gamma} : \Sem{\Gamma} \to \Sub{m}(\Emp, \Syn{\Gamma})
  \]
\end{defi}

A glued context $\Gamma = (\Syn{\Gamma}, \Sem{\Gamma})$ can be thought of as a proof-relevant
predicate over substitutions into the syntactic context $\Syn{\Gamma}$. An element $x \in
\Sem{\Gamma}$ can be thought of as a proof that the predicate holds of the substitution
$\Pred{\Gamma}(x) : \Emp \to \Syn{\Gamma}$. We will henceforth use the metavariable $\Gamma$ to
range over \emph{glued contexts}, and denote contexts of the syntax by $\Syn{\Gamma}$.

\begin{defi}[Glued Substitutions]
  A glued substitution from $\Delta$ to $\Gamma$ at mode $m$ is a pair of a substitution
  $\Syn{\gamma} \in \Sub{m}(\Syn{\Delta},\Syn{\Gamma})$ and a function $\Sem{\gamma} : \Sem{\Delta}
  \to \Sem{\Gamma}$ such that
  \[
    \forall x \in \Sem{\Delta}.\ \Pred{\Gamma}(\Sem{\gamma}(x)) = \Syn{\gamma} \circ \Pred{\Delta}(x) : \Emp \to \Syn{\Gamma}
  \]
\end{defi}
Glued contexts and glued substitutions form a category, \viz\ the comma category
\[
  \InterpMode{m} \defeq \COMMA{1_{\VV}}{\Sub{m}(\Emp, -)}
\]
which we take as the category of contexts at mode $m$. Next, we define a 2-functor from $\Mode$
sending each $m$ to $\InterpMode{m}$. For each $\mu : \Hom[\Mode]{m}{n}$ a functor
$\InterpModal{\mu} : \InterpMode{n} \to \InterpMode{m}$ as by taking the $\InterpMode{n}$ morphism
\[
  \begin{tikzpicture}[on grid, node distance = 2.5cm, baseline = -1.25cm]
    \node (SemD) {$\Sem{\Delta}$};
    \node (SynD) [below = 2cm of SemD] {$\Sub{n}(\Emp, \Syn{\Delta})$};
    \node (SemG) [right = 6cm of SemD] {$\Sem{\Gamma}$};
    \node (SynG) [below = 2cm of SemG] {$\Sub{n}(\Emp, \Syn{\Gamma})$};
    \path[->] (SemD) edge node [above] {$\Sem{\gamma}$} (SemG);
    \path[->] (SynD) edge node [below] {$\Syn{\gamma} \circ -$} (SynG);
    \path[->] (SemD) edge node [upright desc] {$\Pred{\Delta}$} (SynD);
    \path[->] (SemG) edge node [upright desc] {$\Pred{\Gamma}$} (SynG);
  \end{tikzpicture}
\]
to the $\InterpMode{m}$ morphism:
\[
  \begin{tikzpicture}[on grid, node distance = 2.5cm, baseline = -1.25cm]
    \node (SemD) {$\Sem{\Delta}$};
    \node (SynD) [below = 2cm of SemD] {$\Sub{m}(\Emp, \LockCx{\Syn{\Delta}})$};
    \node (SemG) [right = 6cm of SemD] {$\Sem{\Gamma}$};
    \node (SynG) [below = 2cm of SemG] {$\Sub{m}(\Emp, \LockCx{\Syn{\Gamma}})$};
    \path[->] (SemD) edge node [above] {$\Sem{\gamma}$} (SemG);
    \path[->] (SynD) edge node [below] {$\LockSb{\Syn{\gamma}} \circ -$} (SynG);
    \path[->] (SemD) edge node [upright desc] {$\Pred{\LockSb{\Delta}}$} (SynD);
    \path[->] (SemG) edge node [upright desc] {$\Pred{\LockSb{\Gamma}}$} (SynG);
  \end{tikzpicture}
\]
where the function $\Pred{\LockSb{\Delta}}$ is defined by
\[
  \Pred{\LockSb{\Delta}}(x) \defeq \LockSb{\Pred{\Delta}(x)}\ :\ \Emp \to \LockSb{\Syn{\Delta}}
\]
Notice that the equation $\LockCx{\Emp} = \Emp$ is necessary to ensure that this definition is
well-typed. The diagram commutes because locks act functorially on substitutions. It is also
functorial in $\mu$, because $\LockCx{\LockCx{\Gamma}}<\nu> = \LockCx{\Gamma}<\mu \circ \nu>$, and
$\LockCx{\Gamma}<1> = \Gamma$.

We define a 2-cell $\InterpModal{\mu} \To \InterpModal{\nu}$ for each $\alpha : \nu \To \mu$. The
component at $(\Sem{\Gamma}, \Syn{\Gamma}, \Pred{\Gamma})$ is
\[
  \begin{tikzpicture}[on grid, node distance = 2.5cm]
    \node (SemD) {$\Sem{\Gamma}$};
    \node (SynD) [below = 2cm of SemD] {$\Sub{m}(\Emp, \LockCx{\Syn{\Gamma}})$};
    \node (SemG) [right = 6cm of SemD] {$\Sem{\Gamma}$};
    \node (SynG) [below = 2cm of SemG] {$\Sub{m}(\Emp, \LockCx{\Syn{\Gamma}}<\nu>)$};
    \draw[double] (SemD) -- (SemG);
    \path[->] (SynD) edge node [below] {$\Key{\alpha}{\Syn{\Gamma}} \circ -$} (SynG);
    \path[->] (SemD) edge node [upright desc] {$\Pred{\LockCx{\Gamma}}$} (SynD);
    \path[->] (SemG) edge node [upright desc] {$\Pred{\LockCx{\Gamma}<\nu>}$} (SynG);
  \end{tikzpicture}
\]
This diagram commutes because of \eqref{eq:canonicity:terminal-key}, so it is a morphism in comma
category. Naturality follows from the numerous equations pertaining to keys and their composition.

This completes the definition of a strict 2-functor $\Coop{\Mode} \to \CAT_1$ as per
Section~\ref{def:semantics:context-structure}. Next, we must define the modal natural model structure for
each category of contexts.

\begin{rem}
  For the rest of this section we will freely use type-theoretic notation, viewing the predicate
  $\Sem{\Gamma} \to \Sub{m}(\Emp, \Syn{\Gamma})$ as a family fibred over $\Sub{m}(\Emp,
  \Syn{\Gamma})$, \ie~a map $\Sub{m}(\Emp, \Syn{\Gamma}) \to \VV$.

  We will follow the convention that symbols annotated with $\Sem{(-)}$ correspond to proof-relevant
  constructions---i.e. members of the predicate, or maps between predicates---whereas symbols
  annotated with $\Syn{(-)}$ correspond to pieces of syntax (\eg~terms, contexts, substitutions). In
  particular, $\Sem{\gamma}$ will not necessarily refer to a fibred map between proof-relevant
  predicates, but also to a generalized element of $\Sem{\Gamma}$.

  In other words, when $\Sem{\gamma} \in \Sem{\Gamma}$ and $\Pred{\Gamma}(\Sem{\gamma}) =
  \Syn{\gamma}\ :\ \Emp \to \Syn{\Gamma}$, we will abusively write $\Sem{\gamma} :
  \Sem{\Gamma}(\Syn{\gamma})$. That is, we will view $\Sem{\gamma}$ as living in the fibre of
  $\Pred{\Gamma}$ over $\Syn{\gamma}$. This amounts to considering $\Sem{\gamma}$ as a \emph{proof}
  that the predicate $\Sem{\Gamma}$ holds at the substitution $\Syn{\gamma}$. Observe that if
  $\Sem{\gamma} : \Sem{\LockCx{\Gamma}}(\Syn{\gamma})$ then $\Syn{\gamma}$ must be of the form
  $\LockSb{\Syn{\theta}}$ for some $\Syn{\theta}\ :\ \Emp \to \Syn{\Gamma}$ with $\Sem{\gamma} :
  \Sem{\Gamma}(\Syn{\theta})$.
\end{rem}

Types over $\InterpMode{m}$ are given by the presheaf
\begin{align*}
  &\CTy{m}(\Gamma) \defeq \{\\
  &\quad \Syn{A} \in \Ty{m}[1](\Syn{\Gamma});\\
  &\quad \Sem{A} :
    (\Syn{\gamma} : \Sub{m}(\Emp,\Syn{\Gamma})) \to
    (\Sem{\gamma} : \Sem{\Gamma}(\Syn{\gamma})) \to
    \Tm{m}(\Emp, \Sb{\Syn{A}}{\Syn{\gamma}}) \to \VV\\
  &\}
\end{align*}
Extending this presheaf with some additional data gives a presheaf of terms over $\InterpMode{m}$:
\begin{align*}
  &\CTm{m}(\Gamma) \defeq \{\\
  &\quad \Syn{A} \in \Ty{m}[1](\Syn{\Gamma});\\
  &\quad \Sem{A} :
    (\Syn{\gamma} : \Sub{m}(\Emp,\Syn{\Gamma})) \to
    (\Sem{\gamma} : \Sem{\Gamma}(\Syn{\gamma})) \to
    \Tm{m}(\Emp, \Sb{\Syn{A}}{\Syn{\gamma}}) \to \VV\\
  &\quad \Syn{M} \in \Tm{m}(\Syn{\Gamma}, \Syn{A});\\
  &\quad \Sem{M} :
    (\Syn{\gamma} : \Sub{m}(\Emp,\Syn{\Gamma})) \to
    (\Sem{\gamma} : \Sem{\Gamma}(\Syn{\gamma})) \to
    \Sem{A}(\Syn{\gamma}, \Sem{\gamma}, \Sb{\Syn{M}}{\Syn{\gamma}})\\
  &\}
\end{align*}
Thus, a \emph{type} over $\Gamma = (\Syn{\Gamma}, \varphi_\Gamma)$ in the glued model consists of a
type $\IsTy[\Syn{\Gamma}]{\Syn{A}}[1]$ of \MTT{}, along with another predicate, a family of
$\VV$-small sets, indexed over both closing substitutions $\Syn{\gamma}$ that satisfy the predicate
$\Sem{\Gamma}$ and terms of type $\Sb{\Syn{A}}{\Syn{\Gamma}}$.

A \emph{term} over $\Gamma$ in the glued model adds to the above a term
$\IsTm[\Syn{\Gamma}]{\Syn{M}}{\Syn{A}}$ of that type, and a section $\Sem{M}$ of the aforementioned
predicate. This section produces a proof that the predicate holds at that term after we close it by
applying any substitution $\Syn{\gamma}$ of which the $\Sem{\Gamma}$ holds. The reindexing action of
these presheaves is defined by the action of substitution on contexts, types, and terms of \MTT{}.

It can be shown that the projection $\El{m}(\Gamma) \defeq (\Syn{A},\Sem{A},\Syn{M},\Sem{M}) \mapsto
(\Syn{A}, \Sem{A})$ that maps terms to types by forgetting the additional data defines a
representable natural transformation in the sense of Section~\ref{sec:semantics:standard-cwf}; the full
proof can be found in the technical report. With respect to the connectives, we only show how to
interpret the base type $\Bool$, as per Section~\ref{sec:semantics:standard-bool}. For the formation and
introduction rules we define:
\begin{align*}
  \Syn{\CBool} &= \Bool &
  \Sem{\CBool} &=
    \lambda \Syn{\gamma},\Sem{\gamma}, \Syn{M}.\
      (\Sb{\Syn{M}}{\Syn{\gamma}} = \True) + (\Sb{\Syn{M}}{\Syn{\gamma}} = \False)\\
  \Syn{\CTrue} &= \True &
  \Sem{\CTrue} &= \lambda \_.\ \iota_0(\star)\\
  \Syn{\CFalse} &= \False &
  \Sem{\CFalse} &= \lambda \_.\ \iota_1(\star)
\end{align*}
We must now define the left lifting structure $\CRec : [\CTrue,\CFalse] \pitchfork \El{m}$. In
type-theoretic notation:
\begin{align*}
  \Syn{\CRec(C, [c_0,c_1], M)} &= \BoolRec{\Syn{C}}{\Syn{c_0}}{\Syn{c_1}}{\Syn{M}}\\
  \Sem{\CRec(C, [c_0,c_1], M)} &= \lambda \Syn{\gamma},\Sem{\gamma}.\ %
  \begin{cases}
    \Sem{c_0}(\Syn{\gamma},\Sem{\gamma}) & \text{if $\Sem{M}(\Syn{\gamma},\Sem{\gamma}) = \iota_0(\star)$}\\
    \Sem{c_1}(\Syn{\gamma},\Sem{\gamma}) & \text{if $\Sem{M}(\Syn{\gamma},\Sem{\gamma}) = \iota_1(\star)$}
  \end{cases}
\end{align*}

\subsection{Deriving Canonicity}

With the gluing model constructed, the rest of the proof is surprisingly easy and boils down to one
fact, which is immediate by inspection:
\begin{thm}
  \label{thm:canonicity:projection}
  The 2-natural transformation $\pi : \InterpMode{-} \To \mathbb{S}[-]$ from the glued model to the
  syntactic model which forgets the predicates extends to a morphism of models.
\end{thm}

Thus, assuming $\EqCx{\LockCx{\Emp}}{\Emp}$ it follows that

\begin{cor}
  \label{cor:canonicity:witness}
  For any closed term $\IsTm[]{M}{A}$, there is a witness for $\Sem{\Interp{A}}(M)$.
\end{cor}
\begin{proof}
  By Theorem~\ref{thm:canonicity:projection} and initiality we must have $\pi(\Interp{M}) = M$, and
  so $\Sem{\Interp{M}}$ is a witness.
\end{proof}
\begin{thm}[Closed Term Canonicity]
  \label{thm:canonicity:canonicity}
  If $\IsTm[\Emp]{M}{A}$ is a closed term, then
  \begin{itemize}
  \item If $A = \Bool$ then $\EqTm[\Emp]{M}{\True}{\Bool}$ or
    $\EqTm[\Emp]{M}{\False}{\Bool}$.
  \item If $A = \Id{A_0}{N_0}{N_1}$ then $\EqTm[\Emp]{N_0}{N_1}{A_0}$ and
    $\EqTm[\Emp]{M}{\Refl{N_0}}{\Id{A_0}{N_0}{N_1}}$.
  \item If $A = \Modify[\nu]{A_0}$ then there is an $\IsTm[\Emp]{N}{A_0}<n>$ such that
    $\EqTm[\Emp]{M}{\MkBox[\nu]{N}}{\Modify[\nu]{A_0}}$.
  \end{itemize}
\end{thm}
\begin{proof}
  Immediate by Corollary~\ref{cor:canonicity:witness} and the definition of the semantic predicates
  at $\Bool$, $\Id{A_0}{N_0}{N_1}$, and $\Modify[\nu]{A_0}$ respectively.
\end{proof}


\section{Dependent Right Adjoints}
\label{sec:dra}

Over the past couple of years the structure of a \emph{dependent right adjoint (DRA)} has arisen as
a natural notion of dependent modality in Martin-L\"of type theory. In this section we will study
the relationship between \MTT{} modalities and DRAs in detail. After reviewing the definition of a
DRA, we will prove that a suitably functorial collection of DRAs induces a model of \MTT{}. As
mentioned before, this implies that \MTT{} modalities are weaker than DRAs. Following that, we will
investigate sufficient conditions for extending an ordinary right adjoint to a DRA.

\subsection{Dependent right adjoints in natural models}
\label{sec:dra:dra}

A dependent right adjoint\footnote{DRAs were introduced by \cite{clouston:dra:2018} as
  endomodalities, but we generalise them to multiple modes.}
is an adaptation of the notion of adjunction to the dependent setting: instead of acting
on objects of the context category, the `right adjoint' only acts on types and terms.

Given a pair of natural models $(\DD, \El{\DD})$ and  $(\CC, \El{\CC})$, a DRA from the second to the
first comprises a functor $L : \DD \to \CC$ between the underlying context categories, as well as a
pullback diagram of the following shape in $\PSH{\DD}$:
\[
  \begin{tikzpicture}[node distance=2.5cm, on grid]
    \node[pullback] (ModTm) {$\Pre{L}{\CTm{\CC}}$};
    \node (Tm) [right = 4cm of ModTm] {$\CTm{\DD}$};
    \node (ModTy) [below = of ModTm] {$\Pre{L}{\CTy{\CC}}$};
    \node (Ty) [below = of Tm] {$\CTy{\DD}$};
    \path[->] (ModTm) edge node [left] {$\Pre{L}{\El{\CC}}$} (ModTy);
    \path[->] (ModTm) edge node [above] {$r$} (Tm);
    \path[->] (ModTy) edge node [below] {$R$} (Ty);
    \path[->] (Tm) edge node [right] {$\El{\DD}$} (Ty);
  \end{tikzpicture}
\]
$R$ is the action on types, and $r$ is the action on terms. Note that, while the `left adjoint' $L$
acts on context categories, the `right adjoint' $(R, r)$ only acts on types and terms. The fact the
square is a pullback amounts to requiring a multimode generalization of the definition given by
\cite{clouston:dra:2018}. Intuitively, for each term $\Gamma \vdash M : R(A)$, the pullback square
gives a unique term $L(\Gamma) \vdash N : A$ such that $\Gamma \vdash M = r(N) : R(A)$. If we wish
the modality to preserve the size of types, we must also require a $R'$ such that
\[
  \begin{tikzpicture}[node distance=2.5cm, on grid]
    \node (SF) {$\Pre{L}{\CSTy{\CC}}$};
    \node (S) [right = 4cm of ModTm] {$\CSTy{\DD}$};
    \node (TF) [below = 2cm of ModTm] {$\Pre{L}{\CTy{\CC}}$};
    \node (T) [below = 2cm of S] {$\CTy{\DD}$};
    \path[->] (SF) edge node [left] {} (TF);
    \path[->, densely dotted] (SF) edge node [above] {$R'$} (S);
    \path[->] (TF) edge node [below] {$R$} (T);
    \path[->] (S) edge node [right] {$\CLift$} (T);
  \end{tikzpicture}
\]

\subsection{DRAs as Models of \MTT{}}
\label{sec:dra:mtt-models}

We will now show that DRAs can be used to construct models of \MTT{}. As a consequence, \MTT{}
modalities are slightly weaker than DRAs.

\begin{thm}
  \label{thm:dra:dra}
  Suppose that we have
  \begin{itemize}
    \item for each $m \in \Mode$ a natural model $(\InterpMode{m},
      \Mor[\El{m}]{\CTm{m}}{\CTy{m}})_{m \in \Mode}$ of \MLTT{};
    \item for each modality $\mu : \Hom[\Mode]{m}{n}$ a size-preserving DRA $(\InterpModal{\mu},
      \CModify, \CMkBox)$ from $(\InterpMode{m}, \Mor[\El{m}]{\CTm{m}}{\CTy{m}})$ to
      $(\InterpMode{n}, \Mor[\El{n}]{\CTm{n}}{\CTy{n}})$;
    \item for each 2-cell $\alpha : \mu \To \nu$ in $\Mode$ a natural transformation
  $\InterpKey{\alpha} : \InterpModal{\nu} \To \InterpModal{\mu}$.
  \end{itemize}
  Moreover, suppose that the above choices are 2-functorial. Then this data can be assembled into a
  model of \MTT{}, where each each mode $m$ is interpreted by $(\InterpMode{m},
  \Mor[\El{m}]{\CTm{m}}{\CTy{m}})$.
\end{thm}
\begin{proof}
  Define a 2-functor $\Coop{\Mode} \to \CAT$ by $m \mapsto \InterpMode{m}$, $\mu \mapsto
  \InterpModal{\mu}$, and $\alpha \mapsto \InterpKey{\alpha}$. We must show how to define context extension, and how to interpret the connectives. As before, we only show the modal cases, the others being straightforward.

  \begin{description}
    \item[Modal Context Extension]
    For each type $\Mor[\YoEm{A}]{\Yo{\Gamma}}{\Pre{\InterpModal{\mu}}{\CTy{n}}}$ we need a pullback
    \[
      \begin{tikzpicture}[node distance=2.5cm, on grid]
        \node[pullback] (Gex) {$\Yo{\Gamma'}$};
        \node (Tm) [right = 3.5cm of ModTm] {$\Pre{\InterpModal{\mu}}{\CTm{n}}$};
        \node (G) [below = 2cm of Gex] {$\Yo{\Gamma}$};
        \node (Ty) [below = 2cm of Tm] {$\Pre{\InterpModal{\mu}}{\CTy{n}}$};
        \path[->] (Gex) edge node [left] {$\Yo{p'}$} (G);
        \path[->] (Gex) edge node [above] {$\YoEm{q'}$} (Tm);
        \path[->] (G) edge node [below] {$\YoEm{A}$} (Ty);
        \path[->] (Tm) edge node [right] {$\Pre{\InterpModal{\mu}}{\El{n}}$} (Ty);
      \end{tikzpicture}
    \]
    Write $\YoEm{\CModify(A)} \defeq \CModify \circ \YoEm{A}$. $\El{m}$ is a natural model, so form
    the pullback square for $\Gamma' \defeq \CECx{\Gamma}{\CModify(A)}$. Pasting this with the DRA
    pullback for $\CModify$ forms
    \begin{equation}
      \begin{tikzpicture}[node distance=2.5cm, on grid, baseline=(current  bounding  box.center)]
        \node[pullback] (Gex) {$\Yo{\CECx{\Gamma}{\CModify(A)}}$};
        \node (G) [below = 2cm of Gex] {$\Yo{\Gamma}$};
        \node[pullback] (Tm) [right = 3.5cm of ModTm] {$\Pre{\InterpModal{\mu}}{\CTm{n}}$};
        \node (Ty) [below = 2cm of Tm] {$\Pre{\InterpModal{\mu}}{\CTy{n}}$};
        \node (TmM) [right = 3.5cm of Tm] {$\CTm{m}$};
        \node (TyM) [below = 2cm of TmM] {$\CTy{m}$};
        \path[->] (Gex) edge node [left] {$\Yo{\CWk}$} (G);
        \path[->] (Gex) edge[bend left] node [above] {} (TmM);
        \path[densely dotted, ->] (Gex) edge node [above] {$\Transpose{\YoEm{\CVar}}$} (Tm);
        \path[->] (G) edge node [below] {$\YoEm{A}$} (Ty);
        \path[->] (Tm) edge node [right] {$\Pre{\InterpModal{\mu}}{\El{n}}$} (Ty);
        \path[->] (Tm) edge node [above] {$\CMkBox$} (TmM);
        \path[->] (Ty) edge node [below] {$\CModify$} (TyM);
        \path[->] (TmM) edge node [right] {$\El{m}$} (TyM);
      \end{tikzpicture}
      \label{diag:dra-to-modal}
    \end{equation}
    As the outer square commutes, we can fill in the dotted arrow. By the pullback lemma, the square
    on the left is a pullback too. Letting $\ECx{\Gamma}{A} \defeq \CECx{\Gamma}{\CModify(A)}$
    proves that $\left(\El{m}\right)_{m \in \Mode}$ is a modal natural model.

  \item[Modal Types]

    This is the heart of the proof. First, we need a commuting square
    \begin{equation}
      \begin{tikzpicture}[node distance=2.5cm, on grid, baseline=(current  bounding  box.center)]
        \node[pullback] (ModTm) {$\Pre{\InterpModal{\mu}}{\CTm{n}}$};
        \node (Tm) [right = 4cm of ModTm] {$\CTm{m}$};
        \node (ModTy) [below = 1.5cm of ModTm] {$\Pre{\InterpModal{\mu}}{\CTy{n}}$};
        \node (Ty) [below = 1.5cm of Tm] {$\CTy{m}$};
        \path[->] (ModTm) edge node [left] {} (ModTy);
        \path[->] (ModTm) edge node [above] {$\CMkBox$} (Tm);
        \path[->] (ModTy) edge node [below] {$\CModify$} (Ty);
        \path[->] (Tm) edge node [right] {$\El{m}$} (Ty);
      \end{tikzpicture}
      \label{diag:DRA-pullback-modal}
    \end{equation}
    Such a square is given as part of a DRA by definition, and is in fact a pullback!

    To model the elimination rule, recall the definition of the object $M$ used in Section~\ref{sec:semantics:standard-modify}:
    \[
      \begin{tikzpicture}[node distance=2.5cm, on grid]
        \node (ModTm) {$\Pre{\InterpModal{\mu \circ \nu}}{\CTm{o}}$};
        \node[pullback] (M) [below right = 2cm and 2cm of ModTm] {$\Pre{\InterpModal{\mu}}{M}$};
        \node (Tm) [right = 4cm of M] {$\Pre{\InterpModal{\mu}}{\CTm{n}}$};
        \node (ModTy) [below = 2cm of M] {$\Pre{\InterpModal{\mu \circ \nu}}{\CTy{o}}$};
        \node (Ty) [below = 2cm of Tm] {$\Pre{\InterpModal{\mu}}{\CTy{n}}$};
        \path[->, densely dashed] (ModTm) edge node [upright desc] {$\Pre{\InterpModal{\mu}}{m}$} (M);
        \path[->] (M) edge node {} (Tm);
        \path[->] (M) edge  node {} (ModTy);
        \path[->] (ModTm) edge[bend right] node [left] {$\Pre{\InterpModal{\mu \circ \nu}}{\El{o}}$} (ModTy);
        \path[->] (ModTm) edge[bend left] node [above] {$\Pre{\InterpModal{\mu}}{\CMkBox[\nu]}$} (Tm);
        \path[->] (ModTy) edge node [below] {$\Pre{\InterpModal{\mu}}{\CModify[\nu]}$} (Ty);
        \path[->] (Tm) edge node [right] {$\Pre{\InterpModal{\mu}}{\El{n}}$} (Ty);
      \end{tikzpicture}
    \]
    As $\Pre{\InterpModal{\mu}}$ preserves pullbacks, the outer square is a pullback too. Hence
    $\Pre{\InterpModal{\mu}}{m}$ must be an isomorphism. The elimination rule requires a
    left-lifting structure:
    \[
      \vdash \COpen[\nu]^\mu :
        \parens*{\Pre{\InterpModal{\mu}}{m}}
          \pitchfork
        \Pre*{(\Pre{\InterpModal{\mu \circ \nu}}{\CTy{o}})}{\El{m}[-]}
    \]
    Using the inverse of $\Pre{\InterpModal{\mu}}{m}$ we can construct this by
    \[
      \COpen[\nu]^\mu \defeq \lambda C.\ \lambda c.\ c \circ \Pre*{\InterpModal{\mu}}{m^{-1}}
    \]
  \item[$\prod$ Structure] \leavevmode


    Equipping each $\Mor[\El{m}]{\CTm{m}}{\CTy{m}}$ with a modal $\prod$ structure is relatively
    straightforward to do in the internal language; intuitively, the reason is the isomorphism
    \[
      \ReIdx{\parens{\Pre{\InterpModal{\mu}}{\El{n}}}}{A}
        \cong
      \ReIdx{\El{m}}{\CModify{A}}
    \]
    which is derived from the fact $\ECx{\Gamma}{A} \defeq \CECx{\Gamma}{\CModify{A}}$ (where
    the first dot is the defined context extension, and the second dot is given by the natural
    model). However, we can also prove it in a more abstract way: we paste together the two pullback
    squares
    \[
      \begin{tikzpicture}[node distance=2.5cm, on grid]
        \node[pullback] (PiTm) {$\Poly{\Pre{\InterpModal{\mu}}{\El{n}}}{\CTm{m}}$};
        \node (PiTy) [below = 2cm of PiTm] {$\Poly{\Pre{\InterpModal{\mu}}{\El{n}}}{\CTy{m}}$};
        \node[pullback] (PiTm2) [right = 4cm of PiTm] {$\Poly{\El{m}}{\CTm{m}}$};
        \node (PiTy2) [below = 2cm of PiTm2] {$\Poly{\El{m}}{\CTy{m}}$};
        \node (Tm) [right = 4cm of PiTm2] {$\CTm{m}$};
        \node (Ty) [below = 2cm of Tm] {$\CTy{m}$};
        \path[->] (Tm) edge node[right] {$\El{m}$} (Ty);
        \path[->] (PiTm) edge node[above] {$\phi_{\CTm{m}}$} (PiTm2);
        \path[->] (PiTy) edge node[below] {$\phi_{\CTy{m}}$} (PiTy2);
        \path[->] (PiTm) edge node[left] {$\Poly{\Pre{\InterpModal{\mu}}{\El{n}}}{\El{m}}$} (PiTy);
        \path[->] (PiTm2) edge node[left] {$\Poly{\El{m}}{\El{m}}$} (PiTy2);
        \path[->] (PiTm2) edge node[above] {$\CLam$} (Tm);
        \path[->] (PiTy2) edge node[below] {$\CPi$} (Ty);
      \end{tikzpicture}
    \]
    The square on the right is the pullback that interprets $\CPi$ in the natural model $\El{m}$.
    The square on the left is a naturality square of the natural transformation
    \[
      \phi :
        \Poly{\Pre{\InterpModal{\mu}}{\El{n}}}{-}
          \Rightarrow
        \Poly{\El{m}}{-}
    \]
    which exists because the pullback square \eqref{diag:DRA-pullback-modal} defines a morphism of
    polynomials. Moreover, the naturality squares of $\phi$ are cartesian: see \cite[\S\S
    1.2.16--1.2.18]{newstead:2018}.
    \qedhere
  \end{description}
\end{proof}

This theorem is a particularly flexible tool, as many modalities naturally form DRAs, and it is
easier to check the DRA conditions than \MTT{} model conditions as summarized in Definition
\ref{def:semantics:semantics}. As a first example of this flexibility we show that it leads to an
almost immediate proof of consistency.

\begin{cor}
  \label{cor:dra:consistency}
  No matter what the mode theory is, there is no term $\IsTm[\Emp]{M}{\Id{\Bool}{\True}{\False}}$.
  In other words, \MTT{} is consistent.
\end{cor}
\begin{proof}
  Suppose that we have a model of \MLTT{} with one universe in some category $\CC$. We may construct
  a functor $\Coop{\Mode} \to \CAT$ by sending every mode to $\CC$, and everything else to the
  identity. This is stricty 2-functorial, and each identity functor is a DRA. Hence, by
  Theorem~\ref{thm:dra:dra} there is a model of \MTT{} in which each mode is interpreted by
  $\CC$. Therefore, if a term $M : \Id{\Bool}{\True}{\False}$ were definable in \MTT{}, we would
  have a term of that type in every model of \MLTT{}. But \MLTT{} itself is consistent: see
  \cite{coquand:2018} for a short proof.
\end{proof}

\subsection{DRAs from right adjoints}
\label{sec:dra:dras-from-right-adjoints}

Having established that a series of models of \MLTT{} related by DRAs can be used to interpret
\MTT{}, we now turn to the problem of constructing those DRAs themselves. We shall prove a lemma
that allows us to lift any well-behaved right adjoint to a DRA. Versions of this result have
appeared before, both in the paper on DRAs \cite[Lemma 17]{clouston:dra:2018}, and in a technical
report By the third author \cite[Prop.~2.1.4]{nuyts:tech-report:2018}.

In Section~\ref{sec:semantics:morphism} we discussed the notion of a \emph{strict} morphism of natural
models. Using the same notation we define the following weaker notion.

\begin{defi}
  A \emph{weak morphism of natural models} $(\CC, \CTy{c}) \to (\DD, \CTy{d})$ consists of a functor
  $F : \CC \to \DD$, and a commuting square
  \[
    \begin{tikzpicture}[on grid, node distance = 2.5cm]
      \node (TmC) {$\CTm{c}$};
      \node (TmD) [right = of TmC] {$\Pre{F}{\CTm{d}}$};
      \node (TyC) [below = 2cm of TmC] {$\CTy{c}$};
      \node (TyD) [below = 2cm of TmD] {$\Pre{F}{\CTy{d}}$};
      \path[->] (TmC) edge node [left] {$\El{c}$} (TyC);
      \path[->] (TmC) edge node [above] {$\widetilde{\varphi}$} (TmD);
      \path[->] (TyC) edge node [below] {$\varphi$} (TyD);
      \path[->] (TmD) edge node [right] {$\Pre{F}{\El{d}}$} (TyD);
    \end{tikzpicture}
  \]
  such that $F(1) = 1$ and the canonical morphism $F(\CECx{\Gamma}{A}) \to
  \CECx{F\Gamma}{\varphi(A)}$ is an isomorphism. We say that this morphism \emph{preserves size}
  whenever there is a commuting square
  \[
    \begin{tikzpicture}[on grid, node distance = 2.5cm]
      \node (STyC) {$\CSTy{c}$};
      \node (STyD) [right = of STyC] {$\Pre{F}{\CSTy{d}}$};
      \node (TyC) [below = 2cm of STyC] {$\CTy{c}$};
      \node (TyD) [below = 2cm of STyD] {$\Pre{F}{\CTy{d}}$};
      \path[->] (STyC) edge node [left] {$\CLift$} (TyC);
      \path[->] (STyC) edge node [above] {$\CSTy{\varphi}$} (STyD);
      \path[->] (TyC) edge node [below] {$\varphi$} (TyD);
      \path[->] (STyD) edge node [right] {$\Pre{F}{\CLift}$} (TyD);
    \end{tikzpicture}
  \]
\end{defi}

This kind of morphism can also be found in the thesis of \cite[\S\S 2.3.9]{newstead:2018}. We are
interested in it because it captures exactly the necessary good behaviour which is required to
extend a right adjoint to act on types and terms.

\begin{lem}
  \label{lem:constructing-dra:lifting-adjunction}
  Suppose that $(\CC, \El{\CC})$ and $(\DD, \El{\DD})$ are natural models, and that $L \Adjoint R$
  is an adjunction between $\CC$ and $\DD$. If the right adjoint $R : \CC \to \DD$ extends to a weak
  morphism of natural models then it gives rise to a dependent right adjoint. Moreover, the
  resulting DRA is size-preserving whenever $R$ is.
\end{lem}
\begin{proof}
  We first fix some notation: we write $\eta : \textsf{Id} \Rightarrow RL$ for the unit of the
  adjunction $L \Adjoint R$. Moreover, we assume a commuting square
  \begin{equation}
    \begin{tikzpicture}[on grid, node distance = 2.5cm, baseline=(current  bounding  box.center)]
      \node (TmC) {$\CTm{\CC}$};
      \node (TmD) [right = of TmC] {$\Pre{R}{\CTm{\DD}}$};
      \node (TyC) [below = 2cm of TmC] {$\CTy{\CC}$};
      \node (TyD) [below = 2cm of TmD] {$\Pre{R}{\CTy{\DD}}$};
      \path[->] (TmC) edge node [left] {$\El{\CC}$} (TyC);
      \path[->] (TmC) edge node [above] {$\textsf{r}$} (TmD);
      \path[->] (TyC) edge node [below] {$\textsf{R}$} (TyD);
      \path[->] (TmD) edge node [right] {$\Pre{R}{\El{\DD}}$} (TyD);
    \end{tikzpicture}
    \label{diag:constructing-dra:lifting-adjunction:weakcwf}
  \end{equation}
  that witnesses the weak natural model morphism structure of $R$, and write
  \[
    \nu_{\Gamma, A} :
      \CECx{R\Gamma}{\textsf{R}(A)}
        \xrightarrow{\cong}
      R(\CECx{\Gamma}{A})
  \]
  for the canonical isomorphism corresponding to $\YoEm{A} : \Yo{\Gamma} \Rightarrow \CTy{\CC}$.
  
  We construct the DRA by first applying the weak morphism $\textsc{R}$ to a dependent type over a
  context of the form $L(\Delta)$, and then pulling that back along the unit of the adjunction.
  Diagrammatically, we define the square
  \[
    \begin{tikzpicture}[node distance=2.5cm, on grid]
      \node (ModTm) {$\Pre{L}{\CTm{\CC}}$};
      \node (ModTy) [below = 2cm of ModTm] {$\Pre{L}{\CTy{\CC}}$};
      \node (LRTm) [right = 2.5cm of ModTm] {$\Pre{L}\Pre{R}{\CTm{\DD}}$};
      \node (LRTy) [below = 2cm of LRTm] {$\Pre{L}\Pre{R}{\CTy{\DD}}$};
      \node (Tm) [right = 2.5cm of LRTm] {$\CTm{\DD}$};
      \node (Ty) [below = 2cm of Tm] {$\CTy{\DD}$};
      \path[->] (ModTm) edge node [left] {$\Pre{L}{\El{\CC}}$} (ModTy);
      \path[->] (ModTm) edge node [above] {$\Pre{L}{\textsf{r}}$} (LRTm);
      \path[->] (ModTy) edge node [below] {$\Pre{L}{\textsf{R}}$} (LRTy);
      \path[->] (Tm) edge node [right] {$\El{\DD}$} (Ty);
      \path[->] (LRTm) edge node [upright desc] {$\Pre{L}\Pre{R}{\El{\DD}}$} (LRTy);
      \path[->] (LRTm) edge node [above] {$\Pre{\eta}{}_{\CTm{\DD}}$} (Tm);
      \path[->] (LRTy) edge node [below] {$\Pre{\eta}{}_{\CTy{\DD}}$} (Ty);
    \end{tikzpicture}
  \]
  The left part is the image of
  \eqref{diag:constructing-dra:lifting-adjunction:weakcwf} under the $\Pre{L}$ functor, and the
  right part is a naturality square for the natural transformation $\Pre{\eta} : \Pre{L}\Pre{R} =
  \Pre{(RL)} \Rightarrow \textsf{Id}$ induced by the unit.

  To show that this is a DRA we must show that this is a pullback, and it suffices to do so on the representables.  Assume
  we have $\YoEm{A} : \Yo{\Delta} \Rightarrow \Pre{L}\CTy{\CC}$ and a $\YoEm{M} : \Yo{\Delta}
  \Rightarrow \CTm{\DD}$ such that the diagram commutes. Switching to type-theoretic notation, this
  amounts to a type $L(\Delta) \vdash A\ \mathsf{type}$---which gives rise to a type $R(L(\Delta))
  \vdash \textsf{R}(A)\ \mathsf{type}$ by applying $\textsf{R}$---and a term $\Delta \vdash M :
  \Sb{\textsf{R}(A)}{\eta_\Delta}$. The universal property of the pullback dictates that we must
  show the existence of a unique term $L(\Delta) \vdash N : A$ such that
  \begin{equation}
    RL(\Delta) \vdash
    \Sb{\textsf{r}(N)}{\eta_\Delta}
    =
    M : \Sb{\textsf{R}(A)}{\eta_\Delta}
    \label{eq:constructing-dra:lifting-adjunction:weakcwf}
  \end{equation}
  First, observe that we can form the substitution $\ESb{\eta_\Delta}{M} : \Delta \to
  \CECx{RL\Delta}{\textsf{R}(A)}$. We can then postcompose the isomorphism $\nu_{\Delta, A}$ to
  obtain a morphism of type $\Delta \to R(\CECx{L\Delta}{A})$. To this we can apply $L$ and
  postcompose the counit $\epsilon_{\CECx{L\Delta}{A}}$ to obtain a substitution
  \[
    k \defeq \epsilon_{\CECx{L\Delta}{A}} \circ L(\nu_{\Delta, A} \circ \ESb{\eta_\Delta}{M})
             : L\Delta \to \CECx{L\Delta}{A}
  \]
  Using naturality of the counit and the equations satisfied by the canonical isomorphism
  $\nu_{\Delta, A}$, it is easy to show that $\CWk \circ k = \ISb : L\Delta \to L\Delta$, and hence
  that we can extract a term
  \[
    L(\Delta) \vdash N \defeq \Sb{\CVar}{k} : A
  \]
  Using naturality of $\textsf{r}(-)$, naturality of the unit, and one of the triangle identities,
  we can calculate that this term satisfies equation
  \eqref{eq:constructing-dra:lifting-adjunction:weakcwf}. Finally, we can prove this choice is
  unique by calculating that any such $N$ necessarily satisfies $k = \CECx{\ISb}{N}$, and hence that
  $\Sb{\CVar}{k} = N$.

  It is routine to show that this is size-preserving, using the fact that $\textsf{R}$ preserves
  size.
\end{proof}

The converse is not in general true: a dependent right adjoint need not extend to a functor on the
category of contexts. Nevertheless, it does whenever the category of contexts is \emph{democratic}
\cite{clairambault:2014}, i.e. if every context is isomorphic to extending the empty context by
some type: see \cite[\S 4.1]{clouston:dra:2018} for a proof.


\section{Presheaf Models}
\label{sec:presheaves}

It is well-known that the category $\PSH{C}$ of presheaves over any small category $\CC$ is a model
of Martin-L\"of type theory. A functor $\mu : \CC \to \DD$ induces by \emph{precomposition} a
functor
\[
  \Pre{\mu} : \PSH{\DD} \to \PSH{\CC}
\]
between categories of presheaves. This functor has a right adjoint
\[
  \RKan{\mu} : \PSH{\CC} \to \PSH{\DD}
\]
given by \emph{right Kan extension} \cite[\S X.3]{maclane:1978} \cite[\S 9.6]{awodey:2010} \cite[\S
6]{riehl:2016}. We show that an appropriately functorial version of this structure can be
bootstrapped into a model of \MTT{}, where the modalities are right adjoints to this precomposition
functor. More concretely, starting with a small 2-category $\II$, and a functor
\[
  J : \II \to \CAT
\]
we will construct a model of \MTT{} where each mode corresponds to the category $\PSH{J(i)}$, and the modalities are the functors $\RKan{J(f)}$, for each $f \in \Hom[\II]{i_0}{i_1}$.

\paragraph{Context structure.} We define a strict 2-functor $\Interp{-} : \Coop{\II} \to \CAT$ by
\begin{alignat*}{2}
  &i\                &&\longmapsto
    \InterpMode{i} \defeq \PSH{J(i)} \\
  &f : i \to j\      &&\longmapsto
    \InterpModal{f} \defeq \Pre{J(f)} : \PSH{J(j)} \to \PSH{J(i)} \\
  &\alpha : f \To g\ &&\longmapsto
    \InterpKey{\alpha} \defeq \Pre{J(\alpha)} : \Pre{J(g)} \To \Pre{J(f)}
\end{alignat*}
The variance is correct: recall that precomposition is a strict 2-functor
\[
  \Pre{(-)} : \Coop{\CAT} \to \CAT
\]
which maps a functor $f : \CC \to \DD$ to $\Pre{f} : \PSH{\DD} \to \PSH{\CC}$, and a natural
transformation $\alpha : f \To g$ to $\Pre{\alpha} : \Pre{g} \To \Pre{f}$, given by
$\Pre{\alpha}_{P, c} \defeq P(\alpha_c) : P(g(c)) \to P(f(c))$. 2-functoriality is immediate, as for
example $\Pre{J(f)} \circ \Pre{J(g)} = \Pre{(J(g) \circ J(f))} = \Pre{J(g \circ f)}$.

\paragraph{Modal natural models.}

To interpret the `mode-local' structure we must construct a modal natural model in each
$\InterpMode{i}$. It is well-known that every presheaf topos $\PSH{\CC}$ gives rise to a rich model
of \MLTT{}: see e.g. \cite[\S 4.1]{hofmann:1997} or \cite{coquand:2013}.

\emph{Contexts} are interpreted as objects of the presheaf category $\PSH{\CC}$. \emph{Types} are
presheaves $\PSH{\EL{\Gamma}}$ over the category of elements $\EL{\Gamma}$ of a context $\Gamma :
\PSH{\CC}$.\footnote{There is an equivalence $\PSH{\EL{\Gamma}} \Equiv \SLICE{\PSH{\CC}}{\Gamma}$
which shows that types are families $P \To \Gamma$ in the slice category. However, using the latter
definition would lead to strictness issues.} We define the action of a substitution $\sigma : \Delta
\To \Gamma$ on a type $A : \PSH{\EL{\Gamma}}$ by
\[
  A\brackets{\sigma} \defeq
    \Op{(\EL{\Delta})} \xrightarrow{\Op{(\EL{\sigma})}} \Op{(\EL{\Gamma})} \xrightarrow{A} \SET
\]
This is functorial because $\EL{-} : \PSH{\CC} \to \CAT$ and $\Op{-} : \CAT \to \CAT$ are.

A \emph{term} of type $A$ is a global section of $A$, i.e. a morphism
$\Hom[\PSH{\EL{\Gamma}}]{1}{A}$. We define the action of a substitution $\sigma : \Delta \To \Gamma$
on a term $M : \Hom{1}{A}$ by whiskering:
\[
  M\brackets{\sigma} \defeq
    M \ast \Op{(\EL{\sigma})} :
      1 \circ \Op{(\EL{\sigma})} \To A \circ \Op{(\EL{\sigma})} = A\brackets{\sigma}
\]
As $1 \circ \Op{\EL{\sigma}} = 1$, this has the right type. It is functorial because whiskering is.

\begin{rem}[Size Issues]
  One cannot be too careful with size issues when considering presheaf models. In
  Section~\ref{def:semantics:representable-nat} we demanded that the category of contexts be \emph{small},
  so that we can then formulate a large category of models. $\PSH{\CC}$ is certainly not small. We
  can mend this by assuming a Grothendieck universe $\VV$ large enough to contain $\CC$ in the
  ambient set theory, and re-defining $\PSH{\CC}$ to consist of the presheaves
  $P : \Op{\CC} \to \VV$ with \emph{small fibers}. As $\VV$ is closed under all set-theoretic
  operations, this is still a model, and $\PSH{\CC}$ is small.

  To interpret universes we need to know that the fibers of types in $\PSH{\EL{\CC}}$ are even
  \emph{smaller}. Thus, we further assume a second, inner Grothendieck universe $\VV' \subset \VV$.
  To a type theorist, this is just the standard technique of `bumping' a universe level.
\end{rem}

\paragraph{Connectives.}

Presheaf models support dependent sums and products, and extensional identity types (and therefore
intensional identity types): see \cite[\S 4.2]{hofmann:1997}. On the premise that the underlying set
theory has a set-theoretic universe, they also support a universe, through a construction of
\cite{hofmann-streicher:1997}. See also \cite{coquand:2013}.

\paragraph{Modalities.}

It remains to show that each $\InterpModal{f} \defeq \Pre{J(f)} : \PSH{J(j)} \to \PSH{J(i)}$ has a
corresponding modality acting on terms and types. We will do so by using previous results. By
Theorem \ref{thm:dra:dra}, it suffices to construct a \emph{dependent right adjoint} to
$\Pre{J(f)}$. But recall that each lock functor $\Pre{J(f)}$ already has an (ordinary) right
adjoint, viz.\ $\RKan{J(f)}$. Thus, by Lemma \ref{lem:constructing-dra:lifting-adjunction} it will
suffice to show that the action of this right adjoint extends to types and terms. We then have DRAs,
and hence a model of $\MTT{}$.

The following result has previously been shown in a tech report by the third author~\cite[Prop.
2.2.9]{nuyts:tech-report:2018}. We reproduce it here for the sake of completeness.

\begin{lem}
  \label{lem:constructing-dra:direct-image}
  The right adjoint to precomposition $\RKan{\mu} : \PSH{\CC} \to \PSH{\DD}$ induces a DRA for any
  $\mu : \CC \to \DD$. Moreover, $\RKan{\mu}$ is size-preserving for any Grothendieck universe.
\end{lem}
\begin{proof}
  We use Lemma \ref{lem:constructing-dra:lifting-adjunction} once more. The action of $\RKan{\mu}$ on types and terms is given by
  \begin{align*}
    \RKan{\mu}{A} \in \PSH{\EL{\RKan{\mu}{\Gamma}}} &=
    (D \in \DD, a \in \RKan{\mu}{\Gamma}(D)) \mapsto \Hom[\PSH{\EL{\Pre*{\mu}{\Yo{D}}}}]{1}{A\brackets*{\Transpose{\YoEm{a}}}}\\
    \RKan{\mu}{M} \in \Hom{1}{\RKan{\mu}{A}} &= (D \in \DD, a \in \RKan{\mu}{\Gamma}(D)) \mapsto \Pre{(\EL{\Transpose{\YoEm{a}}})}{M}
  \end{align*}
  Both of these actions are well-typed. For types, as $a \in \RKan{\mu}{\Gamma}(D)$ we have
  $\YoEm{a} : \Yo{D} \To \RKan{\mu}{\Gamma}$, so by transposition $\Transpose{\YoEm{a}} :
  \Pre*{\mu}{\Yo{D}} \To \Gamma$. For terms, notice that $\EL{\Transpose{\YoEm{a}}} :
  \EL{\Pre*{\mu}{\Yo{D}}} \to \EL{\Gamma}$, recall that $A\brackets*{\sigma} \defeq A \circ
  \EL{\sigma}$, and that precomposition preserves the terminal object on-the-nose.

  \paragraph{The presheaf action.}
  The action of $\RKan{\mu}{A}$ is subtle: it is given by the functor
  \[
    \Pre{(\EL{\Pre{\mu}{\Yo{f}}})} :
      \PSH{\EL{\Pre{\mu}{\Yo{D}}}} \to \PSH{(\EL{\Pre{\mu}{\Yo{D'}}}}
  \]
  for each $f : \Hom[\DD]{D'}{D}$. In more detail, given $f : \Hom[\DD]{D'}{D}$, $a \in
  \RKan{\mu}{\Gamma}(D)$, $A \in \PSH{\EL{\Gamma}}$, and $x \in \RKan{\mu}{A}(D, a) \defeq
  \Hom{1}{A\brackets*{\Transpose{\YoEm{a}}}}$, we define $x \cdot f \in \RKan{\mu}{A}(D', a \cdot
  f)$ by
  \[
    x \cdot f \defeq \Pre*{(\EL{\Pre{\mu}{\Yo{f}}})}{x}
      : \Hom[\PSH{\EL{\Pre*{\mu}{\Yo{D'}}}}]
        {\Pre{(\EL{\Pre{\mu}{\Yo{f}}})}{1}}
        {\Pre*{(\EL{\Pre{\mu}{\Yo{f}}})}{A\brackets*{\Transpose {\YoEm{a}}}}}
  \]
  This is of the right type; reindexing preserves the terminal,
  $\Pre{(\EL{\Pre{\mu}{\Yo{f}}})}{1} = 1$. Moreover,
  \[
      \Transpose{\YoEm{a}} \circ \Pre{\mu}{\Yo{f}}
    = \Transpose{\YoEm{a} \circ \Yo{f}}
    = \Transpose{\YoEm{a \cdot f}}
  \]
  by naturality of the adjunction and of Yoneda. Using this calculation, we see that
  \[
    \Pre*{(\EL{\Pre{\mu}{\Yo{f}}})}{A\brackets*{\Transpose {\YoEm{a}}}}
    \defeq A \circ \EL{\Transpose {\YoEm{a}}} \circ \EL{\Pre{\mu}{\Yo{f}}}
    = A \circ \EL{\Transpose{\YoEm{a \cdot f}}}
    \defeq A\brackets*{\Transpose{\YoEm{a \cdot f}}}
  \]
  Hence $x \cdot f \in \RKan{\mu}{A}(D', a \cdot f)$. This assignment is functorial because
  $\EL{-}$, $\Pre{(-)}$ and $\Yo{-}$ are.

  \paragraph{Naturality.}
  We must show that both of these definitions are natural with respect to substitution, i.e. that
  $(\RKan{\mu}{A})\brackets*{\RKan{\mu}{\gamma}} = \RKan*{\mu}{A\brackets*{\gamma}}$, and similarly
  for terms.

  For types, suppose we are given $\gamma : \Delta \to \Gamma$ and $A \in \PSH{\EL{\Gamma}}$.
  Carefully unfolding both sides of the desired equation, for any $D \in \DD$ and $a \in
  \RKan{\mu}\Delta(D)$ we must show that
  \[
    \Hom[\PSH{\EL{\Pre{\mu}{\Yo{D}}}}]{1}{A\brackets*{\Transpose{\YoEm{\RKan{\mu}\gamma_D(a)}}}}
    =
    \Hom[\PSH{\EL{\Pre{\mu}{\Yo{D}}}}]{1}{A\brackets{\gamma}\brackets*{\Transpose{\YoEm{a}}}}
  \]
  But, by naturality of both the adjunction and Yoneda:
  \[
    \gamma \circ \Transpose{\YoEm{a}}
    = \Transpose{\RKan{\mu}\gamma \circ \YoEm{a}}
    = \Transpose{\YoEm{\RKan{\mu}\gamma_D(a)}}
  \]
  Hence the two sets are the same. The calculation for terms is of a similar ilk.

  \paragraph{Preservation of context extension.}
  We would like to show that the canonical morphism
  \[
    \Mor[\angles{\RKan{\mu}{\CWk}, \RKan{\mu}{\CVar}}]{\RKan*{\mu}{\CECx{\Gamma}{A}}}{\CECx{\RKan{\mu}{\Gamma}}{\RKan{\mu}{A}}}
  \]
  is invertible. Consider an element $e : \Yo{D} \To \RKan*{\mu}{\CECx{\Gamma}{A}}$. We can
  transpose along the adjunction $\Pre{\mu} \dashv \RKan{\mu}$ and decompose it to obtain a
  substitution and a term
  \begin{align*}
    e_0 &: \Pre{\mu}{\Yo{D}} \To \Gamma
      &
    e_1 &: \Hom[\PSH{\EL{\Pre{\mu}{\Yo{D}}}}]{1}{A\brackets*{e_0}}
  \end{align*}
  We can thus write $e = \Transpose{\angles{e_0, e_1}}$. Thus, we can use naturality of the
  adjunction and of substitution to compute the action of $\angles{\RKan{\mu}{\CWk},
  \RKan{\mu}{\CVar}}$ on this $e$:
  \[
    \angles{\RKan{\mu}{\CWk}, \RKan{\mu}{\CVar}} \circ e
    = \angles{\RKan{\mu}{\CWk} \circ e, (\RKan{\mu}{\CVar})\brackets*{e}}\\
    = \angles{\Transpose{\CWk \circ \angles{e_0,e_1}}, \CVar\brackets{\angles{e_0,e_1}}}\\
    = \angles{\Transpose{e}_0, e_1}
  \]
  We can then specify an inverse on generalized elements by $\angles{\gamma, M} \mapsto
  \Transpose{\angles{\Transpose{\gamma}, M}}$.

  Size preservation is immediate: if $A$ is small then so are its reindexings and the collections of
  points at each slice.
\end{proof}

\subsection{The other adjunction}

We have so far concentrated on the adjunction $\Pre{\mu} \dashv \RKan{\mu}$ that arises through right Kan extension. Nevertheless, precomposition also has a left adjoint $\LKan{\mu}$ arising from \emph{left Kan extension}. Might we also be able to interpret the lock functors by this left adjoint $\LKan{\mu}$, and lift precomposition $\Pre{\mu}$ to a modality instead?

It is in fact relatively easy to show that $\Pre{\mu}$ extends to a dependent right adjoint. However, the
left Kan extensions $\LKan{\mu}$ cannot be assembled into a modal context structure. The reason is
that context structures are strict 2-functors, but left Kan extensions do not compose strictly: we
only have an isomorphism $\LKan{F} \circ \LKan{G} \cong \LKan{(G \circ F)}$. We have proven a
strictification theorem that straightens these issues, but that is beyond the scope of this paper.


\section{Guarded Recursion}
\label{sec:guarded-recursion}

We now show how \MTT{} can be applied to a well-known modal situation: guarded recursion. By
instantiating \MTT{} with a carefully chosen mode theory and axiomitizing certain operations
specific to guarded recursion (\ie{} L{\"o}b induction), we obtain a calculus for guarded recursion
simpler than prior hand-crafted calculi. We demonstrate the practicality of this guarded variant of
\MTT{} by reproducing some examples from prior work on guarded recursion~\cite{bizjak:2016}.

The key idea of guarded recursion~\cite{nakano:2000} is to use a modality $\Later$, usually called
\emph{later}, to mark the types of data that may be used only if some `computational progress' (e.g.
a tick of a clock) has taken place, thereby enforcing productivity at the level of types. The later
modality is usually equipped with three basic operations:
\begin{mathpar}
  \Next : A \to \Later A \and
  (\ZApp{}{}{}) : \Later (A \to B) \to \Later A \to \Later B \and
  \Lob : (\Later A \to A) \to A
\end{mathpar}
The first two make $\Later$ into an \emph{applicative functor}~\cite{mcbride:2008}. The third,
which is commonly known as L{\"o}b induction, is a \emph{guarded fixed point operator}
\cite{milius:2013}: it enables us to make definitions by \emph{provably productive} recursion.

$\Later$ also applies to the universe, so one can define data types by guarded recursion. The
classic example is the \emph{guarded stream type} $\Str_A \cong A \times \Later \Str_A$, with
constructor
\[
 \Cons_A : A \times \Later \Str_A \cong \Str_A
\]
The presence of the modality enforces the requirement that the head of the stream is available
immediately, but the tail may only be accessed after some productive work has taken place. This
allows us to \eg~construct an infinite stream of ones:
\[
  \mathsf{inf\_stream\_of\_ones} \defeq \LobV{s}{\Cons(1, s)}
\]
Unlike the ordinary type of streams, $\Str_A$ does not behave like a coinductive type: we may only
define \emph{causal} operations, which excludes useful functions (e.g. the $\Tail$ function). In
order to regain coinductive behaviour, \cite{clouston:2015} introduced the \emph{always}
modality~$\Box$, an idempotent comonad for which
\begin{equation*}
  \Box \Later A \Equiv \Box A. \TagEq[$\ast$]\label{eq:guarded:force}
\end{equation*}
Combining $\Box$ and $\Later$ in the same system has proven tricky. Previous work has used
\emph{delayed substitutions}~\cite{bizjak:2016}, or replaced $\Box$ with \emph{clock
quantification}~\cite{atkey:2013,moegelberg:2014,bizjak:2015,bahr:2017}. Neither solution is
entirely satisfactory: the former poses serious implementation and usability issues, and the latter
does not enjoy the conceptual simplicity of a single modality. We will show that \MTT{} enables us
to effortlessly combine the two modalities whilst satisfying \eqref{eq:guarded:force}.

To encode guarded recursion inside \MTT{}, we must
\begin{enumerate}
\item construct a mode theory that induces an applicative functor $\Later$ and an idempotent comonad
  $\Box$ satisfying \eqref{eq:guarded:force},
\item construct the \emph{intended model} of $\MTT$ with this mode theory, \ie{} a model where
  these modalities are interpreted in the standard way~\cite{birkedal:2012}, and
\item include L\"ob induction as an axiom, and use it to reason about guarded streams.
\end{enumerate}

\subsection{A guarded mode theory}

\begin{figure}
  \[
    \begin{tikzpicture}[on grid, node distance = 2.5cm]
      \node (T) {$t$};
      \node (S) [right = of T] {$s$};
      \path[->] (T) edge[looseness = 15, out = 210, in = 150] node[left] {$\ell$} (T);
      \path[->] (S) edge[bend left] node[below] {$\delta$} (T);
      \path[->] (T) edge[bend left] node[above] {$\gamma$} (S);
      \node (E) [right = 2.5cm of S]
      {$\begin{aligned}
          \delta \circ \gamma &\le 1 & 1 &= \gamma \circ \delta\\
          1 &\le \ell & \gamma &= \gamma \circ \ell
        \end{aligned}$};
    \end{tikzpicture}
  \]
  \caption{The `adjoint bowling pin' $\ModeGuarded$: a mode theory for guarded recursion.}
  \label{fig:applications:guarded-mode-theory}
\end{figure}

We define $\ModeGuarded$ to be the mode theory generated by the graph and equations of Fig.
\ref{fig:applications:guarded-mode-theory}. We require that $\ModeGuarded$ be poset-enriched,
i.e.~that there be at most one 2-cell between a pair of modalities $\mu, \nu$, which we denote by
$\mu \le \nu$ when it exists. Consequently, we need not state any coherence equations between
2-cells.

Unlike prior guarded type theories, $\ModeGuarded$ has \emph{two modes}. We will think of elements
of $s$ as being \emph{constant types and terms}, while types in $t$ may \emph{vary over time}.
Observe that we can thence construct a composite modality $b \defeq \delta \circ \gamma$. Moreover,
this modality is idempotent, for $b \circ b = \delta \circ \gamma \circ \delta \circ \gamma = \delta
\circ \gamma = b$. We will prove in Section~\ref{sec:guarded-recursion:internally} that
\begin{lem}
  $\Modify[b]$ is an idempotent comonad and $\Modify[\ell]$ is an applicative functor.
\end{lem}

\subsection{Decomposing the standard model}

The above mode theory arises from a careful and informative decomposition of the standard model of
guarded recursion, namely the \emph{topos of trees} $\PSH{\omega}$, along with the later and always endomodalities.

The topos of trees consists of presheaves over the natural numbers, seen as a poset with the usual
order. An element $x_n \in X(n)$ of a presheaf $X : \PSH{\omega}$ represents an element computed
after $n$ steps of computation. The restriction maps $r_n : X(n+1) \to X(n)$ trim an element
computed after $n+1$ steps to its form at the preceding moment in time. The canonical example is
given by $X(n) \defeq \{\text{streams of length $n$}\}$, where $r_n$ deletes the last element of a
stream of length $n+1$. The later and always endomodalities are given by delaying the computation by
one step, and by taking global sections (total elements):
\begin{align*}
  (\Later X)(n) &\defeq {
    \begin{cases}
      \{\ast\} &\text{if $n = 0$} \\
      X(n-1)   &\text{if $n > 0$}
    \end{cases}
  }  &
  (\Box X)(n) &\defeq \Hom[\PSH{\omega}]{1}{X}
\end{align*}
To arrive at the mode theory above, one must notice that the comonad $\Box$ results in a
\emph{constant} presheaf, namely one which consists of the same set at each time. We can thus
decompose it into the adjunction
\begin{equation}
  \begin{tikzpicture}[on grid, node distance = 2.5cm, baseline=(current  bounding  box.center)]
    \node (T) {$\PSH{\omega}$};
    \node (S) [right = of T] {$\SET$};
    \node (Z) [right = 1.3cm of T, rotate=+90] {$\Adjoint$};
    \path[->] (T) edge[looseness = 10, out = 210, in = 150] node[left] {$\Later$} (T);
    \path[->] (T) edge[bend left] node[above] {$\GSec$} (S);
    \path[->] (S) edge[bend left] node[below] {$\Disc$} (T);
  \end{tikzpicture}
  \label{diag:guarded:model}
\end{equation}
$\GSec$ maps $X : \PSH{\omega}$ to the set of its global sections $\Hom[\PSH{\omega}]{1}{X}$, and
$\Delta$ maps a set $S$ to the constant presheaf $(\Delta S)(n) \defeq S$. It is well-known that
$\Delta \Adjoint \Gamma$, and `always' is given by the induced comonad $\Box \defeq \Disc \circ
\GSec$. This explains the provenance of the two modes in Figure
\ref{fig:applications:guarded-mode-theory}: $s$ stands for \emph{sets}, and $t$ for \emph{timed
sets}, i.e. presheaves over $\omega$.

We want to bootstrap \eqref{diag:guarded:model} into a model of \MTT{}. We will do so by leveraging
an impressive sequence of facts:
\begin{itemize}
  \item Both categories in \eqref{diag:guarded:model} are presheaf categories, and hence models of
  \MLTT{}: see Section~\ref{sec:presheaves}.
  \item Every functor in \eqref{diag:guarded:model} is a right adjoint.
  \item The corresponding left adjoints are introduced by precomposition, and hence can easily  be
  arranged into a modal context structure for the mode theory $\ModeGuarded$ as per
  Section~\ref{def:semantics:context-structure}.
  \item Hence, by uniqueness of adjoints the functors in \eqref{diag:guarded:model} are induced by
  right Kan extension. Consequently, they can be bootstrapped into dependent right adjoints, by
  Lemma \ref{lem:constructing-dra:direct-image}.
  \item Therefore, by Theorem \ref{thm:dra:dra}, this data yields a model of \MTT{} with mode theory $\ModeGuarded$.
\end{itemize}
Let us elaborate on this chain of reasoning.
First, we identify the category $\SET$ and the category $\PSH{1}$ of presheaves over the terminal
category. Second, we construct the two left adjoints. As $\omega$ has an initial object $0$, we
obtain a left adjoint to the discrete functor $\Disc$, given by
\[
  \Pi_0(X) \defeq X(0)
\]
It is easy to see that $\Pi_0 \Adjoint \Disc$: by naturality at the unique morphism $0 \leq n$ we
see that any $\alpha : X \To \Disc S$ is fully determined by the component $\alpha_0 : X(0) \to S$.
Furthermore, recall from the work of \cite{birkedal:2012} that the later modality $\Later$ has a
left adjoint $\Earlier : \PSH{\omega} \to \PSH{\omega}$ (pronounced `earlier'), given by
\[
  (\Earlier X)(n) \defeq X(n+1)
\]
It remains to show that the three left adjoints---$\Pi_0$, $\Delta$, and $\Earlier$---are given by
precomposition. We define three monotone functions between the posets $1 \defeq \{\ast\}$ and
$\omega$:
\begin{alignat*}{19}
    K_0         :\ &1\      &&\to\ &&\omega\ \qquad
  & {!}_\omega\ :\ &&\omega\ &&\to\ &&1\      \qquad
  & l\ :\     &&\omega\ &&\to\ &&\omega \\
     &\ast\ &&\mapsto\ &&0\    \qquad
  & &&n\    &&\mapsto\ &&\ast\ \qquad
  & &&n\    &&\mapsto\ &&n+1
\end{alignat*}
Identifying $\SET$ with $\PSH{1}$, we see that
\begin{align*}
  \Pi_0  &= \Pre{K_0}{} : \PSH{\omega} \to \SET &
  \Delta &= \Pre{{!}_\omega}{} : \SET \to \PSH{\omega} &
  \Earlier &= \Pre{l}{} : \PSH{\omega} \to \PSH{\omega}
\end{align*}
Moreover, we trivially have the following pointwise equations and inequalities:
\begin{align*}
  \text{id}_\omega &\leq l
  &
  K_0 \circ {!}_\omega &\leq \text{id}_\omega
  &
  \text{id}_1 &= {!}_\omega \circ K_0
  &
  {!}_\omega &= {!}_\omega \circ l
\end{align*}
Seeing posets as categories, pointwise inequalities are simply natural transformations between
monotone maps. By feeding them into the strict 2-functor $\Pre{(-)}{} : \Coop{\CAT} \to \CAT$, we
are able to define a strict 2-functor $\InterpModal{-} : \Coop{\ModeGuarded} \to \CAT$ which maps
\begin{alignat*}{5}
  &\gamma\ &&: t \to s\quad &&\longmapsto\quad &&\InterpModal{\gamma} = \Delta\ &&:\ \SET \to \PSH{\omega} \\
  &\delta\ &&: s \to t\quad &&\longmapsto\quad &&\InterpModal{\delta} = \Pi_0\  &&:\ \PSH{\omega} \to \SET \\
  &\ell\   &&: t \to t\quad &&\longmapsto\quad &&\InterpModal{\ell} = \Earlier\ &&:\ \PSH{\omega} \to \PSH{\omega}
\end{alignat*}
This fully specifies the modal context structure, which consists of left adjoints. Each of these
left adjoints is given by precomposition. Thus, the unique corresponding right adjoint is given by
right Kan extension (see Section~\ref{sec:presheaves}). Hence, by Lemma
\ref{lem:constructing-dra:direct-image} and Theorem \ref{thm:dra:dra},
\begin{thm}
  There is a model of \MTT{} with mode theory $\ModeGuarded$, interpreting $s$ as $\SET$ and $t$ as
  $\PSH{\omega}$. Furthermore, this model interprets $\delta$ by the dependent right adjoint arising
  from $\CComp \Adjoint \Disc$, $\gamma$ by $\Disc \Adjoint \GSec$, and $\ell$ by $\Earlier \Adjoint
  \Later$.
\end{thm}

\begin{rem}
  This mode theory is a poset-enriched category. As a result, the key substitutions are unique: for
  any $\mu, \nu$ there is at most one substitution $\IsSb[\LockSb{\Gamma}]{\Key{\nu \le
  \mu}{\Gamma}}{\LockSb{\Gamma}<\nu>}$. This property means that we can elide them without
  ambiguity. However, this may sometimes make type-checking on pen-and-paper difficult, so we employ
  a simplified notation: we will write $A^{\nu \le \mu}$ or $M^{\nu \le \mu}$ for the application of
  the unique key substitution $\nu \le \mu$ in context $\LockSb{\Gamma}$. For instance, given a type
  $\LockCx{\Gamma}<1> = \IsTy{A}[1]<t>$ we can form the type $\IsTy[\LockCx{\Gamma}<\ell>]{A^{1 \le
  \ell}}[1]<t>$, and hence the type $\IsTy{\Modify[\ell]{A^{1 \le \ell}}}[1]<t>$.
\end{rem}

\subsection{Guarded recursion, internally}
\label{sec:guarded-recursion:internally}

Given the model that we constructed above, we feel perfectly justified in defining the following
shorthands within \MTT{}:
\begin{align*}
  \Box A &\defeq \Modify[b]{A} &
  \Later A &\defeq \Modify[\ell]{A} &
  \GSec A &\defeq \Modify[\gamma]{A} &
  \Disc A &\defeq \Modify[\delta]{A}
\end{align*}
where $b \defeq \delta \circ \gamma$. The aim of this section is to show that $\MTT{}$ equipped with $\ModeGuarded$ and these shorthands can be used to reason about guarded recursion. In particular, we will show that this is strict improvement on previous solutions, by establishing that
\begin{enumerate}
  \item When restricted to mode $s$, the type theory is simply standard Martin-L{\"o}f Type Theory.
  \item The modalities on mode $t$ give rise to the standard modalities and operations of
  Guarded Type Theory~\cite{bizjak:2016} inside the type theory.
\end{enumerate}

First, we wish to show that if we restrict ourselves to endomodalities $\mu \in \Hom{s}{s}$ from
sets to sets, the type theory is just \MLTT{}. Looking at Fig.
\ref{fig:applications:guarded-mode-theory} as a finite state machine, we see that all loops on $s$
are of the form $\gamma \circ \ell^n \circ \delta$, and the equations of $\ModeGuarded$ allow us to
prove that they are all equal to the identity $1_s$. It follows that $\Modify{A} \simeq A$. Finally,
as there is no non-trivial 2-cell $1_s \To 1_s$ the variable rule reduces to
\[
  \inferrule{
    \mu \in \Hom{s}{s}\\
    \IsCx{\Gamma}<s>\\
    \IsTy{A}<s>\\
    (\DeclVar{x}{A}<\mu>) \in \Gamma
  }{
    \IsTm{x}{A}<s>
  }
\]
which is essentially the usual variable rule of \MLTT{}.

Second, we use the combinators of Section~\ref{sec:programming-in-mtt:comonads} to prove that $\Box$ is an
idempotent comonad.
\[
  \arraycolsep=1.4pt
  \begin{array}{lcllcl}
    \Dup_A    &:&  \Box A \xrightarrow{\simeq} \Box \Box A \qquad\qquad &\Unbox_A &:& \Box A \to A^{b \le 1} \\
    \Dup_A(x) &\defeq& \MComp*{x}{b}{b}   \qquad\qquad &\Unbox_A(x) &\defeq& \Triv*{\Coe{x}{b}{1}}
  \end{array}
\]
Recall the $K$ operator $\ZApp{-}{-}{b} : \Box(A \to B) \to \Box A \to \Box B$ for the modality $b$,
which was defined in Section~\ref{sec:programming-in-mtt:combinators}. Writing $ \Always(M) \defeq
\MkBox[b]{M}$, the claim that $\Box$ is an internal idempotent comonad amounts to defining terms of
the following types.
\begin{align}
  (x : \Box A) &\to \Id{\Box A}{x}{\ZApp{\Always(\Unbox)}{\Dup(x)}{b}}
  \label{eq:guarded-recursion:comonad1}\\
  (x : \Box A) &\to \Id{\Box A}{x}{\Unbox(\Dup(x))}
  \label{eq:guarded-recursion:comonad2}\\
  (x : \Box A) &\to \Id{\Box \Box \Box A}{\Dup(\Dup(x))}{\ZApp{\Always(\Dup)}{\Dup(x)}{b}}
  \label{eq:guarded-recursion:comonad3}
\end{align}
These can be constructed by unfolding and modal induction on $x : \Box A$.

The $K$ operator $\ZApp{-}{-}{\ell} : \Later(A \to B) \to \Later A \to \Later B$ for the modality
$\ell$ almost proves that $\Later$ is an applicative functor. It remains to show that $\Later$ is
pointed:
\[
  \arraycolsep=1.4pt
  \begin{array}{lcl}
    \Next_A    &:&      A \to \Later A \\
    \Next_A(x) &\defeq& \Coe{\Triv{x}}{1}{\ell}
  \end{array}
\]

Next, we show the defining equivalence \eqref{eq:guarded:force}. We calculate that $b \circ \ell
\defeq \delta \circ \gamma \circ \ell = \delta \circ \gamma \defeq b$, and hence that the equivalence is a
corollary of a combinator given in Section~\ref{sec:programming-in-mtt:combinators}:
\[
  \arraycolsep=1.4pt
  \begin{array}{lcl}
    \Force_A(x)    &:&   \Box \Later A \xrightarrow{\simeq} \Box A \\
    \Force_A(x) &\defeq& \MComp*{x}{b}{\ell}
  \end{array}
\]
As a sanity check, we can compute that the following composite is the identity:
\[
  \begin{tikzpicture}[on grid, node distance = 5cm]
    \node (A) {$\Box A$};
    \node (B) [right = of A] {$\Box \Later A$};
    \node (C) [right = of B] {$\Box A$};
    \path[->] (A) edge node[above] {$\ZApp{\Always(\Next)}{-}{}$} (B);
    \path[->] (B) edge node[above] {$\Force$} (C);
  \end{tikzpicture}
\]
The calculation is as follows:
\begin{align*}
  &\MComp{\ZApp{\MkBox[b]{\Coe{\Triv{-}}{1}{\ell}}}{x}{}}{b}{\ell}
  &&\text{by induction, suppose $x = \MkBox[b]{y}$}\\
  &= \MComp{\ZApp{\MkBox[b]{\Coe{\Triv{-}}{1}{\ell}}}{\MkBox[b]{y}}{}}{b}{\ell} &&\\
  &= \MComp{\MkBox[b]{\Coe{\Triv{y}}{1}{\ell}}}{b}{\ell}  &&\\
  &= \MComp{\MkBox[b]{\MkBox[\ell]{y}}}{b}{\ell} &&\\
  &= \MkBox[b]{y} &&\text{as $b \circ \ell = b$}\\
  &= x
\end{align*}

The only thing that remains is to add \emph{L{\"o}b induction}. This is a \emph{modality-specific
operation} that cannot be expressed in the mode theory, so we must add it as an axiom:
\begin{mathpar}
  \inferrule{
    \IsCx{\Gamma}<t>\\
    \IsTy{A}[1]<t>
  }{
    \IsTm{\Lob}{(\Later A^{1 \le \ell} \to A) \to A}<t>
  }
  \and
  \inferrule{
    \IsCx{\Gamma}<t>\\
    \IsTy{A}[1]<t>\\
    \IsTm{M}{\Later A^{1 \le \ell} \to A}<t>
  }{
    \EqTm{\Lob(M)}{M(\Next(\Lob(M)))}{(\Later A^{1 \le \ell} \to A) \to A}<t>
  }
\end{mathpar}
Notice that these rules are only added in mode $t$, as they only admit an interpretation in
$\PSH{\omega}$ \cite[\S 2]{birkedal:2012}. Unfortunately, these ad-hoc additions mean that the
canonicity theorem of Section~\ref{sec:canonicity} no longer applies.

\subsection{Reasoning about Streams}
\label{sec:guarded-recursion:streams}

We now put \MTT{} to work: we will use it to reason about infinite streams defined by guarded
recursion. We will demonstrate that the rules and axioms given in
Section~\ref{sec:guarded-recursion:internally} suffice to carry out coinductive constructions. In
particular, we will reproduce an example of \cite{bizjak:2016}: we will show that $\ZipWith(f)$ on
a coinductive stream is commutative whenever $f$ itself is.

In order to simplify our working, we will swap the intensional equality type $\Id{A}{M}{N}$ with an
\emph{extensional identity type} $\Eq{A}{M}{N}$. This has the same introduction rule, but its
elimination is replaced by the usual \emph{equality reflection rule}
\[
  \inferrule{
    \IsTm{P}{\Eq{A}{M_0}{M_1}}
  }{
    \EqTm{M_0}{M_1}{A}
  }
\]
This is straightforwardly interpreted in the model, as both modes are mapped to presheaf toposes.
The switch to extensional equality is not strictly necessary: we could carry out the following
calculations with intensional identity, at the price of significantly more verbose terms. Moreover,
the need for the function extensionality axiom would arise. However, adding L{\"o}b induction has
already ensured that type-checking is undecidable, so nothing of value is lost by making the switch
to extensional type theory for these examples.

We begin with a simple reasoning principle. Eliding $(-)^{1 \leq \ell}$ annotations:
\begin{lem} \label{lem:guarded:later-eq}
  $(A : \Uni)(x, y : \Dec{A}) \to \Later\Eq{\Dec{A}}{x}{y} \to \Eq{\Later \Dec{A}}{\Next(x)}{\Next(y)}$
\end{lem}
\begin{proof}
  Suppose $x , y : \Dec{A}$ and $p : \Later\Eq{\Dec{A}}{x}{y}$; to show $\Next(x) = \Next(y) :
  \Later \Dec{A}$. By congruence and the elimination rule for the modality, it suffices to prove $x
  = y : \Dec{A}$ in the locked context $\LockCxV{A : \Uni, x : \Dec{A} , y : \Dec{A},
  \DeclVar{p}{\Eq{\Dec{A}}{x}{y}}<\ell>}<\ell>$. But by the variable rule we have $p :
  \Eq{\Dec{A}}{x}{y}$ in this context, and hence $x = y : \Dec{A}$.
\end{proof}

This can be used to prove \emph{internally} that guarded fixed points are unique.
\begin{thm}\
  \label{thm:guarded-recursion:unique-fixed-points}
  $\Lob(M)$ is the unique guarded fixed point of $M : \Later \Dec{A} \to \Dec{A}$,
  \ie{}
  \[
    (A : \Uni)(x : \Dec{A}) \to \Eq{\Dec{A}}{M(\Next(x))}{x} \to \Eq{\Dec{A}}{\Lob(M)}{x}
  \]
\end{thm}
\begin{proof}
  Suppose $A : \Uni$; to show $(x : \Dec{A}) \to \Eq{\Dec{A}}{M(\Next(x))}{x} \to
  \Eq{\Dec{A}}{\Lob(M)}{x}$ by L\"ob induction. Thus, assume that
  \[
    f : \Later ((x : \Dec{A}) \to \Eq{\Dec{A}}{M(\Next(x))}{x} \to \Eq{\Dec{A}}{\Lob(M)}{x})
  \]
  If $x : A$ and $p : \Eq{\Dec{A}}{M(\Next(x))}{x}$, we calculate that
  \begin{align*}
    \Lob(M)
    &= M(\Next(\Lob(M))) && \text{unfolding rule for $\Lob$}\\
    &= M(\Next(x)) && \text{by Lemma \ref{lem:guarded:later-eq} on $\ZApp{\ZApp{f}{\Next(x)}{}}{\Next(p)}{} : \Later \Eq{\Dec{A}}{\Lob(M)}{x}$}\\
    &= x && \text{by $p$}
  \end{align*}
  Thus, this type is inhabited by the term $\Lam[A]{\Lob(\Lam[f]{\Lam[x]{\Lam[p]{\Refl{x}}}})}$.
\end{proof}

We can also use the L{\"o}b operator on the universe to form \emph{guarded recursive types}. For
example, streams can be defined by\footnote{We denote modalities and their counterparts on the
universe by the same notation. For example, we may write $\Delta A$ to mean the type
$\IsTy{\Modify[\delta]{A}}[1]$ whenever $\IsTy[\LockCxV{\Gamma}<\delta>]{A}[1]$, but also to mean
the term $\IsTm{\Enc{\Modify[\delta]{\Dec{A}}}}{\Uni}$ whenever
$\IsTm[\LockCxV{\Gamma}<\delta>]{A}{\Uni}$.}
\begin{alignat*}{3}
  &\Str    &&:\      &&\Uni \to \Uni \Mute{{}\mathop{@} s}\\
  &\Str(A) &&\defeq\ &&\GSec(\Lob(\Lam[X]{\Disc A \times \Open{X}[Y]{\Later Y}<\ell>[1]}))
\end{alignat*}

$\Str$ maps a constant set, \ie{} a type $A \Mute{{}\mathop{@} s}$, to the type of streams over $A$,
which is again a constant set. This is done by first defining a timed set
\[
  \IsTm[\LockCxV{\DeclVar{A}{\Uni}<1>}<\gamma>]{
    \Str'(A) \defeq \Lob(\Lam[X]{\Disc A \times \Open{X}[Y]{\Later Y}<\ell>[1]})
  }{\Uni}<t>
\]
$\Str'(A)$ is defined by L\"ob induction: assuming $\DeclVar{X}{\Later \Uni^{1 \leq \ell}}<1>$ we
must define an element of the timed universe. This is given as the product of
\begin{itemize}
  \item the set $A \Mute{{}\mathop{@} s}$, considered as a constant-everywhere timed set $\Delta A
  \Mute{{}\mathop{@} t}$;
  \item a guarded recursive call, which represents the rest of the stream.
\end{itemize}
Recalling that $\Uni^{1 \leq \ell} = \Uni$, the second component is given by modal elimination.
Nevertheless, it is not immediate that the first component type-checks: we must show that
\[
  \IsTm[\LockCxV{\LockCxV{\DeclVar{A}{\Uni}<1>}<\gamma>}<\delta>]{A}{\Uni}<s>
\]
But $\gamma \circ \delta = 1$, so the context is equal to $\LockCxV{\DeclVar{A}{\Uni}<1>}<1>$ and we
can use $A$. Unfolding the guarded fixed point, we have that
\[
  \Str'(A) = \Disc{A} \times \Later \Str'(A) : U \Mute{{}\mathop{@} t}
\]
We apply $\Gamma$ to `totalize' this into the constant set $\Str(A) \Mute{{}\mathop{@} s}$ of
guarded streams.

Even though not immediately obvious, there is a serious advantage in expressing this definition in a
way that spans two modes. In previous work~\cite{bizjak:2016} the stream type $\Str(A)$ was
coinductive only if $A$ was provably a `constant set,' \ie{} if $A \Equiv \Box A$. Theorems about
streams had to carry around a proof of this equivalence. In our case, defining $\Str(A)$ at the mode
$s$ of constant sets automatically ensures that. Hence, $\Str(A)$ is equivalent to the familiar
definition, but we no longer need to propagate proofs of constancy.

$\Str(A)$ supports the following operations:
\[
  \arraycolsep=1.4pt
  \begin{array}{lcl}
    \Cons &:& (A : \Uni) \to \Dec{A} \to \Dec{\Str(A)} \to \Dec{\Str(A)} \Mute{{}\mathop{@} s}\\
    \Cons_A(h, t) &\defeq& \Open{t}[t']{\MkBox[\gamma]{\Pair{\MkBox[\delta]{h}}{\Next(t')}}}<\gamma>[1]\\[0.3cm]
    \Head &:& (A : \Uni) \to \Dec{\Str(A)} \to \Dec{A} \Mute{{}\mathop{@} s} \\
    \Head_A(s) &\defeq& \Open{s}[s']{\Triv*{\MComp{\MkBox[\gamma]{\Proj[0]{s'}}}{\gamma}{\delta}}}<\gamma>[1]\\[0.3cm]
    \Tail &:& (A : \Uni) \to \Dec{\Str(A)} \to \Dec{\Str(A)} \Mute{{}\mathop{@} s} \\
    \Tail_A(s) &\defeq& \Open{s}[s']{\MComp{\MkBox[\gamma]{\Proj[1]{s'}}}{\gamma}{\ell}}<\gamma>[1]
  \end{array}
\]
Those familiar with prior work on guarded streams may be surprised by the type of $\Tail$. The
expected definition would be
\[
  \Tail_A(s) \stackrel{?}{=} \Open{s}[s']{\MkBox[\gamma]{\Proj[1]{s'}}}<\gamma>[1]
\]
This term has type $\Dec{\Str(A)} \to \GSec(\Later \Dec{\Str'(A)})$. However, in our case the
$\GSec$ modality is sufficiently strong to ``absorb'' this extra $\Later$: the equality $\gamma
\circ \ell = \gamma$ induces an equivalence $\GSec \circ \Later \simeq \GSec$, which we use to
obtain the version given above. This small difference is crucial: it will internally make $\Str(A)$
into a final coalgebra!

\begin{lem}
  These operations satisfy the expected $\beta$ and $\eta$ laws, \ie{}
  \begin{enumerate}
  \item $(h : \Dec{A})(t : \Dec{\Str(A)}) \to \Eq{\Dec{A}}{\Head_A(\Cons_A(h, t))}{h} \Mute{{}\mathop{@} s}$
  \item $(h : \Dec{A})(t : \Dec{\Str(A)}) \to \Eq{\Dec{\Str(A)}}{\Tail_A(\Cons_A(h, t))}{t} \Mute{{}\mathop{@} s}$
  \item $(h : \Dec{A})(t : \Dec{\Str(A)}) \to \Eq{\Dec{\Str(A)}}{s}{\Cons_A(\Head_A(s), \Tail_A(s))}\Mute{{}\mathop{@} s}$
  \end{enumerate}
\end{lem}
\begin{proof}
  We prove (2), the other two being similar. If $h : \Dec{A}$ and $t : \Dec{\Str(A)}$, note that $\Dec{\Str(A)}$ is a type of the form $\Gamma(-)$, and calculate that
  \begin{align*}
    &\ \Tail_A(\Cons_A(h, t)) && \\
    &= \Tail_A(\Cons_A(h, \MkBox[\gamma]{t'}))
      && \text{write $t = \MkBox[\gamma]{t'}$ by modal induction}\\
    &= \Tail_A(\MkBox[\gamma]{\Pair{\MkBox[\delta]{h}}{\Next(t')}}) && \\
    &= \MComp{\MkBox[\gamma]{\Proj[1]{\Pair{\MkBox[\delta]{h}}{\Next(t')}}}}{\gamma}{\ell} &&\\
    &= \MComp{\MkBox[\gamma]{\Next(t')}}{\gamma}{\ell} && \\
    &= \MkBox[\gamma]{t'}
      &&\text{as $\gamma \circ \ell = \gamma$}\\
    &= t && \qedhere
  \end{align*}
\end{proof}

\begin{thm} \label{thm:guarded:final-coalgebra}
  $\Str(A)$ is the final coalgebra for $\Lam[X]{\Dec{A} \times X} : \Uni \to \Uni \Mute{{}\mathop{@} s}$.
\end{thm}
\begin{proof}
  Given $A : \Uni$ we define a coalgebra $\Uncons : \Str(A) \to (\Dec{A} \times \Str(A)) \Mute{{}\mathop{@} s}$ by
  \[
    \Uncons(s) \defeq \Pair{\Head_A(s)}{\Tail_A(s)}
  \]
  To show finality, suppose $c : B \to \Dec{A} \times B \Mute{{}\mathop{@} s}$ is another
  coalgebra. We define a function $f : B \to \Str(A) \Mute{{}\mathop{@} s}$ by
  \[
    \arraycolsep=1.4pt
    \begin{array}{lcl}
      f' &:& \Disc B \to \Dec{\Str'(A)} \Mute{{}\mathop{@} t}\\
      f' &\defeq& \Lob(\Lam[f'', x]{\Open{x}[x']{\Pair{h}{t}}<\delta>[1]})\\
         &&\text{where }       h = \MkBox[\delta]{\Proj[0]{c(x')}}\\
         &&\text{and}\quad\;\, t = {\ZApp{f''}{\Next(\MkBox[\delta]{\Proj[1]{c(x')}})}{\ell}}\\[0.3cm]
      f &:& B \to \Dec{\Str(A)} \Mute{{} \mathop{@} s}\\
      f(x) &\defeq& \MkBox[\gamma]{f'(\MkBox[\delta]{x})}
    \end{array}
  \]
  This is a morphism of coalgebras: for any $x : B$ we have
  \begin{align*}
    \Uncons(f(x))
    &= \Pair{\Head_A(f(x))}{\Tail_A(f(x))}\\
    &= \Pair{\Proj[0]{c(x)}}{\Tail_A(f(x))}\\
    &= \Pair{\Proj[0]{c(x)}}{\MComp{\MkBox[\gamma]{\Proj[1]{f'(x)}}}{\gamma}{\ell}}\\
    &= \Pair{\Proj[0]{c(x)}}{\MComp{\MkBox[\gamma]{\ZApp{\Next(f')}{\Next(\MkBox[\delta]{\Proj[1]{c(x)}})}{\ell}}}{\gamma}{\ell}}\\
    &= \Pair{\Proj[0]{c(x)}}{f(\Proj[1]{c(x)})}
  \end{align*}
  Finally, we must show that $f$ is the unique coalgebra morphism. Suppose we are given
  $g : B \to \Dec{\Str(A)} \Mute{{}\mathop{@} s}$ which also satisfies
  $\Uncons(g(x)) = \Pair{\Proj[0]{c(x)}}{g(\Proj[1]{c(x)})}$. We `shift' this definition to timed sets, by defining
  \[
    \arraycolsep=1.4pt
    \begin{array}{lcl}
    \hat{g}    &:& \Disc{B} \to \Dec{\Str'(A)} \Mute{{}\mathop{@} t} \\
    \hat{g}(x) &\defeq& \Coe{\ZApp{\MkBox[\delta]{g} }{x}{\delta}}{\delta \circ \gamma}{1}
    \end{array}
  \]
  It suffices to show that $\hat{g} = f' \Mute{{}\mathop{@} t}$, and we do so by L\"ob induction and
  function extensionality. Assume $p : \Later \Eq{}{\hat{g}}{f'}$, and $x : \Disc B$. To prove
  $\hat{g}(x) = f'(x) : \Disc B \times \Later \Str'(A)$ it suffices to show componentwise equality.
  By modal induction write $x = \MkBox[\delta]{y}$ for $y : B$.

  First, we have that $\Proj[0]{f'(\MkBox[\delta]{y})} = \MkBox[\delta]{\Proj[0]{c(y)}}$ by the definition of $f'$. On the other hand, we have that
  \[
    \Proj[0]{\hat{g}(\MkBox[\delta]{y})}
    = \Proj[0]{\Coe{\MkBox[\delta]{g(y)}}{\delta \circ \gamma}{1}}
    = \Proj[0]{g_x} : \Disc B \Mute{{}\mathop{@} t}
  \]
  where we have used modal induction to write $g(y) = \MkBox[\gamma]{g_x}$. That  $g$ is a coalgebra
  morphism implies that $\Head(g(y)) = \Proj[0]{c(y)}$. If we now use modal induction to write
  $\Proj[0]{g_x} = \MkBox[\delta]{b}$ for $b : B$ and unfold the definition of $\Head$, we obtain $b
  = \Proj[0]{c(y)}$, so $\Proj[0]{g_x} = \MkBox[\delta]{b} = \MkBox[\delta]{\Proj[0]{c(y)}}$, which shows
  that the two first components are equal.

  For the second component, we compute that
  \begin{align*}
    &\Proj[1]{f'(\MkBox[\delta]{y})} & \\
    &= \ZApp{\Next(f')}{\Next(\MkBox[\delta]{\Proj[1]{c(y)}}))}{\ell} & \\
    &= \Next(f'(\MkBox[\delta]{\Proj[1]{c(y)}})) & \\
    &= \Next(\hat{g}(\MkBox[\delta]{\Proj[1]{c(y)}}))
      & \text{using $p$ through Lemma \ref{lem:guarded:later-eq}}\\
    &= \Next(\Coe{\MkBox[\delta]{g(\Proj[1]{c(x)})}}{\delta \circ \gamma}{1}) & \\
    &= \Next(\Coe{\MkBox[\delta]{\Tail(g(y))}}{\delta \circ \gamma}{1})
      & \text{as $g$ is a coalgebra morphism} \\
    &= \Proj[1]{\Coe{\MkBox[\delta]{g(y)}}{\delta \circ \gamma}{1}}
      & \text{lemma} \\
    &= \Proj[1]{\hat{g}(x)}
  \end{align*}
  The lemma referred to above is the fact that for any $s : \Str(A)$ it is the case that
  \[
    \Next(\Coe{\MkBox[\delta]{\Tail(s)}}{\delta \circ \gamma}{1})
    = \Proj[1]{\Coe{\MkBox[\delta]{s}}{\delta \circ \gamma}{1}}
  \]
  which can be shown by a series of modal inductions.
\end{proof}

We conclude this section by showing how to use these mechanisms in order to prove properties of
coinductive programs. Specifically, we will replicate a proof from \cite{bizjak:2016} which shows
that the \texttt{zipWith} operator on streams preserves commutativity. Let
\[
  \arraycolsep=1.4pt
  \begin{array}{lcl}
    \ZipWith' &:& \Disc (\Dec{A} \to \Dec{B} \to \Dec{C}) \to \Dec{\Str'(A)} \to \Dec{\Str'(B)} \to \Dec{\Str'(C)}\\
    \ZipWith'(f) &\defeq&%
    \Lob(%
      \Lam[r]{
        \Lam[x, y]{%
          \Pair{
            \ZApp{\ZApp{f}{\Proj[0]{x}}{\delta}}{\Proj[0]{y}}{\delta}
          }{
            \ZApp{\ZApp{r}{\Proj[1]{x}}{\ell}}{\Proj[1]{y}}{\ell}
          }
        }
      }
    )\\[0.3cm]
    \ZipWith &:& (\Dec{A} \to \Dec{B} \to \Dec{C}) \to \Dec{\Str(A)} \to \Dec{\Str(B)} \to \Dec{\Str(C)}\\
    \ZipWith(f) &\defeq& \Lam[x, y]{\ZApp{\ZApp{\MkBox[\gamma]{\ZipWith'(\MkBox[\delta]{f})}}{x}{\gamma}}{y}{\gamma}}
  \end{array}
\]
\begin{rem}
  Take note of a useful pattern for programming with guarded recursion, which is visible both here
  and in the proof of Theorem \ref{thm:guarded:final-coalgebra}. We first define an auxiliary
  function in mode $t$, which uses L{\"ob} induction. The main function itself is then just a thin
  wrapper which `corrects' that with the appropriate modalities and modal combinators.
\end{rem}

\begin{thm}
  If $f$ is commutative then $\ZipWith(f)$ is commutative. That is, given $A, B : \Uni$ and
  $f : \Dec{A} \to \Dec{A} \to \Dec{B}$ there is a term of the following type:
  \begin{align*}
    ((x, y &: \Dec{A}) \to \Eq{\Dec{B}}{f(x, y)}{f(y, x)}) \to{}\\
    &(s, t : \Dec{\Str(A)}) \to \Eq{\Dec{\Str(B)}}{\ZipWith(f, s, t)}{\ZipWith(f, t, s)}
  \end{align*}
\end{thm}
\begin{proof}
  Suppose $e : (x, y : \Dec{A}) \to \Eq{\Dec{B}}{f(x, y)}{f(y, x)}$ and $s, t : \Dec{\Str(A)}$. We
  wish to show that $\ZipWith(f, s, t) = \ZipWith(f, t, s)$. By the definition of $\ZipWith$, it is sufficient to prove that for any $u, v : \Dec{\Str'(A)}$ we have
  \[
    \ZipWith'(\MkBox[\delta]{f}, u, v)
    = \ZipWith'(\MkBox[\delta]{f}, v, u)
  \]
  In turn, it suffices to show that
  \[
    \Lob(F_0) = \Lob(F_1)
  \]
  where
  \begin{align*}
    F_0 &\defeq \Lam[r]{\Lam[x, y]{
      \Pair{
        \ZApp{\ZApp{\MkBox[\delta]{f}}{\Proj[0]{x}}{\delta}}{\Proj[0]{y}}{\delta}
      }{
        \ZApp{\ZApp{r}{\Proj[1]{\Alert{x}}}{\ell}}{\Proj[1]{\Alert{y}}}{\ell}
      }
    }} \\
    F_1 &\defeq \Lam[r]{\Lam[x, y]{
      \Pair{
        \ZApp{\ZApp{\MkBox[\delta]{f}}{\Proj[0]{y}}{\delta}}{\Proj[0]{x}}{\delta}
      }{
        \ZApp{\ZApp{r}{\Proj[1]{\Alert{y}}}{\ell}}{\Proj[1]{\Alert{x}}}{\ell}
      }
    }}
  \end{align*}
  because then
  \[
    \ZipWith'(\MkBox[\delta]{f}, v, u)
    \defeq
    \Lob(F_0)(u, v)
    =
    \Lob(F_1)(u, v)
    =
    \ZipWith'(\MkBox[\delta]{f}, u, v)
  \]
  By Theorem~\ref{thm:guarded-recursion:unique-fixed-points} we know guarded fixed points are
  unique, so it suffices to show that
  \begin{equation}
    \Lob(F_1) = F_0(\Next(\Lob(F_1)))
    \label{eq:guarded-recursion:goal}
  \end{equation}
  We use L\"{o}b induction to construct a term of type $\Eq{}{\Lob(F_1)}{F_0(\Next(\Lob(F_1)))}$.
  \begin{align*}
    &F_0(\Next(\Lob(F_1))) \\
    &= \Lam[x,y]{\Pair{\ZApp{\ZApp{\MkBox[\delta]{f}}{\Proj[0]{x}}{\delta}}{\Proj[0]{y}}{\delta}}{\ZApp{\ZApp{\Next(\Lob(F_1))}{\Proj[1]{x}}{\ell}}{\Proj[1]{y}}{\ell}}}\\
    &\qquad\qquad \text{by induction let $\MkBox[\delta]{a} \defeq \Proj[0](x)$ and $\MkBox[\delta]{b} \defeq \Proj[0](y)$} \\
    &= \Lam[x,y]{\Pair{\MkBox[\delta]{f(a, b)}}{\ZApp{\ZApp{\Next(\Lob(F_1))}{\Proj[1]{x}}{\ell}}{\Proj[1]{y}}{\ell}}}\\
    &= \Lam[x,y]{\Pair{\MkBox[\delta]{f(b, a)}}{\ZApp{\ZApp{\Next(\Lob(F_1))}{\Proj[1]{x}}{\ell}}{\Proj[1]{y}}{\ell}}}\\
    &= \Lam[x,y]{\Pair{\MkBox[\delta]{f(b, a)}}{\ZApp{\ZApp{\Next(F_1(\Next(\Lob(F_1))))}{\Proj[1]{x}}{\ell}}{\Proj[1]{y}}{\ell}}}\\
    &\qquad\qquad \text{by induction let $\MkBox[\ell]{s} \defeq \Proj[1]{x}$ and $\MkBox[\ell]{t} \defeq \Proj[1]{y}$}\\
    &= \Lam[x,y]{\Pair{\MkBox[\delta]{f(b, a)}}{\Next(F_1(\Next(\Lob(F_1))(s, t))}}\\
    &= \Lam[x,y]{\Pair{\MkBox[\delta]{f(b, a)}}{\Next(F_0(\Next(\Lob(F_1))(t, s))}}\\
    &= \Lam[x,y]{\Pair{\MkBox[\delta]{f(b, a)}}{\ZApp{\ZApp{\Next(F_0(\Next(\Lob(F_1))))}{\Proj[1]{y}}{\ell}}{\Proj[1]{x}}{\ell}}}\\
      &\qquad\qquad \text{using the IH through Lemma \ref{lem:guarded:later-eq}}\\
    &= \Lam[x,y]{\Pair{\MkBox[\delta]{f(b, a)}}{\ZApp{\ZApp{\Next(\Lob(F_1))}{\Proj[1]{y}}{\ell}}{\Proj[1]{x}}{\ell}}}\\
    &= \Lam[x,y]{\Pair{\ZApp{\ZApp{\MkBox[\delta]{f}}{\Proj[0]{y}}{\delta}}{\Proj[0]{x}}{\delta}}{\ZApp{\ZApp{\Next(\Lob(F_1))}{\Proj[1]{y}}{\ell}}{\Proj[1]{x}}{\ell}}}\\
    &= \Lob(F_1) \qedhere
  \end{align*}
\end{proof}

\begin{rem}[Previous approaches]
  Using dependent type theories to reason about guarded recursion and coinductive types has been a
  problem for some time \cite{moegelberg:2014}. The technical device of \emph{clocks}, due
  to~\cite{atkey:2013}, was introduced to deal with productivity in a simply-typed setting. Clocks
  were then introduced to dependent types \cite{moegelberg:2014}, and later refined into the
  extensional guarded type theory \textbf{gDTT} of \cite{bizjak:2016}.
  
  In  \textbf{gDTT} the problem of `totalising' a type---which corresponds to reasoning by
  coinduction---was not handled through the `always' modality, but through clocks. In essence,
  \textbf{gDTT} does not come with a single $\Later$ modality, but rather with a collection of them,
  each one indexed by a clock name. There is a quantifier which allows clock names to be bound
  inside a particular type, and a crucial isomorphism:
  \[
    \forall \kappa.\ A \cong \forall \kappa.\ \Later^\kappa A
    \TagEq[$\ast$] \label{eq:guarded-recursion:clock-force-iso}
  \]
  \textbf{gDTT} presents several technical complications. The syntactic problems pertaining to
  delayed substitutions were resolved by the introduction of Clocked Type Theory (CloTT)
  \cite{bahr:2017}, which uses additional judgmental structure. It is conjectured that type-checking
  is decidable for CloTT. The complexity of using clocks also appears in the
  semantics of clocked type theory. CloTT is modelled in a collection of presheaf categories, with
  multiple functors navigating between them~\cite{mannaa:2018}.
  
  It was hoped that some of the complexity could be circumvented by replacing clocks with a
  modality. This led Clouston et al. to introduce the comonadic `always' modality $\Box$, which
  replaced the isomorphism \eqref{eq:guarded-recursion:clock-force-iso} with $\Box \Later A \cong
  \Box A$ \cite{clouston:2015}. The main advantage of using $\Box$ is that it can be interpreted in
  $\PSH{\omega}$, which is a much simpler model. On the other hand, the interactions between $\Box$
  and $\Later$ have proven difficult to capture in the syntax. In fact, the mere addition of $\Box$
  to a dependent type theory poses a significant technical challenge:
  see~\cite{bahr:2017,clouston:dra:2018,shulman:2018,gratzer:2019}. Despite this concentrated
  effort, there are still serious technical obstacles to adding $\Later$ to a type theory for
  $\Box$. \MTT{} is the first syntax to accomodate both $\Box$, $\Later$, and validate $\Box \Later
  A \cong \Box A$.
\end{rem}


\section{Internal Adjoints}
\label{sec:adjoints}

In many cases of interest, the need for a pair of \emph{adjoint modalities} arises: we would like a
pair of modalities $\mu : n \to m$ and $\nu : m \to n$ so that, in some sense,
\[
  \Modify[\nu]{-} \Adjoint \Modify{-}
\]
But what does it mean to have an adjunction between two modalities \emph{within} \MTT{}? Does it
correspond to an external adjunction? And do all known results from category theory apply? The only
thing that is certain is that this scenario is fundamental to modal type theory, as a number of
intended models can be elegantly presented through
adjunctions~\cite{schreiber:2014,nuyts:2018,shulman:2018}.

In this section we show that when \MTT{} is equipped with the \emph{walking adjunction} as a mode
theory, it becomes a useful syntax for reasoning about adjoint modalities. Of course, the adjoint
modalities themselves are not exactly adjoint functors: they are something slightly weaker than
DRAs, whose `left adjoints' constitute an adjunction. Nevertheless, we prove that the induced
modalities largely behave as expected: the unit and counit are internally definable; some limited
forms of internal transposition can be recovered; and left adjoints preserve colimits, as expressed
through \emph{crisp induction principles}.

\subsection{The walking adjunction}

As ever, we begin by freely defining a mode theory $\ModeAdj$. Its generators are two 1-cells $\nu :
m \to n$ and $\mu : n \to m$, and two 2-cells
\begin{align*}
  \eta &: 1_m \To \mu \circ \nu
  &
  \epsilon &: \nu \circ \mu \To 1_n
\end{align*}
subject to the triangle equations
\[
  \begin{tikzcd}
    \mu
      \arrow[r, Rightarrow, "{\eta \Whisker 1_\mu}"]
      \arrow[dr, equal]
    & \mu \circ \nu \circ \mu
      \arrow[d, Rightarrow, "{1_\mu \Whisker \epsilon}"]
    \\
    & \mu
  \end{tikzcd}
  \quad
  \begin{tikzcd}
    \nu
      \arrow[r, Rightarrow, "{1_\nu \Whisker \eta}"]
      \arrow[dr, equal]
    & \nu \circ \mu \circ \nu
      \arrow[d, Rightarrow, "{\epsilon \Whisker 1_\nu}"]
    \\
    & \nu
  \end{tikzcd}
\]
$\ModeAdj$ is sometimes called the \emph{walking adjunction} \cite[\S 5.1]{licata:2016}. It is the
\emph{classifying 2-category} for an adjunction: 2-functors $\Mor{\ModeAdj}{\CC}$ correspond
precisely to (2-categorical) adjunctions in $\CC$. The mode theory $\ModeAdj$ has a very curious
property: it is \emph{self-dual}, \ie~there is an equivalence $\Coop{\ModeAdj} \Equiv \ModeAdj$.
This equivalence sends the modes to each other, the adjoints to themselves and the 2-cells $\eta$
and $\epsilon$ again to each other.

\subsection{Models of adjoint modalities}

Recall that a modal context structure of a model of \MTT{} with mode theory $\ModeAdj$ is a strict
2-functor $\Interp{-} : \Coop{\ModeAdj} \to \CAT$. The self-duality of $\ModeAdj$ implies that such
a context structure consists of two categories and an adjunction between them. We immediately obtain
the following result.

\begin{cor}
  If $\CC$ and $\DD$ carry models of \MLTT{}, and there is a pair of dependent right adjoints
  between them whose `left adjoints' are themselves adjoint, then we can construct a model of \MTT{}
  with mode theory $\ModeAdj$.
\end{cor}
\begin{proof}
  Write $\InterpModal{\nu} : \CC \to \DD$ and $\InterpModal{\mu} : \DD \to \CC$ for the functors
  given as part of the DRAs. The notation is then suggestive: $\InterpModal{\nu} \Adjoint
  \InterpModal{\mu}$, and Theorem~\ref{thm:dra:dra} applies.
\end{proof}

Conversely,

\begin{thm}
  Any model of $\ModeAdj$ must interpret $\InterpModal{\nu}$ and $\InterpModal{\mu}$ as adjoint
  functors. Moreover, if $\CModify[\mu]$ and $\CModify[\nu]$ are induced by lifting the adjunctions
  $\InterpModal{\mu} \Adjoint R_\mu$ and $\InterpModal{\nu} \Adjoint R_\nu$ to a dependent
  right adjoints (by Lemma~\ref{lem:constructing-dra:lifting-adjunction}), then $R_\nu \Adjoint
  R_\mu$.
\end{thm}
\begin{proof}
  Adjoint functors are precisely adjoint morphisms in the 2-category $\CAT$. As $\ModeAdj$ is the
  walking adjunction, and 2-functors preserve adjunctions, we have that $\InterpModal{\nu} \Adjoint
  \InterpModal{\mu}$.

  If $\InterpModal{\nu} \Adjoint R_\nu$,
  then by the uniqueness of adjoint pairs we must have that $R_\nu \cong \InterpModal{\mu}$.
  If moreover $\InterpModal{\mu} \Adjoint R_\mu$,
  then the previous isomorphism yields $R_\nu \Adjoint R_\mu$.
\end{proof}
The last situation in this lemma is sometimes known as an `adjunction of adjunctions' \cite[\S
5.1]{licata:2016}. In particular, the action of the right adjoint modality $\mu$
on contexts, viz.\ $\InterpModal{\mu}$, is in some sense internalized on types and terms by the action of the left adjoint modality $\nu$ on types and terms,
viz.\ $\Modify[\nu]{-}$.

\subsection{Recovering the adjunction internally}
\label{sec:adjoints:internalization}

The foregoing construction of a model interpreted the lock functors required by $\ModeAdj$ by an
adjunction. Consequently, substitutions $\Delta \to \LockCx{\Gamma}$ are in natural bijection with
substitutions $\LockCx{\Delta}<\nu> \to \Gamma$. We would like to strengthen this setting by
bootstrapping this adjunction into an \emph{internal adjunction}.

It is not immediately clear what an internal adjunction should be. However, we can construct an
appropriate definition by internalizing the unit and counit as functions. But that is not immediate
either: if $\IsTy{A}[1]$, the construction $\Modify[\mu]{\Modify[\nu]{A}}$ that we would na\"{i}vely
try as the codomain of the unit is ill-typed. This can be mended through key substitutions. Recall
that $\eta : 1_m \To \mu \circ \nu$. The corresponding key substitution at $\Gamma
\Mute{{}\mathop{@} m}$ is $\Key{\eta}{\Gamma} : \LockCxV{\Gamma}<\mu \circ \nu> \to \Gamma
\Mute{{}\mathop{@} m}$. We can use this to formally define the notation of
Section~\ref{sec:towards-mtt:multiple-modalities} by
\[
  A^\eta \defeq \Sb{A}{\Key{\eta}{\Gamma}}
\]
As substitutions can be eliminated (e.g. through a subset of the canonicity algorithm), this defines
an admissible operation from type $\IsTy{A}[1]$ to type $\IsTy[\LockCxV{\Gamma}<\mu \circ
\nu>]{A^\eta}[1]$. We can thus define the unit component at $\IsTy{A}[1]$ by
\[
  \arraycolsep=1.4pt
  \begin{array}{lcl}
  \Unit    &:&      A \to \Modify[\mu]{\Modify[\nu]{A^\eta}} \Mute{{}\mathop{@} m} \\
  \Unit(x) &\defeq& \MkBox[\mu]{\MkBox[\nu]{x^\eta}}
  \end{array}
\]
Dually, for any type $\IsTy[\LockCxV{\Gamma}<\nu \circ \mu>]{A}[1]<n>$ we can define the counit
component by
\[
  \arraycolsep=1.4pt
  \begin{array}{lcl}
  \Counit    &:& \Modify[\nu]{\Modify[\mu]{A}} \to A^\epsilon \Mute{{}\mathop{@} n} \\
  \Counit(x) &\defeq& \Open{x}[y_0]{\Open{y_0}[y_1]{y_1^\epsilon}<\mu>[\nu]}<\nu>[1]
  \end{array}
\]
We thus obtain the unit and counit internally, but the types of the components have to be adjusted
in the presence of dependence. Moreover, we can prove internal versions of the triangle equations;
they are given by modal induction:
\[
  \arraycolsep=1.4pt
  \begin{array}{lcl}
    \_ &:& (x : \Modify[\nu]{A}) \to \Id{\Modify[\nu]{A}}{x}{\Counit(\ZApp{\MkBox[\nu]{\Unit}}{x}{\nu})}\\
    \_ &\defeq& \Lam[x]{\Open{x}[y]{\Refl{\MkBox[\nu]{y}}}<\nu>[1]}\\[0.3cm]
    \_ &:& (x : \Modify[\mu]{A}) \to \Id{\Modify[\mu]{A}}{x}{\ZApp{\MkBox[\mu]{\Counit}}{\Unit(x)}{\mu}}\\
    \_ &\defeq& \Lam[x]{\Open{x}[y]{\Refl{\MkBox[\mu]{y}}}<\mu>[1]}
  \end{array}
\]
The most difficult part is proving that these terms are well-typed. For example, in the first
instance we must show that $\MkBox[\nu]{y} =
\Counit(\ZApp{\MkBox[\nu]{\Unit}}{\MkBox[\nu]{y}}{\nu})$ definitionally:
\begin{align*}
  \Counit(\ZApp{\MkBox[\nu]{\Unit}}{\MkBox[\nu]{y}}{\nu})
  &= \Counit(\MkBox[\nu]{\App{\Unit}{y}})\\
  &= \Counit(\MkBox[\nu]{\MkBox[\mu]{\MkBox[\nu]{y^\eta}}})\\
  &= \MkBox[\nu]{y^\eta}^\epsilon\\
  &= \MkBox[\nu]{(y^\eta)^{\epsilon \Whisker 1_\nu}}\\
  &= \MkBox[\nu]{y^{(\epsilon \Whisker 1_\nu) \circ (1_\nu \Whisker \eta)}}\\
  &= \MkBox[\nu]{y}
\end{align*}
Because we are using slightly informal syntax here, it is difficult to see that the steps that
introduce whiskering are correct. They become much more perspucious if we expand
$\MkBox[\nu]{y^\eta}^\epsilon$ into algebraic syntax, and use the last equation of
Fig.~\ref{fig:mtt-gat-eqsubst} twice to absorb locks:
\[
\Sb{\MkBox[\nu]{\Sb{y}{\Key{\eta}{\LockCx{\Gamma}<\nu>}}}}{\Key{\epsilon}{\Gamma}}
= \MkBox[\nu]{\Sb{y}{\Key{\eta}{\LockCx{\Gamma}<\nu>} \circ \LockSb{\Key{\epsilon}{\Gamma}}<\nu>}}
= \MkBox[\nu]{\Sb{y}{\Key{1_\nu \Whisker \eta}{\Gamma} \circ \Key{\epsilon \Whisker 1_\nu}{\Gamma}}}
= \MkBox[\nu]{\Sb{y}{\Key{1_\mu}{\Gamma}}}
\]

\subsection{Internal transposition}

The previous section offers a perfectly good internal representation of an external adjunction.
However, it is usually much more economical to present an adjunction by a natural isomorphism
\[
  \Hom{L(A)}{B} \cong \Hom{A}{R(B)}
\]
Unfortunately, this is not achievable in \MTT{} for a multitude of reasons. First, notice that
$\Modify[\nu]{A} \to B  \Mute{{}\mathop{@} n}$ and $A \to \Modify{B} \Mute{{}\mathop{@} m}$ are
types in different modes, so the putative type $(\Modify[\nu]{A} \to B) \Equiv (A \to \Modify{B})$
that would represent the isomorphism is ill-typed.

Second, even if the two modes coincide---so that $\nu, \mu$ are endomodalities---the aforementioned
type is a bit too strong for our purpose: it is inhabited by \emph{internal equivalences}, which are
stronger than bijections of hom-sets. Such equivalences correspond to isomorphisms $B^{L(A)} \cong
R(B)^A$ of exponential objects. In turn, these are equivalent to hom-set bijections only if the
involved functors are \emph{internal}, which is to say that we have functions
\begin{align*}
  (A \to B) &\to (\Modify[\nu]{A} \to \Modify[\nu]{B})
  &
  (A \to B) &\to (\Modify{A} \to \Modify{B})
\end{align*}
that compute the action of the modality $\Modify{-}$ on morphisms \emph{within}
\MTT{}.\footnote{Such functors are usually called \emph{enriched} (recall that cartesian closure is
  a self-enrichment).}

Third, even if we could internalize our modalities, we would be flying too close to the sun. As we
have $(1 \to A) \Equiv A$ for any $A$ and $\Modify[\xi]{1} \Equiv 1$ for any $\xi$, we may calculate
that
\[
  A\ \Equiv\ 1 \to A\ \Equiv\ \Modify[\nu]{1} \to A\ \Equiv\ 1 \to \Modify{A}\ \Equiv\ \Modify{A}
\]
Hence, $\Modify{-}$ must be the identity functor up to equivalence. This short argument, which is
due to \cite[Theorem 5.1]{licata:2018}, is a \emph{no-go theorem} that obstructs the
internalization of an adjunction whose left adjoint preserves terminal objects.

\cite{licata:2018} overcame this barrier by introducing the \emph{global sections
modality}~$\flat$. Terms of $\flat A$ represent \emph{global} elements of $A$: terms of $\flat(A \to
B)$ are in bijection with morphisms in $\Hom{A}{B}$. Thus, the previously problematic equivalence
holds under $\flat$.

We can rephrase this argument in our syntax. The key thing to notice is that the functor $\flat :
\PSH{\CC} \to \PSH{\CC}$ which maps a presheaf to the constant presheaf $\_ \mapsto
\Hom[\PSH{\CC}]{1}{P}$ of global sections is initial amongst functors that preserve the terminal
object. Thus, we postulate an initial modality: suppose that $n = m$, and that $\Hom{m}{m}$ is
equipped with an initial object, \ie~a 1-cell $\flat : m \to m$ and a unique 2-cell ${!} : \flat \To
\xi$ for all $\xi$. As a consequence, we are able to use variables $\DeclVar{x}{A}<\flat>$ in any
context. Assuming function extensionality, we have that
\begin{thm}
  \label{thm:programming-in-mtt:transposition}
  There is an equivalence $\Modify[\flat]{\Modify[\nu]{A^{!}} \to B} \Equiv \Modify[\flat]{A \to
  \Modify{B^{!}}}$.
\end{thm}
\begin{proof}
  The equivalence is given by the functions
  \[
  \arraycolsep=1.4pt
  \begin{array}{lclclcl}
    F & : & \Modify[\flat]{\Modify[\nu]{A^{!}} \to B} \to \Modify[\flat]{A \to
    \Modify{B^{!}}} \\
    F(f) &\defeq& \Open{f}[g]{\MkBox[\flat]{\Lam[x]{\ZApp{\MkBox[\mu]{g^!}}{\Unit(x)}{\mu}}}}<\flat>[1] \\[0.3cm]
    G &:& \Modify[\flat]{A \to \Modify{B^{!}}} \to \Modify[\nu]{A} \to\Modify[\flat]{\Modify[\nu]{A^{!}} \to B}\\
    G(g)&\defeq& \Open{g}[f]{\MkBox[\flat]{\Lam[x]{\Counit(\ZApp{\MkBox[\nu]{f^!}}{x}{\nu})}}}<\flat>[1]
  \end{array}
  \]
  These are well-typed because, by initiality of $\flat$, $A^\eta = (A^{!})^\eta = A^!$,
  $(B^!)^\epsilon = B^!$. By function extensionality and $\eta$ for modalities, they are mutually
  inverse.
\end{proof}

The closest we can get to defining internal transposition (without using an initial modality)
amounts to the following two functions.
\[
  \arraycolsep=1.4pt
  \begin{array}{lclclcl}
    \TranspAdj{\nu}{\mu} & : & \Modify[\mu]{\Modify[\nu]{A^\eta} \to B} \to A \to \Modify{B} \\
    \TranspAdj{\nu}{\mu} &\defeq& \Lam[f]{\Lam[x]{\ZApp{f}{\Unit(x)}{\mu}}} \\[0.3cm]
    \UntranspAdj{\nu}{\mu} &:& \Modify[\nu]{A \to \Modify{B}} \to \Modify[\nu]{A} \to B^\epsilon\\
    \UntranspAdj{\nu}{\mu} &\defeq& \Lam[f]{\Lam[x]{\Counit(\ZApp{f}{x}{\nu})}}
  \end{array}
\]
The first is an equivalence (again up to function extensionality), but neither have the expected
type. The first transposition $\TranspAdj{\nu}{\mu}$ is not without precedent: it is the internal
formulation of transposition for adjunctions between monoidal closed categories when the left
adjoint preserves monoidal products.

\subsection{Crisp induction}
\label{sec:adjoints:modal-eliminators}

Having internalized the definition of an adjunction, it is natural to ask whether standard facts
about adjoint functors carry over. In this section we prove an internal version of the fact that
left adjoints preserve colimits. Within type theory this result takes the form of \emph{crisp
  induction principles} for various types that arise from colimits.

As a first approximation to the notion of crisp induction, recall the rule for modal induction,
i.e. the elimination rule for modal types from Section~\ref{sec:towards-mtt}:
\[
  \inferrule{
    \mu : n \to m \\
    \nu : m \to o \\
    \IsTy[\ECxV{\Gamma}{x}{\Modify[\mu]{A}}<\Alert{\nu}>]{B}[1]<o> \\\\
    \IsTm[\LockCxV{\Gamma}<\Alert{\nu}>]{M_0}{\Modify[\mu]{A}}<m> \\
    \IsTm[\ECxV{\Gamma}{x}{A}<\Alert{\nu} \circ \mu>]{M_1}{\Sb{B}{\MkBox[\mu]{x}/x}}<o>
  }{
    \IsTm{\Open{M_0}[x]{M_1}<\mu>[\Alert{\nu}]}{\Sb{B}{M_0/x}}<o>
  }
\]
Notice that there is an ``extra'' modality parameterizing this rule, $\nu$, which modifies $M_0$ as
well as the data supplied to $M_1$. This extra generality is not frivolous: we can only define the
equivalence $\MComp{}{\mu}{\nu}$ of Section~\ref{sec:programming-in-mtt} because we can eliminate one
modality `under' another.

One might hope for a similar level of flexibility in all positive eliminators. However, the
elimination rule for booleans---stated here in its algebraic form of
Section~\ref{sec:algebraic-mtt}---does not allow it:
\[
  \inferrule{
    \IsCx{\Gamma}\\
    \IsTy[\ECx{\Gamma}{\Bool}<\Alert{1}>]{A}[1]\\
    \IsTm{M_t}{\Sb{A}{\ESb{\ISb}{\True}}}\\
    \IsTm{M_f}{\Sb{A}{\ESb{\ISb}{\False}}}\\
    \IsTm[\LockCx{\Gamma}<\Alert{1}>]{N}{\Bool}
  }{
    \IsTm{\BoolRec{A}{M_t}{M_f}{N}}{\Sb{A}{\ESb{\ISb}{N}}}
  }
\]
Were we to replace $1$ with an arbitrary modality, then this rule would state something considerably
stronger: not only would we have the expected elimination principle for $\Bool$, but all of our
modalities would \emph{preserve} $\Bool$. Semantically, this is nonsense: modalities intuitively
correspond to right adjoints, and therefore do not necessarily preserve colimits. For example, the
later $\Later$ modality of Section~\ref{sec:guarded-recursion} does not preserve booleans.

Yet, in some circumstances---\eg{} when a modality is a left adjoint---the stronger rule is valid.
This is the idea behind Shulman's crisp induction principles \cite[\S 5]{shulman:2018}: cohesive
type theory enables the proof of elimination principles for the coproducts and the identity type
under the left adjoint in the adjunction $\flat \Adjoint \sharp$. We will demonstrate that similar
principles are derivable within \MTT{} with mode theory $\ModeAdj$.

Fix a motive $\IsTy[\ECxV{\LockCxV{\Gamma}<\nu \circ \mu>}{b}{\Bool}<\nu>]{C}[1]<n>$. Crisp
induction is given by a term
\[
  \Gamma \vdash{} \CrispBoolRec_C :
  (\DeclVar{b}{\Bool}<\nu>) \to
  \Modify[\nu \circ \mu]{C(\True)} \to
  \Modify[\nu \circ \mu]{C(\False)} \to
  C^\epsilon(b)
  \Mute{{}\mathop{@} n}
\]
This is a well-formed type, as $\ECxV{\Gamma}{b}{\Bool}<\nu> =
\IsTy[\ECxV{\LockCxV{\Gamma}<1>}{b}{\Bool}<\nu>]{C^\epsilon}[1]$.

To obtain the crisp induction principle we first use the ordinary one at mode $n$, and apply a
number of modal combinators to bring it to mode $m$.
\[
  \arraycolsep=1.4pt
  \begin{array}{lcl}
    \LockCxV{\Gamma}<\nu> \vdash{} h & : & (b : \Bool) \to \Modify{C(\True)} \to \Modify{C(\False)} \to \Modify{C(b^\eta)} \Mute{{}\mathop{@} m}\\
    h(b, t, f) & \defeq & \BoolRec{b.\ \Modify{C(b^\eta)}}{t}{f}{b}\\[0.3cm]
    \Gamma \vdash{} \CrispBoolRec_C & : &
    (\DeclVar{b}{\Bool}<\nu>) \to
    \Modify[\nu \circ \mu]{C(\True)} \to
    \Modify[\nu \circ \mu]{C(\False)} \to
    C^\epsilon(b) \Mute{{}\mathop{@} m}\\
    \CrispBoolRec_C(b,t,f) & \defeq &
      \Counit(\ZApp{\ZApp{\MkBox[\nu]{h(b)}}{\MComp*{t}{\nu}{\mu}}{\nu}}{\MComp*{f}{\nu}{\mu}}{\nu})
  \end{array}
\]
The reasons why this term is well-typed is subtle. We have that
$\IsTm[\LockCxV{\ECxV{\LockCxV{\Gamma}<\nu>}{b}{\Bool}<1>}<\mu \circ \nu>]{b^\eta}{\Bool}<m>$, so
$\IsTy[\LockCxV{\ECxV{\LockCxV{\Gamma}<\nu>}{b}{\Bool}<1>}<\mu>]{C(b^\eta)}[1]<n>$ by the
application rule. Thus, $h$ is well-typed. It remains to show that $C(b^\eta)^\epsilon =
C^\epsilon(b)$, which intuitively follows from the triangle identities. We may show it by precisely
specifying what these operations mean in the algebraic syntax. First, we construct the substitutions
\begin{alignat*}{2}
  \sigma_0 &\defeq \ESb{\LockSb{\Wk}<\mu>}{\Sb{\Var{0}}{\Key{\eta}{\ECx{\LockCx{\Gamma}<\nu>}{\Bool}<1>}}}\
  &&: \LockCx{\ECx{\LockCx{\Gamma}<\nu>}{\Bool}<1>}<\mu> \to \ECx{\LockCx{\Gamma}<\nu \circ \mu>}{\Bool}<\nu> \Mute{{}\mathop{@} n}\\
  \sigma_1 &\defeq \LockSb{\ESb{\LockSb{\Wk}<\nu>}{\Var{0}}}\
  &&: \LockCx{\ECx{\Gamma}{\Bool}<\nu>}<\nu \circ \mu> \to \LockCx{\ECx{\LockCx{\Gamma}<\nu>}{\Bool}<1>}<\mu> \Mute{{}\mathop{@} n}\\
  \sigma_2 &\defeq \ESb{\parens*{\Key{\epsilon}{\Gamma} \circ \Wk}}{\Var{0}}\
  &&: \ECx{\Gamma}{\Bool}<\nu> \to \ECx{\LockCx{\Gamma}<\nu \circ \mu>}{\Bool}<\nu>
  \Mute{{}\mathop{@} n}
\end{alignat*}
We can then interpret $C(b^\eta)$ as the type
$\IsTy[\LockCx{\ECx{\LockCx{\Gamma}<\nu>}{\Bool}<1>}<\mu>]{\Sb{C}{\sigma_0}}[1]<n>$. Similarly,
$C(b^\eta)^\epsilon$ is the type
$\IsTy[\ECx{\Gamma}{\Bool}<\nu>]{\Sb{\Sb{\Sb{C}{\sigma_0}}{\sigma_1}}{\Key{\epsilon}{\ECx{\Gamma}{\Bool}<\nu>}}}[1]<n>$.
Finally, $C^\epsilon(b)$ is the type $\IsTy[\ECx{\Gamma}{\Bool}<\nu>]{\Sb{C}{\sigma_2}}[1]<n>$, so
it suffices to show that $\sigma_0 \circ \sigma_1 \circ \Key{\epsilon}{\ECx{\Gamma}{\Bool}<\nu>} =
\sigma_2$. This is a monstrous equation which is primarily structural. Moreover, $\eta$ occurs in
$\sigma_0$, and $\epsilon$ in the key that follows it, so one of the triangle equations must
somehow be implicated. Indeed, we can use one of the two equations along with the rules of
Section~\ref{sec:algebraic-mtt} to prove the desired result.

We can now prove that

\begin{thm}
  \label{thm:adjoints:preserve-bool}
  $\Modify[\nu]{\Bool} \Equiv \Bool$
\end{thm}
\begin{proof}
  We define the two functions
  \[
    \arraycolsep=1.4pt
    \begin{array}{lcl}
      \BoolIso &:& \Modify[\nu]{\Bool} \to \Bool \Mute{{}\mathop{@} m}\\
      \BoolIso(x) &\defeq& \Open{x}[y]{\CrispBoolRec_{\Bool}(y, \MkBox[\nu \circ \mu]{\True}, \MkBox[\nu \circ \mu]{\False})}<\nu>[1]\\[0.3cm]
      \BoolIso* &:& \Bool \to \Modify[\nu]{\Bool} \Mute{{}\mathop{@} m}\\
      \BoolIso* &\defeq& \Lam[x]{\BoolRec{\_.\ \Modify[\nu]{\Bool}}{\MkBox[\nu]{\True}}{\MkBox[\nu]{\False}}{x}}
    \end{array}
  \]
  We now use full crisp induction to construct for every $x : \Modify[\nu]{\Bool}$ a proof of
  $\Id{\Modify[\nu]{\Bool}}{x}{\BoolIso*(\BoolIso(x))}$. First, use modal induction to write $x =
  \MkBox[\nu]{y}$ for some $\DeclVar{y}{\Bool}<\nu>$. We then have to prove that
  $\Id{\Modify[\nu]{\Bool}}{\MkBox[\nu]{y}}{\BoolIso*(\BoolIso(\MkBox[\nu]{y}))}$, so we perform
  crisp induction on $y$. If $y \defeq \True$, we have that $\BoolIso*(\BoolIso(\MkBox[\nu]{\True}))
  = \MkBox[\nu]{\True}$, so $\MkBox[\nu \circ \mu]{\Refl{\MkBox[\nu]{\True}}}$ has the right type.
  The case for $y \defeq \False$ is similar. The other direction is simpler, and follows by induction on $\Bool$.
\end{proof}

Similar results hold for other types with `positive,' `pattern-matching,' or `closed-scope'
elimination rules. For example, we can also formulate a crisp induction principle for identity
types, which can be used to prove that
\begin{thm}
  $\Modify[\nu]{\Id{A}{M_0}{M_1}} \Equiv \Id{\Modify[\nu]{A}}{\MkBox[\nu]{M_0}}{\MkBox[\nu]{M_1}}$
\end{thm}


\section{Related Work}
\label{sec:related}

Modal type theory has been an active area of research for two decades and, as with any active field,
a precise taxonomy of modal type theories would be a paper in and of itself. Accordingly, we have
not attempted such a task here, and have instead focussed on separating modal type theories into
distinct strands based on their judgmental structure. Some of our characterizations are slightly
artificial, in that these lines of work are not nearly so separate as we seem to suggest. We feel,
however, that this is the simplest way to position \MTT{} in relation to current work.

\subsection{Dual-context modal calculi}
\label{sec:related:dual-context}

One of the first papers on (non-linear)\footnote{The idea of dual contexts arose in linear logic:
  see \cite{andreoli:1992,girard:1993,plotkin:1993}.} modal type theory was by
\cite{pfenning-davies:2001}, who constructed a proof theory for \textsf{S4}, \ie{} a comonadic
modality. The central idea of this approach was to reflect the distinction between modal and
non-modal assumptions (referred to as `truth vs.\ validity' in \emph{op.\ cit.}) in the judgmental
structure of the system itself.  The judgments for this calculus then contained not just a context
of true propositions, but rather two contexts: one for intuitionistic propositions, and one for
modal ones. Following this methodology, Davies and Pfenning internalized previously known patterns
of sequent calculus in a natural deduction style \cite{kavvos:2020}.

This kind of judgment straightforwardly allows the incorporation of a product-preserving comonad.
The type $\Box A$ merely internalizes a restriction to modal contexts only:
\begin{mathpar}
  \inferrule{
    \Delta; \Emp \vdash A\ \mathsf{true}
  }{
    \Delta; \Gamma \vdash \Box A\ \mathsf{true}
  }
  \and
  \inferrule{
    A \in \Delta \cup \Gamma
  }{
    \Delta; \Gamma \vdash A\ \mathsf{true}
  }
  \and
  \inferrule{
    \Delta; \Gamma \vdash \Box A\ \mathsf{true}\\
    \Delta, A; \Gamma \vdash B\ \mathsf{true}
  }{
    \Delta; \Gamma \vdash B\ \mathsf{true}
  }
\end{mathpar}
The second author showed that this pattern adapts well to the necessity fragment of a number of
normal modal logics \cite{kavvos:2020}. The dual-context style has been succesfully adapted to
dependent types: see e.g. the work of \cite{de-paiva:2015}, and the spatial and cohesive type
theories of \cite{shulman:2018}. Similarly, contextual modal type
theory~\cite{nanevski:2008,boespflug:2011,bock:2015,pientka:2019} has used a dual-context-like
structure in order to give a systematic account of higher-order abstract syntax.

\cite{zwanziger:2019} continues this program by formulating a precise categorical semantics
based on Awodey's natural models for a dependent type theory with either an adjunction
(\textsf{AdjTT}) or comonad (\textsf{CoTT})~\cite{zwanziger:2019}. The categorical semantics of
\MTT{} and \textsf{AdjTT} are closely related, though with minor differences in the precise
definition of the modality. For instance, in \MTT{} only the $\Lock$ operator is required to act
upon the context, while in \textsf{AdjTT} the modalities themselves must extend to
contexts.\footnote{This is similar to the relation between a CwF+A and a CwDRA from
\cite{clouston:dra:2018}, and we expect a similar relation to exist between the semantics of \MTT{}
with a single modality and \textsf{AdjTT}.} These differences arise because Zwanziger characterizes
only a certain, semantically well-behaved subclass of models, while in Section~\ref{sec:semantics} we
describe more general models, which also support the syntactic model and the gluing model of
Section~\ref{sec:canonicity}. Syntactically, \textsf{AdjTT} is a multimode type theory that includes a
mode for both ends of the adjunction.

Despite these stories of success, the dual-context style is difficult to generalize: as the
complexity of the modal situation increases, so must the complexity of the context structure. For
instance, the structure of a dependent dual-context type theory enforces that a `modal' type (one
belonging to $\Delta$) may not depend on an `intuistionistic' type (one belonging to $\Gamma$). This
is a reasonable restriction in the case of $\Box$, but it is already somewhat limiting. For
instance, it should be allowed for a valid type to depend on a merely true $\Box A$. Making such an
adjustment would not only present a typographical problem (with a type occurring to the left of one
of its dependencies), it would render the introduction rule for $\Box A$ nonsensical.

This restriction proves even more difficult to manage once there is not merely one modality, but two
distinct modalities ones, say $\mu$ and $\nu$. Questions such as ``should the $\mu$-modified types
be allowed to depend on $\nu$-modified types?'' defy general answers. These questions can be
addressed for each specific modal situation. For example, both \cite{shulman:2018} and
\cite{zwanziger:2019} hand-craft a system for two modalities. However, these constructions strongly
depend on the structure of the underlying model, encouraging the proliferation of tiresome
metatheoretic work as we discussed in Section~\ref{sec:intro}.

What is lacking with the dual-context style is the ability to work systematically with a large class
of modal situations without reconsidering the properties of the system in each case. Some of the
rules of \MTT{} can be directly traced to rules in dual-context calculi (in particular, the
elimination rule for modal types), but the structure of our contexts is radically different, in a
way which is far more accommodating.

\subsection{Modal type theories based on left division}
\label{sec:related:irrelevance}

A separate strand of modal type theories builds its syntax around a structure that is termed \emph{left division} by
\cite{nuyts:2018}. Rather than having a fixed number of distinct modal and intuitionistic contexts,
there is a single context consisting of variables with \emph{modal annotations}. The earliest
appearance of this pattern is in the work of \cite{pfenning:2001}, where the annotations described a
variable as having various degrees of proof (ir)relevance.

In a non-dependent type system, the distinction between annotations and different contexts is
artificial: we could simply sort variables by their annotation, and separate them into different
context zones. However, once generalized to a dependent type theory have a distinct advantage: they
do not impose a fixed dependence schedule between different contexts. Instead, a type may depend on
anything preceding it in the context, but the nature of that dependence is moderated by the modal
annotations.

The term `left division' is chosen to describe this structure because of the behavior of the
introduction rules for modal types. For instance, in \cite{pfenning:2001}, there is a rule for
introducing a term in an irrelevant context:
\[
  \inferrule{
    \Gamma^{\oplus} \vdash M : A
  }{
    \Gamma \vdash M :_{\mathsf{irr}} A
  }
\]
Here $-^{\oplus}$ is a metatheoretic operation, which traverses the context and removes irrelevance
annotations. The effect of this is that all the variables in $\Gamma^{\oplus}$ can be used freely
when type-checking $M$. This is acceptable, because $M$ itself is irrelevant. Viewed properly this
is a division operation which `divides' all the annotations in $\Gamma$ by $\mathsf{irr}$. The
metatheory of a full dependent type theory based on this idea was considered by \cite{abel:2012},
who prove that modelling irrelevance in this way is sound and decidable.

More recent work by the third author~\cite{nuyts:2017,nuyts:2018} has carried this idea to its
natural conclusion by incorporating an entire hierarchy of modalities. In a related but distinct
line of work, the \texttt{Granule} Project~\cite{gaboardi:2016,orchard:2019} has exploited a
similar structure to give a systematic account of substructurality. There is ongoing work to extend
this to a full dependent type theory.

The modal annotations of \MTT{} are very similar to the modal annotations of variables in calculi with
left division. Contrasting \MTT{} with \cite{pfenning:2001} in particular, we find that there are
three classes of variables in \emph{op.\ cit.}: normal variables (written $x : A$), irrelevant
variables ($x \div A$), and valid variables ($x :: A$). Such a situation would be modeled in \MTT{}
by a single mode that has three endomodalities: irrelevance, extensionality (the identity modality),
and validity. A composition table for these modalities can be built from the relations in
\cite{pfenning:2001}'s calculus.

The rules for interacting with the modalities in \emph{op.\ cit.} traverse the context and modify the
binding used for each variable. This bulk operation is very different to \MTT{}-style locks, but
amounts to similar constraints on variable use. By tagging the context with a lock, every time we
use a variable we must ensure that the annotated modality sufficiently strong to overcome the lock.
When we bulk-update the context, the same restrictions occur but they are performed `eagerly.'

The use of `lazy' locks has several advantages over `eager' bulk updates.
For instance, we do not have to explain what it means to divide one modality by another, and
non-trivial 2-cells are possible. Furthermore, when interpreting the calculus in a model, it is
unnecessary to describe variable by variable how modality update affects the interpretation of the
entire context (which can be challenging: see \eg~\cite{nuyts:tech-report:2018}).

\subsection{Fitch-style modal type theories}
\label{sec:related:fitch}

A recent series of papers has used a judgmental structure that is similar to \MTT{} in order to
manage a variety of modalities~\cite{bahr:2017,clouston:dra:2018,gratzer:2019}. This kind of
structure, informally often referred to as the \emph{Fitch-style}~\cite{clouston:fitch:2018},
divides the context into regions of variables separated by locks,
but does not use annotations on individual variables.
Locks are dynamically included or
removed by the typing rules.

The central advantage of the Fitch-style is the impressively simple introduction rule for
modalities: whenever we wish to introduce a modality we simply append a lock to the context---which
tags the modal shift---and continue typechecking. In particular, we never need to remove variables
from the context during the introduction of a modal term. Of course, like in \MTT{} this style is
only sound for a modality which comes equipped with some sort of left-adjoint-like operation.

Another desirable property of the Fitch-style calculi is their support for strong elimination rules
for modalities. Instead of the pattern matching-style rules of other systems, Fitch-style calculi
have had an \emph{open scope} elimination rule for their modalities, which often permits a
definitional $\eta$-rule for $\Box A$. It is generally of the following shape:
\[
  \inferrule{
    \mathfrak{F}(\Gamma) \vdash M : \Box A
  }{
    \Gamma \vdash \mathsf{open}(M) : A
  }
\]
$\mathfrak{F}$ is a meta-theoretic operation on contexts which removes some number of locks and
variables from $\Gamma$. For instance, in \cite{clouston:dra:2018} the operation
$\mathfrak{F}(\Gamma)$ was defined by
\[
  \mathfrak{F}(\Gamma, \text{\faUnlock}, \Gamma') = \Gamma \text{ where } \text{\faUnlock} \not\in \Gamma'.
\]
This rule is convenient, and strictly more powerful than that of \MTT{} (see Section~\ref{sec:dra}).
However, it is metatheoretically less than ideal. The source of the trouble in this case is that we
must show that substitutions can be pushed under the $\mathsf{open}$ construct. For instance,
suppose we have some substitution $\gamma : \Delta \to \Gamma, \text{\faUnlock}, \Gamma'$. It is
necessary to ensure that this substitution uniquely gives rise to a substitution
$\mathfrak{F}(\gamma) : \mathfrak{F}(\Delta) \to \Gamma$, which will then be applied to the body $M$
of the term. This property can only be shown by lengthy induction on syntax. Such a property is
proven laboriously in \cite{gratzer:2019} for the $\MLTTLock$ type theory, and several complex and
seemingly artificial typing rules are necessary to show it. The situation is in some ways similar to
dual-context calculi, where meticulous expert attention is needed to show the admissibility of
substitution in each modal setting.

The final and most serious issue with the Fitch-style is the difficulty of accounting for multiple
distinct modalities. Each modality should give rise to a different lock, but the structural rules
governing their interactions are complex. It is well-understood how to model the $\Later$ modality
in a Fitch-style type theory, and \cite{gratzer:2019} developed an extensive account of the $\Box$
modality. However, it is an open problem whether the two may be combined. There is work to this
effect in a simple type theory~\cite{bahr:2019}, but even in this case there are restrictions on
$\Box$ and $\Later$ which prevent the recovery of the \MLTTLock~type theory of \cite{gratzer:2019}
as a subsystem.

These issues seem to converge to one cause: rules that `remove' elements from the context during
type-checking appear difficult to manage when combining modalities. As they operate on a syntactic
level, they also seem to prohibit the formulation of internal languages. Drawing on this intuition,
\MTT{} has adopted the simple introduction rules of Fitch-style calculi, but not the elimination
rules. The result is a less powerful type theory, with a weaker definitional equality, and no
definitional $\eta$-principle. In return, \MTT{} scales to any mode theory, including any number of
interacting modalities.

\subsection{Other work}
\label{sec:related:general}

The question of a multimodal framework for type theory has also been tackled by other recent
work~\cite{licata:2016,licata:2017}. This line of research is commonly referred to as the \emph{LSR}
framework, after the initials of the authors. LSR is designed to handle a wide variety of modal
situations in combination with a variety of different \emph{substructural} settings. There has been
ongoing work on extending this system to a full dependent type theory, but as of late 2020 this work
remains unpublished.

The impetus for the LSR framework is mainly derived from a long-standing wish to address the
interaction between dependent types and substructural logics. This is an axis of generalization
which is entirely outside the scope of \MTT{}. However, we may compare LSR to \MTT{} along the modal
axis.

The idea of parametrizing a type theory by a mode theory, as we have done with \MTT{}, originates in
a paper preceding the LSR framework~\cite{licata:2016}. In fact, the modal situations that can be
handled by \MTT{} are a strict subset of those which can be handled by pre-LSR/LSR, which also
includes a modality representing the \emph{left adjoint} as an operation on types (and not just
contexts). By contrast, \MTT{} has a simpler syntax which is amenable to current proof and
implementation techniques. This is reflected in our proof of canonicity, and our experimental
implementation efforts~\cite{menkar}. We therefore believe that \MTT{} is a natural halfway point
between current modal type theories (which are custom-fitted for each modal situation) and the full
generality of LSR: it is a simpler theory which accounts for many situations of interest.


\section{Conclusions}
\label{sec:conclusions}

We introduced and studied \MTT{}, a dependent type theory parametrized by a mode theory that
describes interacting modalities. We have demonstrated that \MTT{} may be used to reason about
several important modal settings, and proven basic metatheorems about its syntax, including
canonicity.

Several distinct directions of future work present themselves.

\paragraph{Towards an Implementation of \MTT{}}
A major point of future work is the development of an implementation of \MTT{}. Substantial
preliminary implementation efforts are already underway with \texttt{Menkar}~\cite{menkar}. In
addition to the engineering effort, a systematic account for an algorithmic syntax of \MTT{} as well
as proof of normalization is needed. We believe that the general ideas of \cite{gratzer:2019} are
applicable to this situation and there is ongoing work to apply them to \MTT{} through more modern
\emph{gluing} techniques~\cite{coquand:2018}. Eventually, this work should prove that
$\EqTm{M}{N}{A}$ and $\EqTy{A}{B}$ are decidable relative to a decision procedure for equality in
the underlying mode theory.

\paragraph{Left Adjoints}
As discussed in Section~\ref{sec:related:general}, \MTT{} trades a measure of generality for a degree of
simplicity, as compared to LSR. One might hope, however, that it would be possible to include a
connective for \emph{left adjoints}, as well as the current connective which models right adjoints
without losing all of this simplicity. It is not obvious that this can be done without significantly
changing \MTT{}: the introduction rule for modalities is exceptionally specific to a right adjoint.
This additionally flexibility would allow us to model several modalities which are currently out of
reach. For instance, when modeling a string of adjoints, we always fail to model the final left
adjoint. Concretely speaking, the inclusion of left adjoints would allow \MTT{} to model
computational effects~\cite{moggi:1991,levy:2012}, as we will be able to internally recover the
corresponding monad as the composite of the two parts of an adjoint pair.

\subsection*{Acknowledgements}
We are grateful for productive conversations with Carlo Angiuli, Dominique Devriese, Adrien
Guatto, Magnus Baunsgaard Kristensen, Daniel Licata, Rasmus Ejlers M\o{}gelberg, Matthieu Sozeau,
Jonathan Sterling, and Andrea Vezzosi.

Alex Kavvos was supported in part by a research grant (12386, Guarded Homotopy Type Theory) from
the VILLUM Foundation.
Andreas Nuyts was supported by a PhD Fellowship from the Research Foundation -
Flanders (FWO) at imec-DistriNet, KU Leuven.
This work was supported in part by a Villum Investigator grant (no.  25804),
Center for Basic Research in Program Verification (CPV), from the VILLUM Foundation.


\bibliographystyle{alpha}
\bibliography{refs.bib}

\end{document}